\documentclass[11pt]{article}

\usepackage{amsthm}
\usepackage{graphicx} 
\usepackage{array} 

\usepackage{amsmath, amssymb, amsfonts, verbatim}
\usepackage{hyphenat,epsfig,subcaption,multirow,mathtools}
\usepackage{nicefrac}
\usepackage{paralist}
\usepackage{thm-restate}
\usepackage{mleftright}

\usepackage[font=small, labelfont=bf]{caption}

\usepackage[usenames,dvipsnames]{xcolor}
\usepackage[ruled]{algorithm2e}

\DeclareFontFamily{U}{mathx}{\hyphenchar\font45}
\DeclareFontShape{U}{mathx}{m}{n}{
      <5> <6> <7> <8> <9> <10>
      <10.95> <12> <14.4> <17.28> <20.74> <24.88>
      mathx10
      }{}
\DeclareSymbolFont{mathx}{U}{mathx}{m}{n}
\DeclareMathSymbol{\bigtimes}{1}{mathx}{"91}

\usepackage{tcolorbox}
\tcbuselibrary{skins,breakable}
\tcbset{enhanced jigsaw}

\usepackage[normalem]{ulem}
\usepackage[compact]{titlesec}

\definecolor{DarkRed}{rgb}{0.5,0.1,0.1}
\definecolor{DarkBlue}{rgb}{0.1,0.1,0.5}

\usepackage{nameref}
\definecolor{ForestGreen}{rgb}{0.1333,0.5451,0.1333}
\definecolor{Red}{rgb}{0.9,0,0}
\usepackage[linktocpage=true,
	pagebackref=true,colorlinks,
	linkcolor=DarkRed,citecolor=ForestGreen,
	bookmarks,bookmarksopen,bookmarksnumbered]
	{hyperref}
\usepackage[noabbrev,nameinlink]{cleveref}
\crefname{property}{property}{Property}
\creflabelformat{property}{(#1)#2#3}
\crefname{equation}{eq}{Eq}
\creflabelformat{equation}{(#1)#2#3}

\usepackage{bm}
\usepackage{url}
\usepackage{xspace}
\usepackage[mathscr]{euscript}

\usepackage{tikz}
\usetikzlibrary{arrows}
\usetikzlibrary{arrows.meta}
\usetikzlibrary{shapes}
\usetikzlibrary{backgrounds}
\usetikzlibrary{positioning}
\usetikzlibrary{decorations.markings}
\usetikzlibrary{patterns}
\usetikzlibrary{calc}
\usetikzlibrary{fit}
\usetikzlibrary{decorations}

\usepackage[framemethod=tikz]{mdframed}

\usepackage[noend]{algpseudocode}
\makeatletter
\def\BState{\State\hskip-\ALG@thistlm}
\makeatother

\usepackage{cite}
\usepackage{enumitem}

\usepackage[margin=1in]{geometry}

\newtheorem{theorem}{Theorem}
\newtheorem{lemma}{Lemma}[section]
\newtheorem{proposition}[lemma]{Proposition}
\newtheorem{corollary}[lemma]{Corollary}
\newtheorem{claim}[lemma]{Claim}
\newtheorem{fact}[lemma]{Fact}

\newtheorem{definition}[lemma]{Definition}
\newtheorem{problem}{Problem}

\newtheorem*{claim*}{Claim}
\newtheorem*{assumption*}{Assumption}
\newtheorem*{proposition*}{Proposition}
\newtheorem*{lemma*}{Lemma}
\newtheorem*{problem5*}{Problem}

\crefname{lemma}{Lemma}{Lemmas}
\crefname{claim}{claim}{claims}

\newtheorem{mdresult}{Result}
\newenvironment{result}{\begin{mdframed}[backgroundcolor=lightgray!40,topline=false,rightline=false,leftline=false,bottomline=false,innertopmargin=2pt]\begin{mdresult}}{\end{mdresult}\end{mdframed}}

\newtheorem{remark}[lemma]{Remark}
\newtheorem{observation}[lemma]{Observation}

\newtheoremstyle{restate}{}{}{\itshape}{}{\bfseries}{~(restated).}{.5em}{\thmnote{#3}}
\theoremstyle{restate}

\theoremstyle{definition}
\newtheorem{mdinvariant}[lemma]{Definition}
\newenvironment{Definition}{\begin{mdframed}[roundcorner=10pt, hidealllines=false,innerleftmargin=5pt,backgroundcolor=white!10,innertopmargin=2pt]\begin{mdinvariant}}{\end{mdinvariant}\end{mdframed}}

\crefname{mdinvariant}{Definition}{Definitions}

\allowdisplaybreaks

\renewcommand{\geq}{\geqslant}
\renewcommand{\leq}{\leqslant}

\renewcommand{\qed}{\nobreak \ifvmode \relax \else
      \ifdim\lastskip<1.5em \hskip-\lastskip
      \hskip1.5em plus0em minus0.5em \fi \nobreak
      \vrule height0.75em width0.5em depth0.25em\fi}

\usepackage[ruled]{algorithm2e}

\newcommand{\tvd}[2]{\ensuremath{\norm{#1 - #2}_{\mathrm{tvd}}}}
\newcommand{\Ot}{\ensuremath{\widetilde{O}}}
\newcommand{\eps}{\ensuremath{\varepsilon}}
\newcommand{\Paren}[1]{\Big(#1\Big)}

\newcommand{\bracket}[1]{\left[#1\right]}
\newcommand{\paren}[1]{\ensuremath{\mleft(#1\mright)}\xspace}
\newcommand{\card}[1]{\left\vert{#1}\right\vert}

\newcommand{\norm}[1]{\ensuremath{\|#1\|}}

\newcommand{\prob}[1]{\Pr\paren{#1}}
\newcommand{\expect}[1]{\Exp\bracket{#1}}

\newcommand{\set}[1]{\ensuremath{\left\{ #1 \right\}}}

\newcommand{\polylog}{\mbox{\rm  polylog}}

\DeclareMathOperator*{\Exp}{\ensuremath{{\mathbb{E}}}}
\DeclareMathOperator*{\Prob}{\ensuremath{\textnormal{Pr}}}
\renewcommand{\Pr}{\Prob}

\newenvironment{tbox}{\begin{tcolorbox}[
		enlarge top by=5pt,
		enlarge bottom by=5pt,
		breakable,
		boxsep=0pt,
		left=4pt,
		right=4pt,
		top=10pt,
		arc=0pt,
		boxrule=1pt,toprule=1pt,
		colback=white
		]
	}
	{\end{tcolorbox}}

\newcommand{\event}{\ensuremath{\mathcal{E}}}

\newcommand{\II}{\ensuremath{\mathbb{I}}}

\newcommand{\mireal}[1][]{
	\ifx\relax#1\relax%
	\II(\mione \,; \mitwo)%
	\else%
	\II(\mione \,; \mitwo\mid #1)%
	\fi
}

\newcommand{\commentout}[1]{}

\newtheorem{mdalg}{Algorithm}
\newenvironment{Algorithm}{\begin{tbox}\begin{mdalg}}{\end{mdalg}\end{tbox}}

\crefname{mdalg}{algorithm}{Algorithms}

\newcommand{\istar}{\ensuremath{h^*}}

\newcommand{\ngc}{\NGC}
\newcommand{\distHNGC}{\mathcal{H}\xspace}

\newcommand{\outs}{\ensuremath{\textnormal{\textsf{out}}}}

\newcommand{\xormatch}[1]{\ensuremath{\textnormal{\textsf{XOR-Matching}}\paren{#1}}\xspace}
\newcommand{\permmatch}[1]{\ensuremath{\textnormal{\textsf{Perm-Matching}}\paren{#1}}\xspace}

\newcommand{\block}[1]{\ensuremath{\textnormal{\textsf{Block}}(#1)}\xspace}
\newcommand{\multiblock}[1]{\ensuremath{\textnormal{\textsf{Multi-Block}}(#1)}\xspace}

\newcommand{\distNGC}{\mu_{\textnormal{\texttt{ngc}}}\xspace}
\newcommand{\distDHX}{\mu_{\textnormal{\texttt{dhx}}}\xspace}
\newcommand{\NGC}{\ensuremath{\textnormal{\textbf{NGC}}}\xspace}
\newcommand{\DHXOR}{\ensuremath{\textnormal{\textbf{DHX}}}\xspace}
\newcommand{\LG}{\ensuremath{\mathcal{L}}}

\newcommand{\oDD}[2]{\ensuremath{\overrightarrow{\textnormal{D}}_{\!#1,#2}}\xspace}
\newcommand{\oRR}[1]{\ensuremath{\overrightarrow{\textnormal{R}}_{#1}}\xspace}

\newcommand{\ww}{w}
\newcommand{\dd}{d}

\newcommand{\conc}{\mathbin\Vert}

\newcommand{\cU}{\ensuremath{\mathcal{X}}}

\newcommand{\bias}[2]{\ensuremath{\textnormal{\textsf{bias}}_{#1}\paren{#2}}}
\newcommand{\Tau}{\Phi}
\newcommand{\ttau}{\phi}

\newcommand{\Left}{\ensuremath{f_L}}
\newcommand{\Mid}{\ensuremath{f_M}}
\newcommand{\Right}{\ensuremath{f_R}}
\newcommand{\Aedge}{\ensuremath{\alpha}}
\newcommand{\Bedge}{\ensuremath{\beta}}
\newcommand{\indclean}{\ensuremath{R}}

\newcommand{\msg}{\ensuremath{\pi}}
\newcommand{\goodmsg}{\ensuremath{\textnormal{\textsf{GoodMsg}}}\xspace}

\newcommand{\clninv}[1]{\ensuremath{\textnormal{\textsf{clean}}^{-1}(#1)}}

\newcommand{\sigmaandx}[2]{\ensuremath{\mathcal{Q}}(#1, #2)}
\newcommand{\sigmaxunclean}[1]{\ensuremath{\mathcal{Q}^{\textnormal{\textsf{uc}}}(#1)}}
\newcommand{\sigmaxclean}[1]{\ensuremath{\mathcal{Q}^{\textnormal{\textsf{c}}}(#1)}}

\newcommand{\tomsg}[1]{\ensuremath{I(#1)}}

\newcommand{\msglgth}{\ensuremath{c}}

\newcommand{\actt}{\ensuremath{t_a}}
\newcommand{\clw}{\ensuremath{\ww_c}}

\newcommand{\NGCst}{\ensuremath{\textnormal{\textbf{SGC}}}}
\newcommand{\DHXORst}{\ensuremath{\textnormal{\textbf{SHX}}}}
\newcommand{\cons}{\ensuremath{c}}

\newcommand{\invsigma}{\ensuremath{\sigma_{-1}}}

\renewcommand{\eps}{\epsilon}

\title{(Noisy) Gap Cycle Counting Strikes Back: 
Random Order Streaming Lower Bounds for Connected Components and Beyond\footnote{An extended abstract of this paper appears at ACM Symposium on Theory of Computing (STOC) 2023. \medskip}}
\author{Sepehr Assadi\footnote{Department of Computer Science, Rutgers University. Research supported in part by an NSF CAREER Grant CCF-2047061, a Sloan Research Fellowship, a Google Research gift, and a Fulcrum award from Rutgers Research Council. Emails: \texttt{sepehr@assadi.info} and \texttt{sun.j@rutgers.edu}.} \and Janani Sundaresan\footnotemark[2]}

\date{}

\begin{document}
\maketitle

\pagenumbering{roman}

\begin{abstract}

We continue the study of the communication complexity of gap cycle counting problems. These problems have been introduced  by Verbin and Yu [SODA 2011] and have found numerous applications
in proving {streaming} lower bounds. In the {noisy} gap cycle counting problem (NGC), there is a small integer $k \geq 1$ and an $n$-vertex graph consisted of  vertex-disjoint union of either $k$-cycles or $2k$-cycles, plus $O(n/k)$ disjoint paths of length $k-1$ in both cases (``noise''). The edges of this graph are partitioned between Alice and Bob whose goal is to decide which case the graph belongs to with minimal communication  from Alice to Bob.  

\medskip

We study the \textbf{robust communication complexity}---\`a la Chakrabarti, Cormode, and McGregor [STOC 2008]---of NGC, namely, 
when edges are partitioned {randomly} between the  players. This is in contrast to all prior work on gap cycle counting problems in adversarial partitions. 
While NGC can be solved trivially with zero communication when $k < \log{n}$, we prove that when $k$ is a constant factor larger than $\log{n}$, the robust ({one-way}) communication complexity of NGC is $\Omega(n)$ bits. 

\medskip

As a corollary of this result, we can prove several new graph streaming lower bounds for \textbf{random order streams}. In particular, we show that any streaming algorithm that for every $\eps > 0$ estimates the number of connected components 
of a graph presented in a random order stream to within an $\eps \cdot n$ additive factor requires $2^{\Omega(1/\eps)}$ space, settling a conjecture of Peng and Sohler [SODA 2018]. 
We further discuss new implications of our lower bounds to other problems such as estimating size of maximum matchings and independent sets on planar graphs, random walks, as well as to stochastic streams. 

\end{abstract}

\clearpage

\setcounter{tocdepth}{3}
\tableofcontents

\clearpage

\pagenumbering{arabic}
\setcounter{page}{1}

\section{Introduction}

The streaming model of computation, introduced by~\cite{AlonMS96}, is studied extensively in the literature as a way of capturing some challenges of processing massive datasets, in particular, space efficiency. 
In this model, the input is presented to the algorithm in a stream, and the algorithm needs to process this stream on the fly with minimal space, ideally polylogarithmic in the input size. 
Numerous such efficient algorithms are known for various statistical estimation problems like frequency moments, heavy hitters, quantiles, and alike (see, e.g.~\cite{Muthukrishnan05} for a summary of early work on this model). 
But, estimating graph parameters, say, the number of connected components or size of maximum matchings, from a stream of their edges turned out to be quite more challenging, and   
for many graph problems of interest, one can rule out existence of such space efficient algorithms (see, e.g.,~\cite{McGregor14}).   

This provable impossibility of efficient  streaming algorithms for most graph problems has led researchers to consider natural relaxations of the model, e.g., by allowing more memory or a  few more passes\footnote{Often in the semi-streaming 
model~\cite{FeigenbaumKMSZ05} wherein the memory is proportional to the number of vertices $n$. In this work, we do \emph{not} consider semi-streaming algorithms and  solely focus on algorithms with $o(n)$ space.}.
A particularly successful relaxation has been to go beyond the doubly worst case analysis of streaming algorithms and instead consider \emph{random order} streams, wherein the input graph is still adversarial, but its edges
are presented in a uniformly random order to the streaming algorithm~\cite{KapralovKS14,KapralovKS15,MonemizadehMPS17,PengS18,KapralovMNT20,CzumajFPS20,KallaugherKP22,ChipKKP22,SaxenaSSV23}. 
Our goal in this paper is to study the inherent limitations of random order graph streaming algorithms. 
 
One of the most successful techniques in designing random order graph streaming algorithms is through implementing \emph{local exploration} (query) algorithms. For instance, consider
the problem of estimating the number of connected components in an $n$-vertex graph to within an $\eps \cdot n$ factor. A seminal work of~\cite{ChazelleRT01} gave a local exploration query algorithm for this problem
by performing truncated BFS of size $O(1/\eps)$ from few random vertices of the graph. In general, computing BFS trees in the (adversarial arrival) streaming model is  quite challenging 
(see, e.g.,~\cite{FeigenbaumKMSZ08,AssadiKSY20,AssadiN21}). Yet,~\cite{PengS18} gave an elegant implementation of this algorithm in random order streams by crucially exploiting the fact that in such streams, edges of a BFS tree of size $O(1/\eps)$ may 
come in the ``right'' order of exploration
in the stream with probability  $(1/\eps)^{-\Theta(1/\eps)}$. Thus, by sampling roughly $(1/\eps)^{O(1/\eps)}$ vertices and performing a BFS from them by extending each one directly given the arriving edges,~\cite{PengS18}, and later~\cite{ChipKKP22}, showed 
that one can indeed estimate the number of connected components in random order streams in small space\footnote{The  algorithm of~\cite{PengS18} required $2^{O(1/\eps^3)} \cdot (\log{n})$ space but it was  improved to $(1/\eps)^{O(1/\eps)} \cdot (\log{n})$ 
 by~\cite{ChipKKP22}.}. 

While the random order streaming algorithms of~\cite{PengS18,ChipKKP22} for connected components are quite efficient for constant values of $\eps > 0$, their exponential dependence on $\eps$ in the space can become 
 prohibitive even for moderately small values of $\eps$. Yet, these exponential-in-$\eps$ (or even much larger) dependences are quite common among  random order graph streaming algorithms such as $\eps$-property testers in~\cite{MonemizadehMPS17,CzumajFPS20}, $\eps \cdot n$ additive approximation of  matching size in bounded degree graphs~\cite{MonemizadehMPS17}, $(1+\eps)$-approximation of minimum spanning tree weight~\cite{PengS18}, 
$(1+\eps)$-approximation of maximum independent set in planar graphs~\cite{PengS18}, and random walks~\cite{KallaugherKP22}. 

But, are such exponential dependences also necessary? This question was originally posed by~\cite{PengS18} who stated that: 
\begin{quote}
	\emph{``It seems to be plausible to conjecture that approximating the number of connected
components requires space exponential in $1/\eps$.
 It would be nice to have lower bounds that confirm this conjecture.''}
\end{quote}

Very recently~\cite{ChipKKP22} provided a strong evidence in favor of this conjecture by relaxing the question. They showed that: $(i)$ if we slightly relax the random order streams by allowing some correlation in the ordering of the edges\footnote{\cite{ChipKKP22} defined the \emph{hidden-batch} random order model wherein the edges of the graph are first adversarially batched together in groups of $O(1)$ size, and then the groups are randomly permuted and presented to the algorithm.}, 
and $(ii)$ if we further relax our goal in \emph{finding} a  connected component of size $O(1/\eps)$ (as opposed to estimating their numbers), then they can indeed prove an exponential (even factorial) lower bound on $\eps$ in the space of 
such algorithms\footnote{Among these two relaxations, the latter seems to us a stronger condition to impose on algorithms as proving \emph{search} lower bounds can be quite different than \emph{decision} lower bounds in this context; see, e.g.,
the classical communication/streaming lower bounds on the Hidden Matching problem (search)~\cite{BarJK04} and Boolean Hidden Matching (decision)~\cite{GavinskyKKRW07}, or recent advances on search lower bounds for better-than-$2$-approximation of multi-pass streaming MAX-CUT~\cite{ChenKPSSY23} in absence of any decision lower bounds.}. 

Nevertheless, the original conjecture for (truly) random order streams and, more importantly, for  estimation problem remains unanswered. This state-of-affairs is the starting point of our work. 

We develop a  new tool for proving random order streaming lower bounds that leads to $2^{\Omega(1/\eps)}$ space lower bounds for a wide range of graph estimation problems, including  the 
connected component problem, settling the conjecture of~\cite{PengS18}. Similar to most other graph streaming lower bounds, our proofs are through communication complexity. 
In the following, we first define the main communication problem we work with and the lower bound we establish for it. We then present  rather direct implications of this communication lower bound
to  random order streams. 

\subsection{Local Exploration in Graph Streams and (Noisy) Gap Cycle Counting}\label{sec:gcc}

At its core, the question we study in this paper boils down to understanding the power of random order streaming algorithms in performing local exploration. In general, 
such questions have a rich history in the graph streaming literature. Already more than a decade ago,~\cite{VerbinY11} introduced gap cycle counting problems in communication complexity as a way to prove lower bounds 
for streaming problems that rely on these explorations. In this problem, Alice and Bob are given edges of a graph $G$ and an integer $k \geq 1$ and want to test if $G$ is  
a disjoint union of $k$-cycles or $2k$-cycles. By building on the Fourier-analytic approach of~\cite{GavinskyKKRW07},~\cite{VerbinY11} proved that any one-way communication protocol wherein Alice sends a single message to Bob requires $n^{1-O(1/k)}$ communication to solve this problem. As communication complexity 
implies streaming lower bounds, this leads to an $n^{1-O(1/k)}$-space lower bounds for single-pass streaming algorithms for gap cycle counting. 

Gap cycle counting is a highly versatile problem that captures the difficulty in solving various graph streaming problems. For instance, since the number of connected components in the two cases of this problem differs by an $n/2k$ additive factor,
one can use the lower bound of~\cite{VerbinY11} to also obtain that estimating the number of connected components to within an $\eps \cdot n$ factor in single-pass (adversarial order) streams requires $n^{1-O(\eps)}$ space; see~\cite{HuangP16} for this lower 
bound and its optimality. Gap cycle counting problems have been extensively used for proving space lower bounds for various streaming
problems, e.g., in~\cite{KapralovKS15,KoganK15,BuryS15,EsfandiariHLMO15,HuangP16,LiW16,GuruswamiVV17,BravermanCKLWY18,KapralovMTWZ22} (see~\cite[Appendix B]{AssadiKSY20} for a survey of these results).  
As a result, this problem has found its way among the few canonical communication problems---alongside, say, Index, Disjointness, and Gap Hamming Distance---for proving streaming lower bounds. 

Given the canonical role of gap cycle counting in streaming lower bounds, recent years have witnessed several attempts in extending the lower bounds for this problem (or its natural extensions) to other streaming settings as well. 
In particular,~\cite{AssadiKSY20} extended the lower bounds for this problem to multi-round communication (multi-pass streaming) algorithms, proving that $\Omega(\log{k})$ rounds/passes or $n^{\Omega(1)}$ communication/space is needed. Subsequently~\cite{AssadiN21} extended these lower bounds to an optimal $\Omega(k)$-pass for $n^{o(1)}$-space streaming algorithms for a ``noisy'' variant of this problem wherein, the input graph additionally contains $\Theta(n/k)$ paths 
of length $k$ in either cases\footnote{The lower bound of~\cite{AssadiN21} is not based on two-player communication complexity and is proven  directly for streaming algorithms but can also be stated as a 
(NIH) multi-party communication lower bound.} -- in other words, depth-$k$ exploration requires $\Omega(k)$ passes in adversarially ordered streams. 
Combining this with prior reductions, the results in~\cite{AssadiKSY20,AssadiN21} provided the first, and in many cases the only, known lower bounds for multi-pass streaming algorithms for graph estimation problems. 

\paragraph{(Noisy) gap cycle counting in random order streams.} All prior lower bounds for gap cycle counting problems were focused on adversarial ordered streams\footnote{The only exception we are aware of is the $\Omega(\sqrt{n})$-space lower
bound of~\cite{KapralovKS15} for better-than-$2$-approximation of MAX-CUT in random order stream and its very recent extension in~\cite{SaxenaSSV23} to other streaming CSPs. These results however do not imply
random order lower bounds for gap cycle counting and instead borrow ideas from this problem to establish their own lower bounds using several other technical ingredients.}. 
However, given the previous successes  in using this problem for addressing the role of local exploration in streaming algorithms and the myriad of reductions known already from this problem, it is quite natural to wonder if we can 
also prove an analogue of these lower bounds for random order streaming algorithms. This is precisely the approach we take in this paper. 

To be able to lift our communication lower bounds to random order streaming algorithms, we work with the notion of \emph{robust communication complexity} introduced by~\cite{ChakrabartiCM08}. In this model, a graph $G$ is adversarially chosen
but then each of its edges is sent to Alice and Bob uniformly at random. This random partitioning of the input then allows one to extend communication lower bounds in this model to streaming algorithms on random order streams (see~\Cref{prop:stream-cc}). 
This motivates the definition of our main communication problem (similar-in-spirit to~\cite{AssadiN21}).

\begin{problem}[\textbf{Noisy Gap Cycle Counting (NGC)}] \label{def:ngc}
	For integers $n,k \geq 1$, in $\NGC_{n,k}$, we have a graph $G = (V,E)$ on $n $ vertices such that $G$ either contains: $(i)$ $(n/4k)$ vertex-disjoint cycles of length $2k$, or $(ii)$ $(n/2k)$ vertex-disjoint cycles of length $k$.
	In addition, in both cases $G$ contains $(n/2k)$ vertex-disjoint paths of length $k-1$.\footnote{This choice of number of vertex-disjoint paths is  for simplicity of exposition and can be replaced with any $o(n/k)$ or even smaller; see~\Cref{rem:reduce}.}
	
	Each edge $e \in E$ is sent uniformly at random to one of the sets $E_A$ and $E_B$, given to Alice and Bob, respectively. Alice can send one message to Bob, and Bob outputs
	which of the above two cases the graph belongs to. 
\end{problem}

It is easy to see that  when $k$ is sufficiently smaller than $\log{n}$, $\NGC_{n,k}$, admits a trivial zero communication protocol. In this case, one of the $(n/4k)$ vertex-disjoint cycles of length $2k$ 
of case $(i)$ appear entirely in $E_B$ given to Bob with constant probability, which allows him to distinguish this case from the other one. 

Our main technical result shows that, in contrast, the moment $k$ becomes sufficiently larger than $\log{n}$, this problem requires 
a large communication. 

\begin{result}\label{res:ngc}
	There exists an absolute constant $b > 0$ such that the following is true. Any one-way protocol for $\NGC_{n,k}$ for $k \geq b \cdot \log{n}$ requires $\Omega(n)$ communication from Alice
	to Bob to solve the problem with probability of success at least $2/3$. 
\end{result}

\Cref{res:ngc} establishes that while ``short-range'' exploration, namely, depth-$o(\log{n})$ exploration, can be trivial in the robust communication complexity model, exploring slightly longer ranges requires exponential-in-range communication. 
Prior to our work, no non-trivial robust communication lower bounds were known for this or any other gap cycle counting problems. We hope this result paves the path toward a more systematic study of lower bounds 
in random order streams. 

\subsection{Random Order Streaming Lower Bounds from NGC}

We can now use~\Cref{res:ngc} alongside existing reductions in~\cite{VerbinY11,AssadiKSY20,AssadiN21}, as well as some new ones, to establish several new random order graph streaming lower bounds (see~\Cref{sec:stream} for the definition
and background on each problem).

\begin{result}\label{res:stream}
	None of the following problems admit an algorithm in random order streams that for every $\eps > 0$ uses $2^{o(1/\eps)}$ space and computes the  answer with probability at least $2/3$ (the references list the known random order
	streaming algorithms for these problems): 
	\begin{enumerate}[label=$(\roman*)$]
		\item estimating connected components to within an $\eps \cdot n$ additive factor (cf.~\cite{PengS18,ChipKKP22}); 
		\item $(1+\eps)$-approximation of weight of the minimum spanning tree (cf.~\cite{PengS18}); 
		\item $(1+\eps)$-approximation of  matching size in bounded degree graphs (cf.~\cite{MonemizadehMPS17}); 
		\item $(1+\eps)$-approximation of maximum independent size in planar graphs (cf.~\cite{PengS18}); 
		\item random walk generation; see~\Cref{sec:rw} for the precise definition (cf.~\cite{KallaugherKP22}). 
	\end{enumerate}
\end{result}

Prior to our work, no non-trivial random order streaming lower bounds were known for any of these problems (but as mentioned earlier, \cite{ChipKKP22} proved lower bounds for the search version of some of these problems and in almost random order streams). 
Our~\Cref{res:stream} establishes the necessity of an exponential-in-$\eps$ space-dependence in general, for all of these problems. In particular, part $(i)$ of this result now fully settles the conjecture of~\cite{PengS18} mentioned earlier. 

An important remark on~\Cref{res:stream} is in order. Direct reductions from NGC when proving~\Cref{res:stream} typically require setting $\eps \approx 1/\log{n}$ so that $2^{o(1/\eps)} = o(n)$ to apply~\Cref{res:ngc}.  This is reminiscent of other communication complexity lower bounds, 
most prominently for  $(1+\eps)$-approximation of the number of distinct elements from the Gap Hamming Distance problem~\cite{IndykW03,ChakrabartiR11}. Similar to those cases, for some problems (e.g., parts $(iii)$ and $(iv)$ of~\Cref{res:stream}), this can be easily extended to other choices of $\eps > 0$
through a simple padding argument, but this approach does not work for all problems, in particular for part $(i)$. Thus, for some problems, our lower bound only holds for certain values of $\eps$ and not all choices (see~\Cref{rem:small-eps} for a detailed discussion of this topic)\footnote{This in particular implies that while our lower bounds using communication complexity bypasses the relaxations imposed by~\cite{ChipKKP22} and target the original estimation problem on random order streams, unlike~\cite{ChipKKP22}, our results cannot be applied to the full range of parameter $\eps > 0$. It is however also worth mentioning that the $(1/\eps)^{\Omega(1/\eps)}$ lower bounds in~\cite{ChipKKP22} for constant $\eps > 0$ (or even $\eps \approx 1/\log\log{n}$) are subsumed by the $\Omega(\log{n})$ bits lower bound needed simply to read an edge from the stream. Thus, even in~\cite{ChipKKP22}, one needs slightly sub-constant values of $\eps$ to obtain non-trivial lower bounds, although those bounds are exponentially better than ours.}. 

We note that our lower bounds in~\Cref{res:stream} are somewhat stronger than what is stated and in most cases rule out even $2^{o(1/\eps)} \cdot n^{o(1)}$ space algorithms. This is necessary to argue near-optimality of these results 
as all existing algorithms also have some mild dependence on $n$, typically $O(\log{n})$, in addition to $\eps$, simply for storing the edges or some counters. 

Finally, we also extend our lower bounds to \emph{stochastic} graph streams defined in~\cite{CrouchMVW16}, wherein the stream consists of independently chosen random samples from the edges of the graph,
and the goal is to solve the problem with minimal space and minimal number of samples.  To obtain this, we prove an analogue of~\Cref{res:ngc} but for a communication model that corresponds to stochastic streams, and 
use that to obtain similar lower bounds as in~\Cref{res:stream} for stochastic streams as well; see~\Cref{sec:stochastic} for more information on these results. 

\subsection{Our Techniques} 

The bulk of the technical effort of our paper is in proving~\Cref{res:ngc} and the lower bounds in~\Cref{res:stream} follow easily from this result. So, in this section, we only 
focus on the former result, namely, a robust communication lower bound for $\NGC$, and postpone the details of the latter one to~\Cref{sec:stream}. 

The first step of our lower bound follows a by-now familiar ``decorrelation” step to break the strong promise in the input graphs (that the cycles are
either all short or all long). This approach dates all the way back  to~\cite{GavinskyKKRW07,VerbinY11} and was done more explicitly in~\cite{AssadiKSY20,AssadiN21,KapralovMTWZ22,AssadiCLMW22}. At a high level, the intuition is as follows. 
A natural strategy for a  protocol to solve $\NGC$ is to attempt to solve the problem for ``many'' cycles (i.e., determine whether they are $k$-cycles or $2k$-cycles), but each one with only a ``minor'' chance of success -- given the strong correlation of the cycles in the input, the protocol is likely to solve the problem correctly for one of the cycles and use that to infer which case it belongs to. These decorrelation steps use a hybrid argument to show that the converse is also true: any protocol 
for $\NGC$ leads to a protocol for a ``hybrid problem'' that roughly speaking corresponds to the following: given a fixed vertex $v \in V$, is this vertex part of a $k$-cycle or a $2k$-cycle without the strong promise on remaining cycles (as in, now, combination of  length-$k$ and length-$2k$ cycles are also possible). Nevertheless, the challenge is that this new protocol only has a ``tiny'' advantage over random guessing the answer, i.e., 
it solves this new problem with probability only $1/2+O(k/n)$. Thus, for the next step of the argument, we should prove a ``low probability of success'' lower bound for our new hybrid problem. 

The next step, and the main part of the argument, is to prove this low probability of success lower bound. The main recipe for this part is to apply some ``hardness amplification'' technique, typically an XOR-lemma~\cite{GavinskyKKRW07,VerbinY11,AssadiN21} 
or a direct product argument~\cite{AssadiKSY20}. This allows one to obtain a low probability of success lower bound from a more typical lower bound by amplifying the hardness often through certain gadgets embedded in the input. 
This is where we entirely deviate from these prior works, as unlike all of them, our goal is to prove a robust communication complexity lower bound, 
i.e., when the inputs are partitioned randomly between Alice and Bob. This random partitioning of the input has this undesired side effect that whenever we fix a certain gadget in our input, it is quite likely that the edges of this gadget are partitioned 
the ``wrong way'' across the players, entirely destroying the role of this gadget. 

To address this challenge, we design a hard input distribution for $\NGC$ which is supported on highly structured input graphs. This then allows us to argue that after this decorrelation step, the new hybrid problem we have to deal with is also highly structured. 
Quite informally speaking, this means that the answer to our hybrid problem now is a function of a series of ``localized'' gadgets that we can keep track of. In particular, we can show that even under random partitioning of the input, with high probability, many of these gadgets ``survive'', namely, retain their hardness amplifying property. Moreover, the ``failed'' gadgets now are not able to leak any information about the surviving ones or negate their effects. This then, in a highly non-black-box way, 
bring us to the more familiar setting of proving low probability of success lower bounds using hardness amplification ideas, and more specifically, an XOR-lemma-type argument (although we emphasize that we neither prove nor use 
any form of a generic XOR-lemma, say, as in~\cite{AssadiN21} in this work, as the random partitioning of the input does not allow using these results in a black-box way). 

The final step will be then to establish the lower bound using these surviving gadgets which is based on a white-box application of the Fourier analytic approaches in~\cite{GavinskyKKRW07,VerbinY11} (although interestingly, a technical corollary of the main Fourier analytic tool used in these works, namely, the KKL inequality~\cite{KahnKL88} suffices for our purpose and thus 
we only work with this corollary, without us performing any Fourier analysis in this paper; see~\Cref{par:biasxor} for more details). 






\clearpage

\clearpage

\section{Preliminaries}

\paragraph{Notation.} For any integer $m \geq 1$, let $[m] := \set{1,2,\ldots,m}$ and $\mathcal{S}_m$ be the set of all permutations on $[m]$. We use $\invsigma$ to define the inverse of permutation $\sigma \in \mathcal{S}_{m}$. 
In any graph $G = (V, E)$, we use $V$ to denote the vertices and $E$ for the edges. Throughout, we set $n := \card{V}$ to denote the number of vertices. 
For any two vertices $s,t \in V$, we write $s \rightarrow t$ to denote that there is a path of length one, namely, an edge, from $s$ to $t$, and $s \rightsquigarrow t$ to denote there is a path of arbitrary length (the direction of the path will be clear from the context). 

The \emph{total variation distance} of  any two distributions $\mu$ and $\nu$ over the same support $\Omega$ is: 
\[
\tvd{\mu}{\nu} = \max_{\Omega' \subseteq \Omega} (\mu(\Omega') - \nu(\Omega')) = \frac12 \sum_{s \in \Omega} \card{\mu(s)-\nu(s)}. 	
\]

\begin{fact}\label{prop:tvdsample}
	Given a single sample $s$ chosen uniformly at random from one of the distributions $\mu$ and $\nu$, the best probability of successfully deciding whether $s$ came from $\mu$ or $\nu$ is $\frac{1}{2} + \frac{1}{2} \cdot \tvd{\mu}{\nu}$ (obtained 
	via the maximum likelihood estimator). 
\end{fact}


\begin{proposition}[cf.~\cite{DubhashiP09}]\label{prop:chernoff}
	For any $m$ independent $0/1$-random variables $X_1, X_2, \ldots, X_m$ with $X = \sum_{i=1}^{m} X_i$, and $ 0 \leq \eps \leq 1$,
	\[
	\prob{X \leq (1-\eps) \cdot \expect{X}} \leq \exp\paren{\frac{-\eps^2 \cdot \expect{X}}{3}}.
	\]
\end{proposition}

\subsection{Robust Communication Complexity}\label{sec:cc}

Our lower bounds for random order streaming algorithms are proven through (one-way) robust communication complexity introduced by~\cite{ChakrabartiCM08}. We provide some brief context here and refer
the interested reader to~\cite{ChakrabartiCM08} for more information. 

Let $\mathcal{G}_n$ be any  subset of $n$-vertex graphs  and $f: \mathcal{G}_n \rightarrow \set{0,1}$ be a  function. 
In this model, we pick a graph $G$ from $\mathcal{G}_n$ adversarially, and then each edge of $G$ is assigned uniformly at random to either $G_A$ or $G_B$, given to Alice and Bob, respectively. 
The players have access to a shared public tape of random bits $R$ and follow a protocol $\Pi$. In the protocol $\Pi$, Alice sends a single message $\pi$ based on $(G_A,R)$ to Bob, and Bob needs to output $f(G_A \cup G_B)$ based on
$\pi$, $G_B$, and $R$. We use $\Pi(G_A, G_B, R)$ to denote the output of the protocol on inputs $G_A,G_B$ and random bits $R$.  Protocol $\Pi$ is said to be a $\delta$-error protocol for $0 \leq \delta \leq 1$ for function $f$
if for any graph $G$ in $\mathcal{G}_n$, we have, 
\[
\Pr\paren{\Pi(G_A,G_B,R) = f(G)} \geq 1-\delta,
\]
where the probability is over the random partitioning into $(G_A,G_B)$ and random bits of $R$. 
The communication cost of a protocol $\Pi$ for $f$ is the worst-case length of its message. 

\begin{Definition}\label{def:ccf}
	The \textbf{randomized one-way communication complexity} of a function $f$ with error at most $\delta$, denoted by $\oRR{\delta}(f)$ is the minimum possible communication cost of a protocol for $f$ with error at most $\delta$ for every 
	graph $G$ chosen from $\mathcal{G}_n$. 
	
	The \textbf{distributional one-way communication complexity} of a function $f$ with error at most $\delta$ over a distribution $\mu$ on $\mathcal{G}_n$, denoted by $\oDD{\delta}{\mu}(f)$ is the minimum possible communication cost of a deterministic protocol for $f$ with error at most $\delta$ over graphs $G$ chosen from $\mu$. 
\end{Definition}

\begin{proposition}[Yao's minimax principle~\cite{Yao77}]\label{prop:minmax}
	For any function $f$ and error $0 \leq \delta \leq 1$, 
	\[
	\oRR{\delta}(f) = \max_{\mu} \oDD{\delta}{\mu}(f).
	\]
\end{proposition}

The following standard result relates robust communication complexity to streaming algorithms. 
\begin{proposition}[cf.~\cite{ChakrabartiCM08}]\label{prop:stream-cc}
	Any single-pass streaming algorithm that given edges of a graph $G$ from $\mathcal{G}_n$ in a random-order stream computes $f(G)$ w.p. at least $1-\delta$ requires $\Omega(\oRR{\delta}(f))$ space. 
\end{proposition}
\begin{proof}
	Consider the following protocol $\Pi$: Alice randomly permutes the edges of $G_A$, runs the streaming algorithm over them, and sends the memory content as the message $\pi$ to Bob. 
	Bob takes a random permutation of $G_B$, continues running the algorithm, and outputs the answer. The claim follows since the random partitioning of the input plus the random permutation of the edges leads 
	to a random order stream, and size of $\pi$ is bounded by the communication cost of $\Pi$.  
\end{proof}

\subsection{Biases and XORs}\label{par:biasxor}

We use a basic application of the \emph{KKL inequality}~\cite{KahnKL88} from Fourier analysis in our results. This approach has been at the core of various Fourier analytic lower bounds in streaming (see, e.g.~\cite{GavinskyKKRW07,VerbinY11,KapralovKS15}). In this work, we directly apply a technical corollary of this result stated in~\cite{Wol08} without relying on the Fourier analytic parts explicitly. As such, in the following, we only provide the necessary background for 
this corollary.

Let $X$ be any random variable with support $\set{0,1}$. The \emph{bias} of $X$ is defined as: 
\[
\bias{}{X} := \card{\prob{X=1}-\prob{X=0}}. 
\]
Notice that this way, for any $b \in \set{0,1}$, $\Pr\paren{X=b} \leq ({1+\bias{}{X}})/2$. 

Let $A \subseteq \set{0,1}^m$ and $Y$ be a random variable denoting a string chosen uniformly at random from $A$. For any set $S \subseteq [m]$, 
we slightly abuse the notation and write: 
\[
\bias{A}{S} =  \bias{}{\bigoplus_{i \in S}Y_i},
\]
where $\bigoplus$ denotes the XOR function. In words, $\bias{A}{S}$ is the bias of XOR of indices in $S$ of a random string from $A$. 

The following corollary of {KKL inequality}~\cite{KahnKL88} captures the property we need from biases. 
\begin{proposition}[cf.~{\cite[Section 4.2]{Wol08}}]\label{prop:kklparity}
	For any integers $m,k \geq 1$, let $A \subseteq \set{0,1}^m$ be arbitrary and $S$ be chosen uniformly at random from $k$-subsets of $[m]$. Then, 
	\begin{align*}
		\Exp_S\bracket{\bias{A}{S}^2} = O\paren{\frac1m \cdot \log\paren{\frac{2^m}{\card{A}}}}^k.
	\end{align*}
\end{proposition}

An interpretation of~\Cref{prop:kklparity} is in order. Suppose $\card{A} = 2^{m-t}$ for some (small) $t \geq 1$, thus the RHS above is $O(\frac{t}{m})^k$. 
When $k=1$, the bound above is rather straightforward as we expect the bias-squared of a single coordinate in $Y$ chosen randomly from $A$ to be roughly $O(t/m)$. This is because the entropy
of a random coordinate of $Y$ is at least $1-t/m$, making this bit quite close to uniform. But for $k=2$, already~\Cref{prop:kklparity} becomes quite interesting. It is an easy exercise
to show that XOR of \emph{independent} bits has a dampened bias compared to each bit, i.e., $\bias{}{P \oplus Q} = \bias{}{P} \cdot \bias{}{Q}$ for two independent random variables $P,Q \in \set{0,1}$. 
\Cref{prop:kklparity} extends this property to random indices $(i,j)$ of $Y$ chosen uniformly at random from $A$, even though coordinates $Y_i$ and $Y_j$ are \emph{not} independent anymore! 
Thus,~\Cref{prop:kklparity} is a vast generalization of the bias-dampening property of XOR  to any \emph{arbitrary} set $A$ (of sufficiently large size).

\clearpage

\section{A Robust Lower Bound for Noisy Gap Cycle Counting}

The following theorem is our main technical contribution which formalizes~\Cref{res:ngc}. 

\begin{theorem}\label{thm:main-ngc}
	For sufficiently large $n \geq 1$ and $k \geq 1600 \cdot \ln{n}$,  
	\[
	\oRR{1/3}(\NGC_{n,k}) = \Omega(n). 
	\]
\end{theorem}

To prove \Cref{thm:main-ngc}, we want to give a hard distribution $\distNGC$ such that the deterministic communication complexity of $\NGC_{n,k}$ over distribution $\distNGC$ is $\Omega(n)$. First, we describe the primitive constructs needed to define our hard distribution $\distNGC$. Once we have these constructs, we define $\distNGC$ and give a reduction from $\NGC_{n,k}$ over $\distNGC$ to a new intermediate problem we introduce here, called the distributional hidden XOR ($\DHXOR$) problem. We conclude by showing how a lower bound for the $\DHXOR$ problem gives a lower bound for $\NGC_{n,k}$. The lower bound for $\DHXOR$ itself is proved in the next section.   

\subsection{Building Blocks}

In this section, we define the basic constructs needed to describe our hard distribution for $\NGC_{n,k}$. We start by defining a group-layered graph (see~\Cref{fig:gplayergraph} for an illustration). 

\begin{Definition}[\textbf{Grouped-Layered Graph}]\label{def:layered-graph}
	For integers $\ww,\dd \geq 1$, we define a \textbf{group-layered graph} as any graph $G=(V,E)$ satisfying the following properties: 
	\begin{enumerate}[label=$(\roman*)$]
		\item Vertices of $G$ are partitioned into $\dd$ equal-size \textbf{layers} $V^1,\ldots,V^\dd$, each of size $2\ww$. We  identify each of these vertex-sets 
		by $[2\ww]$. 
		\item The vertices in each layer $V^i$ for $i \in [\dd]$ are partitioned into $\ww$ \textbf{groups} 
		\[
		g^i_1=({a^i_1,b^i_1}), \quad g^i_2 = (a^i_2,b^i_2), ~~ \cdots ~~, \quad g^i_\ww = ({a^i_{\ww},b^i_{\ww}}).
		\] 
		(we always pick the lexicographically-first way of grouping the vertices).  
		\item  The edges of $G$ are $\dd-1$ \emph{perfect matchings} $M^1,\ldots,M^{\dd-1}$ between consecutive layers such that for any $i \in [\dd-1]$ and $j \in [\ww]$, both vertices in group $g^i_j$ are matched to vertices of a single group $g^{i+1}_{j'}$ for $j' \in [w]$ in the next layer. 
	\end{enumerate}
	
	We refer to $\ww$ as the \textbf{width} of the layered graph and $\dd$ as its \textbf{depth}. We use $\LG_{\ww,\dd}$ to denote the set of all layered graphs with width $\ww$ and depth $\dd$. 
	
	Finally, for any $i \in [\dd]$ and $j \in [\ww]$, we define $P_G(j) = j'$ where $j'$ is the index of the \emph{unique} group $g^\dd_{j'}$ in the last layer whose vertices are reachable from those of group $g^1_j$ in the first layer. 
\end{Definition}
\begin{figure}[h!]\label{fig:gplayergraph}
	\centering
	\input{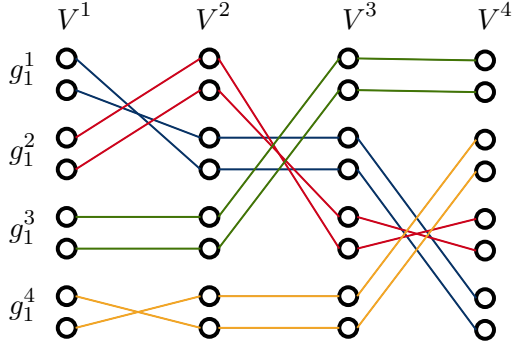}\caption{Illustration of a group layered graph in $\LG_{4, 4}$ from \Cref{def:layered-graph}.}
\end{figure}

We need a simple notation for \emph{concatenating} group-layered graphs. For any $\ww,\dd_1,\dd_2 \geq 1$, let $G_1 \in \LG_{\ww,\dd_1}$ and $G_2 \in \LG_{\ww,\dd_2}$ be two layered graphs
on layers $(V^1,\ldots,V^{\dd_1})$ and $(U^1,\ldots,U^{\dd_2})$. We use $G_1 \conc G_2$ to denote the graph consisting of $\dd_1 + \dd_2-1$ layers $(W^1,\ldots,W^{\dd_1+\dd_2-1})$ such that: 
\[
(W^1,\ldots,W^{\dd_1}) = (V^1,\ldots,V^{\dd_1}) \quad \text{and} \quad (W^{\dd_1},\ldots,W^{\dd_1+\dd_2-1}) = (U^1,\ldots,U^{\dd_2}).
\]
We are setting $W^{\dd_1} = V^{\dd_1} = U^{1}$ (as both sets are identified by $[2\ww]$). The edges of $G_1 \conc G_2$ are the union of edges 
of $G_1$ and $G_2$. This makes $G_1 \conc G_2$ a grouped-layered graph in $\LG_{\ww,\dd_1+\dd_2-1}$. 
We also observe the following about how the edges and functions $P_{G_1}(\cdot)$, $P_{G_2}(\cdot)$ interact in $G = G_1 \conc G_2$. 

\begin{observation}\label{obs:concatenate-2}
	Let $G=G_1 \conc G_2$ for $G_1 \in \LG_{\ww, \dd_1}$, $G_2 \in \LG_{\ww,\dd_2}$; for  $j \in [\ww]$, $P_G(j) = P_{G_2}(P_{G_1}(j))$. 
\end{observation}
\begin{proof}
$P_{G_1}(j)$ takes the group $g^1_j$ to a group $g^{\dd_1}_{j'}$ in the ``middle'' layer of $G$ (the shared layer of $G_1$ and $G_2$ in $G$), and $P_{G_2}(j')$ then takes $g^{\dd_1}_{j'}$ to $g^{\dd_1+\dd_2-1}_{j''}$, which makes $P_G(j) = P_{G_2}(j') = j''$. 
\end{proof}

We need some more structure on the edges in a group-layered graph in our hard distribution, which motivates the following definitions (see~\Cref{fig:xorpermmatching}). 	

\begin{Definition}[\textbf{XOR-Matching}]\label{def:xor-matching}
	Let $x \in \set{0,1}^\ww$ be a string. We define \textbf{XOR-matching} of $x$ as the grouped-layered graph $G \in \LG_{\ww,2}$ on two layers, denoted by $G=\xormatch{x}$, 
	so that its single perfect matching $M$ is defined as follows: 
	
	\smallskip
	\noindent
	For any $j \in [\ww]$ and groups $g^1_j = ({a^1_j,b^1_j})$ and $g_j^{2} = (a^{2}_j,b^{2}_j)$, if $x_j = 0$, then both edges $(a^1_j,a^{2}_j),(b^1_j,b^{2}_j)$ belong to $M$, 
	and otherwise both edges $(a^1_j,b^{2}_j),(b^1_j,a^{2}_j)$ belong to $M$. 
\end{Definition}


\begin{observation}\label{obs:xor-matching}
	In  $G=\xormatch{x}$, for any $j \in [\ww]$, $P_G(j) = j$. 
\end{observation}

\begin{proof}
	For any group $g^1_j \in V^1$, irrespective of $x_j$,  the vertices $a^1_j, b^1_j$ are connected by edges only to vertices of the group $g^2_j \in V^2$. Thus, $P_G(j) = j$ for each $j \in [\ww]$. 
\end{proof}

\begin{Definition}[\textbf{Perm-Matching}]\label{def:perm-matching}
	Let $\sigma \in \mathcal{S}_\ww$ be a permutation. We define \textbf{perm-matching} of $\sigma$ as the grouped-layered graph $G \in \LG_{\ww,2}$ on two layers, denoted by $G=\permmatch{\sigma}$, 
	so that its single perfect matching $M$ is defined as follows: 
	
	\smallskip
	\noindent
	For any $j \in [\ww]$ and group $g^1_j = ({a^1_j,b^1_j})$, the edges $(a^1_j,a^{2}_{\sigma(j)})$ and $(b^1_j,b^{2}_{\sigma(j)})$ belong to $M$. 
\end{Definition}

\begin{observation}\label{obs:perm-matching}
	In  $G=\permmatch{x}$, for any $j \in [\ww]$, $P_G(j) = \sigma(j)$. 
\end{observation}

\begin{proof}
	The vertices in group $g^1_j$ for any $j \in [\ww]$ are mapped by edges to vertices in group $g^2_{\sigma(j)}$ in $\permmatch{\sigma}$, making $P_G(j) = \sigma(j)$. 
\end{proof}


\begin{figure}[h!]
\centering
  \subcaptionbox{$\xormatch{x}$ for  $x = 1001$. \label{fig:xormatching}}%
  [.45\linewidth]{
	\centering
	\tikzset{every picture/.style={line width=0.75pt}} 

\begin{tikzpicture}[x=0.75pt,y=0.75pt,yscale=-1,xscale=1]
	
	\draw  [line width=1.5]  (110,35.5) .. controls (110,33.01) and (112.01,31) .. (114.5,31) .. controls (116.99,31) and (119,33.01) .. (119,35.5) .. controls (119,37.99) and (116.99,40) .. (114.5,40) .. controls (112.01,40) and (110,37.99) .. (110,35.5) -- cycle ;
	\draw  [line width=1.5]  (110,51.5) .. controls (110,49.01) and (112.01,47) .. (114.5,47) .. controls (116.99,47) and (119,49.01) .. (119,51.5) .. controls (119,53.99) and (116.99,56) .. (114.5,56) .. controls (112.01,56) and (110,53.99) .. (110,51.5) -- cycle ;
	\draw  [line width=1.5]  (110,75.5) .. controls (110,73.01) and (112.01,71) .. (114.5,71) .. controls (116.99,71) and (119,73.01) .. (119,75.5) .. controls (119,77.99) and (116.99,80) .. (114.5,80) .. controls (112.01,80) and (110,77.99) .. (110,75.5) -- cycle ;
	\draw  [line width=1.5]  (110,91.5) .. controls (110,89.01) and (112.01,87) .. (114.5,87) .. controls (116.99,87) and (119,89.01) .. (119,91.5) .. controls (119,93.99) and (116.99,96) .. (114.5,96) .. controls (112.01,96) and (110,93.99) .. (110,91.5) -- cycle ;
	\draw  [line width=1.5]  (110,115.5) .. controls (110,113.01) and (112.01,111) .. (114.5,111) .. controls (116.99,111) and (119,113.01) .. (119,115.5) .. controls (119,117.99) and (116.99,120) .. (114.5,120) .. controls (112.01,120) and (110,117.99) .. (110,115.5) -- cycle ;
	\draw  [line width=1.5]  (110,131.5) .. controls (110,129.01) and (112.01,127) .. (114.5,127) .. controls (116.99,127) and (119,129.01) .. (119,131.5) .. controls (119,133.99) and (116.99,136) .. (114.5,136) .. controls (112.01,136) and (110,133.99) .. (110,131.5) -- cycle ;
	\draw  [line width=1.5]  (110,155.5) .. controls (110,153.01) and (112.01,151) .. (114.5,151) .. controls (116.99,151) and (119,153.01) .. (119,155.5) .. controls (119,157.99) and (116.99,160) .. (114.5,160) .. controls (112.01,160) and (110,157.99) .. (110,155.5) -- cycle ;
	\draw  [line width=1.5]  (110,171.5) .. controls (110,169.01) and (112.01,167) .. (114.5,167) .. controls (116.99,167) and (119,169.01) .. (119,171.5) .. controls (119,173.99) and (116.99,176) .. (114.5,176) .. controls (112.01,176) and (110,173.99) .. (110,171.5) -- cycle ;
	\draw  [line width=1.5]  (182,35.5) .. controls (182,33.01) and (184.01,31) .. (186.5,31) .. controls (188.99,31) and (191,33.01) .. (191,35.5) .. controls (191,37.99) and (188.99,40) .. (186.5,40) .. controls (184.01,40) and (182,37.99) .. (182,35.5) -- cycle ;
	\draw  [line width=1.5]  (182,51.5) .. controls (182,49.01) and (184.01,47) .. (186.5,47) .. controls (188.99,47) and (191,49.01) .. (191,51.5) .. controls (191,53.99) and (188.99,56) .. (186.5,56) .. controls (184.01,56) and (182,53.99) .. (182,51.5) -- cycle ;
	\draw  [line width=1.5]  (182,75.5) .. controls (182,73.01) and (184.01,71) .. (186.5,71) .. controls (188.99,71) and (191,73.01) .. (191,75.5) .. controls (191,77.99) and (188.99,80) .. (186.5,80) .. controls (184.01,80) and (182,77.99) .. (182,75.5) -- cycle ;
	\draw  [line width=1.5]  (182,91.5) .. controls (182,89.01) and (184.01,87) .. (186.5,87) .. controls (188.99,87) and (191,89.01) .. (191,91.5) .. controls (191,93.99) and (188.99,96) .. (186.5,96) .. controls (184.01,96) and (182,93.99) .. (182,91.5) -- cycle ;
	\draw  [line width=1.5]  (182,115.5) .. controls (182,113.01) and (184.01,111) .. (186.5,111) .. controls (188.99,111) and (191,113.01) .. (191,115.5) .. controls (191,117.99) and (188.99,120) .. (186.5,120) .. controls (184.01,120) and (182,117.99) .. (182,115.5) -- cycle ;
	\draw  [line width=1.5]  (182,131.5) .. controls (182,129.01) and (184.01,127) .. (186.5,127) .. controls (188.99,127) and (191,129.01) .. (191,131.5) .. controls (191,133.99) and (188.99,136) .. (186.5,136) .. controls (184.01,136) and (182,133.99) .. (182,131.5) -- cycle ;
	\draw  [line width=1.5]  (182,155.5) .. controls (182,153.01) and (184.01,151) .. (186.5,151) .. controls (188.99,151) and (191,153.01) .. (191,155.5) .. controls (191,157.99) and (188.99,160) .. (186.5,160) .. controls (184.01,160) and (182,157.99) .. (182,155.5) -- cycle ;
	\draw  [line width=1.5]  (182,171.5) .. controls (182,169.01) and (184.01,167) .. (186.5,167) .. controls (188.99,167) and (191,169.01) .. (191,171.5) .. controls (191,173.99) and (188.99,176) .. (186.5,176) .. controls (184.01,176) and (182,173.99) .. (182,171.5) -- cycle ;
	\draw [color={rgb, 255:red, 65; green, 117; blue, 5 }  ,draw opacity=1 ][fill={rgb, 255:red, 0; green, 0; blue, 0 }  ,fill opacity=1 ][line width=0.75]    (119,35.5) -- (182,51.5) ;
	\draw [color={rgb, 255:red, 65; green, 117; blue, 5 }  ,draw opacity=1 ][fill={rgb, 255:red, 0; green, 0; blue, 0 }  ,fill opacity=1 ][line width=0.75]    (119,51.5) -- (182,35.5) ;
	\draw [color={rgb, 255:red, 65; green, 117; blue, 5 }  ,draw opacity=1 ][fill={rgb, 255:red, 0; green, 0; blue, 0 }  ,fill opacity=1 ][line width=0.75]    (119,91.5) -- (182,91.5) ;
	\draw [color={rgb, 255:red, 65; green, 117; blue, 5 }  ,draw opacity=1 ][fill={rgb, 255:red, 0; green, 0; blue, 0 }  ,fill opacity=1 ][line width=0.75]    (119,75.5) -- (182,75.5) ;
	\draw [color={rgb, 255:red, 65; green, 117; blue, 5 }  ,draw opacity=1 ][fill={rgb, 255:red, 0; green, 0; blue, 0 }  ,fill opacity=1 ][line width=0.75]    (119,115.5) -- (182,115.5) ;
	\draw [color={rgb, 255:red, 65; green, 117; blue, 5 }  ,draw opacity=1 ][fill={rgb, 255:red, 0; green, 0; blue, 0 }  ,fill opacity=1 ][line width=0.75]    (119,131.5) -- (182,131.5) ;
	\draw [color={rgb, 255:red, 65; green, 117; blue, 5 }  ,draw opacity=1 ][fill={rgb, 255:red, 0; green, 0; blue, 0 }  ,fill opacity=1 ][line width=0.75]    (119,155.5) -- (182,171.5) ;
	\draw [color={rgb, 255:red, 65; green, 117; blue, 5 }  ,draw opacity=1 ][fill={rgb, 255:red, 0; green, 0; blue, 0 }  ,fill opacity=1 ][line width=0.75]    (119,171.5) -- (182,155.5) ;

\end{tikzpicture}
	}
  \subcaptionbox{$\permmatch{\sigma}$ for  $\sigma= (3,1,2,4)$. \label{fig:permmatching}}%
  [.45\linewidth]{
	\centering
	\tikzset{every picture/.style={line width=0.75pt}} 

\begin{tikzpicture}[x=0.75pt,y=0.75pt,yscale=-1,xscale=1]
	
	\draw  [line width=1.5]  (110,35.5) .. controls (110,33.01) and (112.01,31) .. (114.5,31) .. controls (116.99,31) and (119,33.01) .. (119,35.5) .. controls (119,37.99) and (116.99,40) .. (114.5,40) .. controls (112.01,40) and (110,37.99) .. (110,35.5) -- cycle ;
	\draw  [line width=1.5]  (110,51.5) .. controls (110,49.01) and (112.01,47) .. (114.5,47) .. controls (116.99,47) and (119,49.01) .. (119,51.5) .. controls (119,53.99) and (116.99,56) .. (114.5,56) .. controls (112.01,56) and (110,53.99) .. (110,51.5) -- cycle ;
	\draw  [line width=1.5]  (110,75.5) .. controls (110,73.01) and (112.01,71) .. (114.5,71) .. controls (116.99,71) and (119,73.01) .. (119,75.5) .. controls (119,77.99) and (116.99,80) .. (114.5,80) .. controls (112.01,80) and (110,77.99) .. (110,75.5) -- cycle ;
	\draw  [line width=1.5]  (110,91.5) .. controls (110,89.01) and (112.01,87) .. (114.5,87) .. controls (116.99,87) and (119,89.01) .. (119,91.5) .. controls (119,93.99) and (116.99,96) .. (114.5,96) .. controls (112.01,96) and (110,93.99) .. (110,91.5) -- cycle ;
	\draw  [line width=1.5]  (110,115.5) .. controls (110,113.01) and (112.01,111) .. (114.5,111) .. controls (116.99,111) and (119,113.01) .. (119,115.5) .. controls (119,117.99) and (116.99,120) .. (114.5,120) .. controls (112.01,120) and (110,117.99) .. (110,115.5) -- cycle ;
	\draw  [line width=1.5]  (110,131.5) .. controls (110,129.01) and (112.01,127) .. (114.5,127) .. controls (116.99,127) and (119,129.01) .. (119,131.5) .. controls (119,133.99) and (116.99,136) .. (114.5,136) .. controls (112.01,136) and (110,133.99) .. (110,131.5) -- cycle ;
	\draw  [line width=1.5]  (110,155.5) .. controls (110,153.01) and (112.01,151) .. (114.5,151) .. controls (116.99,151) and (119,153.01) .. (119,155.5) .. controls (119,157.99) and (116.99,160) .. (114.5,160) .. controls (112.01,160) and (110,157.99) .. (110,155.5) -- cycle ;
	\draw  [line width=1.5]  (110,171.5) .. controls (110,169.01) and (112.01,167) .. (114.5,167) .. controls (116.99,167) and (119,169.01) .. (119,171.5) .. controls (119,173.99) and (116.99,176) .. (114.5,176) .. controls (112.01,176) and (110,173.99) .. (110,171.5) -- cycle ;
	\draw  [line width=1.5]  (182,35.5) .. controls (182,33.01) and (184.01,31) .. (186.5,31) .. controls (188.99,31) and (191,33.01) .. (191,35.5) .. controls (191,37.99) and (188.99,40) .. (186.5,40) .. controls (184.01,40) and (182,37.99) .. (182,35.5) -- cycle ;
	\draw  [line width=1.5]  (182,51.5) .. controls (182,49.01) and (184.01,47) .. (186.5,47) .. controls (188.99,47) and (191,49.01) .. (191,51.5) .. controls (191,53.99) and (188.99,56) .. (186.5,56) .. controls (184.01,56) and (182,53.99) .. (182,51.5) -- cycle ;
	\draw  [line width=1.5]  (182,75.5) .. controls (182,73.01) and (184.01,71) .. (186.5,71) .. controls (188.99,71) and (191,73.01) .. (191,75.5) .. controls (191,77.99) and (188.99,80) .. (186.5,80) .. controls (184.01,80) and (182,77.99) .. (182,75.5) -- cycle ;
	\draw  [line width=1.5]  (182,91.5) .. controls (182,89.01) and (184.01,87) .. (186.5,87) .. controls (188.99,87) and (191,89.01) .. (191,91.5) .. controls (191,93.99) and (188.99,96) .. (186.5,96) .. controls (184.01,96) and (182,93.99) .. (182,91.5) -- cycle ;
	\draw  [line width=1.5]  (182,115.5) .. controls (182,113.01) and (184.01,111) .. (186.5,111) .. controls (188.99,111) and (191,113.01) .. (191,115.5) .. controls (191,117.99) and (188.99,120) .. (186.5,120) .. controls (184.01,120) and (182,117.99) .. (182,115.5) -- cycle ;
	\draw  [line width=1.5]  (182,131.5) .. controls (182,129.01) and (184.01,127) .. (186.5,127) .. controls (188.99,127) and (191,129.01) .. (191,131.5) .. controls (191,133.99) and (188.99,136) .. (186.5,136) .. controls (184.01,136) and (182,133.99) .. (182,131.5) -- cycle ;
	\draw  [line width=1.5]  (182,155.5) .. controls (182,153.01) and (184.01,151) .. (186.5,151) .. controls (188.99,151) and (191,153.01) .. (191,155.5) .. controls (191,157.99) and (188.99,160) .. (186.5,160) .. controls (184.01,160) and (182,157.99) .. (182,155.5) -- cycle ;
	\draw  [line width=1.5]  (182,171.5) .. controls (182,169.01) and (184.01,167) .. (186.5,167) .. controls (188.99,167) and (191,169.01) .. (191,171.5) .. controls (191,173.99) and (188.99,176) .. (186.5,176) .. controls (184.01,176) and (182,173.99) .. (182,171.5) -- cycle ;
	\draw [color={rgb, 255:red, 208; green, 2; blue, 27 }  ,draw opacity=1 ][fill={rgb, 255:red, 0; green, 0; blue, 0 }  ,fill opacity=1 ][line width=0.75]    (119,35.5) -- (182,115.5) ;
	\draw [color={rgb, 255:red, 208; green, 2; blue, 27 }  ,draw opacity=1 ][fill={rgb, 255:red, 0; green, 0; blue, 0 }  ,fill opacity=1 ][line width=0.75]    (119,51.5) -- (182,131.5) ;
	\draw [color={rgb, 255:red, 208; green, 2; blue, 27 }  ,draw opacity=1 ][fill={rgb, 255:red, 0; green, 0; blue, 0 }  ,fill opacity=1 ][line width=0.75]    (119,91.5) -- (182,51.5) ;
	\draw [color={rgb, 255:red, 208; green, 2; blue, 27 }  ,draw opacity=1 ][fill={rgb, 255:red, 0; green, 0; blue, 0 }  ,fill opacity=1 ][line width=0.75]    (119,75.5) -- (182,35.5) ;
	\draw [color={rgb, 255:red, 208; green, 2; blue, 27 }  ,draw opacity=1 ][fill={rgb, 255:red, 0; green, 0; blue, 0 }  ,fill opacity=1 ][line width=0.75]    (119,115.5) -- (182,75.5) ;
	\draw [color={rgb, 255:red, 208; green, 2; blue, 27 }  ,draw opacity=1 ][fill={rgb, 255:red, 0; green, 0; blue, 0 }  ,fill opacity=1 ][line width=0.75]    (119,131.5) -- (182,91.5) ;
	\draw [color={rgb, 255:red, 208; green, 2; blue, 27 }  ,draw opacity=1 ][fill={rgb, 255:red, 0; green, 0; blue, 0 }  ,fill opacity=1 ][line width=0.75]    (119,155.5) -- (182,155.5) ;
	\draw [color={rgb, 255:red, 208; green, 2; blue, 27 }  ,draw opacity=1 ][fill={rgb, 255:red, 0; green, 0; blue, 0 }  ,fill opacity=1 ][line width=0.75]    (119,171.5) -- (182,171.5) ;

\end{tikzpicture}
	}
  \caption{An illustration of XOR-matchings and perm-matchings in $\LG_{4,2}$ defined in~\Cref{def:xor-matching,def:perm-matching}.}\label{fig:xorpermmatching}
\end{figure}
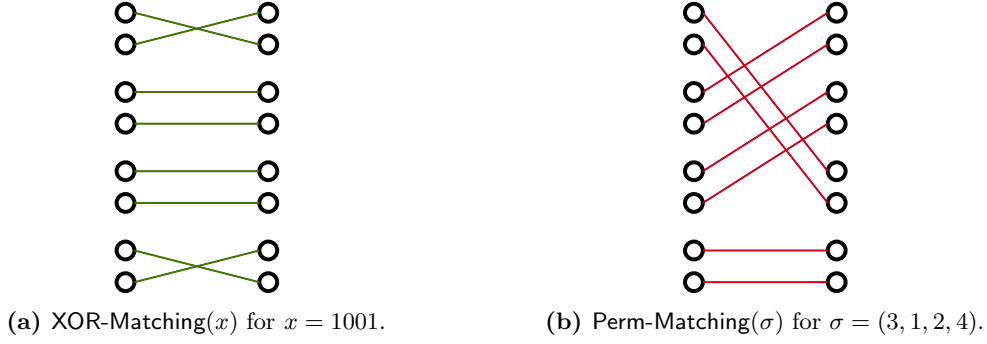

The basic structure in our hard distribution is a block which is described next. It is constructed using  \Cref{def:xor-matching} and \Cref{def:perm-matching}, and is used extensively throughout the rest of the paper. 

\begin{Definition}[\textbf{Block}]\label{def:block}
	Let $x \in \set{0,1}^\ww$, $\sigma \in \mathcal{S}_\ww$, and 
	\[
	G_1 = \permmatch{\sigma}, \quad G_2 = \xormatch{x}, \quad G_3 = \permmatch{\invsigma}.
	\]
	We define \textbf{block} $B := \block{x,\sigma} \in \LG_{\ww,4}$ as $B =  G_1 \conc G_2 \conc G_3$. 
\end{Definition}

\begin{figure}[h!]\label{fig:block}
	\centering
	\input{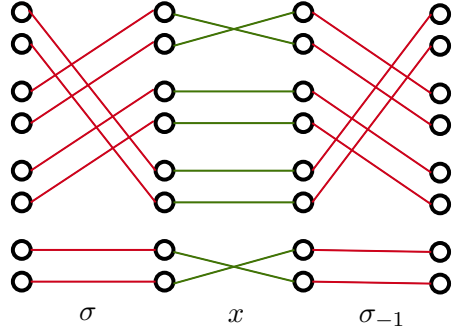}\caption{Illustration of $\block{x, \sigma}$ for $x = 1001$ and $\sigma= (3,1,2,4)$ from \Cref{def:block}.}
\end{figure}

A few useful properties obeyed by paths in blocks are stated now. 

\begin{observation}\label{obs:block}
	In any $B = \block{x,\sigma} \in \LG_{\ww,4}$, we have, 
	\begin{enumerate}[label=$(\roman*)$]
		\item For  $j \in [\ww]$,  $P_B(j) = j$; 
		\item For  $j \in [\ww]$, if $x_{\sigma(j)}=0$, then $a^1_j \rightsquigarrow a^4_j$ and $b^1_j \rightsquigarrow b^4_j$; else,  $a^1_j \rightsquigarrow b^4_j$ and $b^1_j \rightsquigarrow a^4_j$.
	\end{enumerate}
\end{observation}
\begin{proof}
	We have, 
	\[
		P_B(j) = P_{G_3}(P_{G_2}(P_{G_1}(j))) = \invsigma(\sigma_I ( \sigma(j))) = j, 
	\]
	where $\sigma_I$ is the identity permutation, the first equality is by~\Cref{obs:concatenate-2}, and the second one is by applying~\Cref{obs:perm-matching},~\Cref{obs:xor-matching}, and~\Cref{obs:perm-matching} again.  
	This proves part $(i)$. 

	If $x_{\sigma(j)} = 0$, $(a^1_j \rightarrow a^2_{\sigma(j)} \rightarrow a^3_{\sigma(j)} \rightarrow a^4_{j})$ is a path from $a^1_j$ to $a^4_j$ , and $(b^1_j \rightarrow b^2_{\sigma(j)} \rightarrow b^3_{\sigma(j)} \rightarrow b^4_{j})$ is a path from $b^1_j$ to $b^4_j$. If $x_{\sigma(j)} = 1$, $(a^1_j \rightarrow a^2_{\sigma(j)} \rightarrow b^3_{\sigma(j)} \rightarrow b^4_{j})$ is a path from $a^1_j$ to $b^4_j$, and $(b^1_j \rightarrow b^2_{\sigma(j)} \rightarrow a^3_{\sigma(j)} \rightarrow a^4_{j})$ is a path from $b^1_j$ to $a^4_j$, proving part $(ii)$. 
\end{proof}

We can concatenate multiple blocks to get larger group-layered graphs as follows. 

\begin{Definition}[\textbf{Multi-Block Graph}]\label{def:multi-block}
	Let $\ww, t \geq 1$, $X = ({x^1,\ldots,x^t}) \in (\set{0,1}^{\ww})^{t}$, and $\Sigma = (\sigma^1,\ldots,\sigma^{t}) \in (\mathcal{S}_\ww)^{t}$. 
	We define \textbf{multi-block graph} $G := \multiblock{X,\Sigma}$ as the group-layered graph $G \in \LG_{\ww,3t+1}$ obtained as $G:= B_1 \conc B_2 \conc \ldots \conc B_t$, 
	where $B_i = \block{x^i,\sigma^i}$. 
	
	For any $j \in [\ww]$, we define $z_G(j) :=  \bigoplus_{i \in [t]} x^i_{\sigma^i(j)}$. 
\end{Definition}

\begin{figure}[h!]
	\centering
	\input{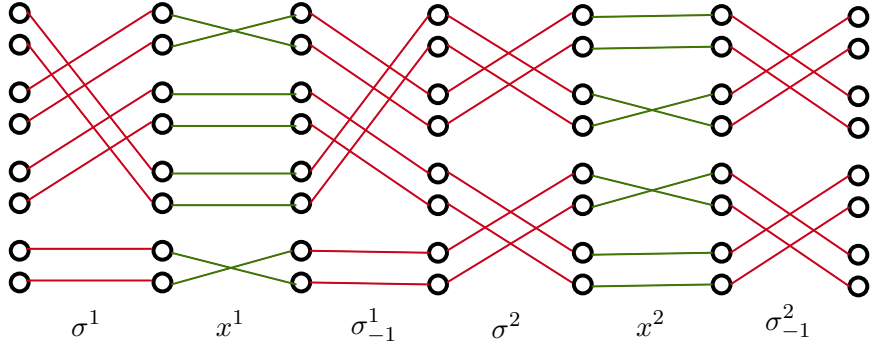}
	\caption{An illustration of $\multiblock{X, \Sigma}$ for $X = (1001, 0110)$,  and $\Sigma = \paren{(3,1,2,4), (2,1,4,3)}$.}\label{fig:multi-block}
\end{figure}

\Cref{obs:block} on blocks can be extended to multi-block graphs as well. 

\begin{observation}\label{obs:multi-block}
	In any $G = \multiblock{X,\Sigma} \in \LG_{\ww,\dd}$, we have, 
	\begin{enumerate}[label=$(\roman*)$]
		\item For $j \in [\ww]$,  $P_G(j) = j$;
		\item For  $j \in [\ww]$, if $z_G(j)=0$, then $a^1_j \rightsquigarrow a^{\dd}_j$ and $b^1_j \rightsquigarrow b^{\dd}_j$; else,  $a^1_j \rightsquigarrow b^{\dd}_j$ and $b^1_j \rightsquigarrow a^{\dd}_j$.
	\end{enumerate}
\end{observation}

\begin{proof}
	Let $G_i = B_1 \conc B_2 \ldots  \conc B_i$ be used to denote the concatenation of the first $i$ blocks for $i \in [t]$ with $G = G_t$. Each $G_i \in\LG_{\ww,3i+1}$ for $i \geq 2$ and $G_1 \in \LG_{\ww, 4}$.  These statements can be proven by a straightforward induction on $i$.
	
	To prove part $(i)$, we know $P_{B_1}(j) = j$ from \Cref{obs:block}. Let us assume $P_{G_{i-1}}(j) = j$. 
	\Cref{obs:concatenate-2} gives $P_{G_i}(j) = P_{B_i}(P_{G_{i-1}}(j)) = P_{B_i}(j) = j$ for each $i \in [t]$. 
	
	We now prove part $(ii)$. It is easy to see that $z_{G_i}(j) = z_{G_{i-1}}(j) \oplus x^i_{\sigma_i(j)}$ and $z_{B_1}(j) = x^1_{\sigma_1(j)}$. From \Cref{obs:block}, $a^1_j \rightsquigarrow a^4_j$ and $b^1_j \rightsquigarrow b^4_j$ if $z_{G_1}(\sigma(j)) = 0$. If $z_{G_1}(\sigma(j)) = 1$, $a^1_j \rightsquigarrow b^4_j$ and $b^1_j \rightsquigarrow a^4_j$. Let us assume that the statement is true for the graph $G_{i-1}$ with $1 \leq i < t$. 
	We know $G_i = G_{i-1} \conc B_i$, hence the layer $V^{3(i-1)+1} $ is treated as the first layer of $B_i$. 
	We analyze the two possible cases for the value taken by $z_{G_i}(j)$ as follows.
	\begin{itemize}
		\item If $z_{G_{i-1}}(j) = x^i_{\sigma^i(j)} $, then $z_{G_i}(j) = 0$. If $z_{G_{i-1}}(j) = x^i_{\sigma^i(j)} = 0$, we have paths from $a^1_j \rightsquigarrow a^{3(i-1)+1}_j \rightsquigarrow a^{3i+1}_j$ and $b^1_j \rightsquigarrow b^{3(i-1)+1}_j \rightsquigarrow  b^{3i+1}_j$.  If $z_{G_{i-1}}(j) = x^i_{\sigma^i(j)} = 1$, then we have paths from $a^1_j \rightsquigarrow b^{3(i-1)+1}_j \rightsquigarrow a^{3i+1}_j$ and $b^1_j \rightsquigarrow a^{3(i-1)+1}_j \rightsquigarrow  b^{3i+1}_j$. 
		\item If $z_{G_i}(j) = 1$, then either $z_{G_{i-1}}(j) = 0, x^i_{\sigma^i(j)} = 1$, where we have with paths from $a^1_j \rightsquigarrow a^{3(i-1)+1}_j \rightsquigarrow b^{3i+1}_j$ and $b^1_j \rightsquigarrow b^{3(i-1)+1}_j \rightsquigarrow  a^{3i+1}_j$,
		or $z_{G_{i-1}}(j) = 1, x^i_{\sigma^i(j)} = 0$, with paths from $a^1_j \rightsquigarrow b^{3(i-1)+1}_j \rightsquigarrow b^{3i+1}_j$ and $b^1_j \rightsquigarrow a^{3(i-1)+1}_j \rightsquigarrow  a^{3i+1}_j$. \qedhere
	\end{itemize}
\end{proof}

This concludes our setup. 
We now have all the constructs necessary to describe our hard distribution for $\NGC_{n,k}$. 

\subsection{A Hard Distribution for NGC}\label{subsec:harddistribution}

In this section, we will describe our hard distribution $\distNGC$ for $\NGC_{n,k}$ only for \emph{certain values} of $k$, and prove the validity of the distribution (the extension to all choices of $k$ is straightforward and is done
in the proof of~\Cref{thm:main-ngc} itself). 

Our distribution will be a multi-block graph from \Cref{def:multi-block} with strings $x^i$ and permutations $\sigma^i$ sampled at random from $\{0,1\}^{(n/2k)}$ and $\mathcal{S}_{n/2k}$ respectively, conditioned on correlation of the bit $z_G(j)$ for some \emph{subset} of groups $j$ in the first layer.  We add additional edges to create cycles so that the graph is a valid instance of $\NGC_{n,k}$. Formally, 

\begin{tbox}
	\textbf{A hard distribution $\distNGC$ for $\NGC_{n,k}$} (when $k=3t+1$ for $t \geq 1$ and $n=4k \cdot m$ for $m \geq 1$).
	\begin{enumerate}[label=$(\roman*)$]
		\item Sample $\theta \in \set{0,1}$ uniformly at random. 
		\item Sample $X=\paren{x^1,\ldots,x^t} \in \paren{\set{0,1}^{2m}}^{t}$ and $\Sigma = (\sigma^1,\ldots,\sigma^{t}) \in \paren{\mathcal{S}_{2m}}^{t}$ uniformly at random conditioned on the event that for all $j \in [m]$\footnote{We emphasize that this property is only
			for the first $m$ indices and not all $2m$ of them.}:
		\[
		\bigoplus_{i \in [t]} x^i_{\sigma^i(j)} = \theta.
		\]
		\item Let $G = \multiblock{X,\Sigma} \in \LG_{2m,k}$ plus the \emph{auxiliary edges} $(a^{k}_j,a^1_j)$ and $(b^k_j,b^1_j)$ for $j \in [m]$. 
	\end{enumerate}
\end{tbox}

The following figure gives an illustration of the distribution $\distNGC$. 

\begin{figure}[h!]
	\centering
	\resizebox{15cm}{!}{\input{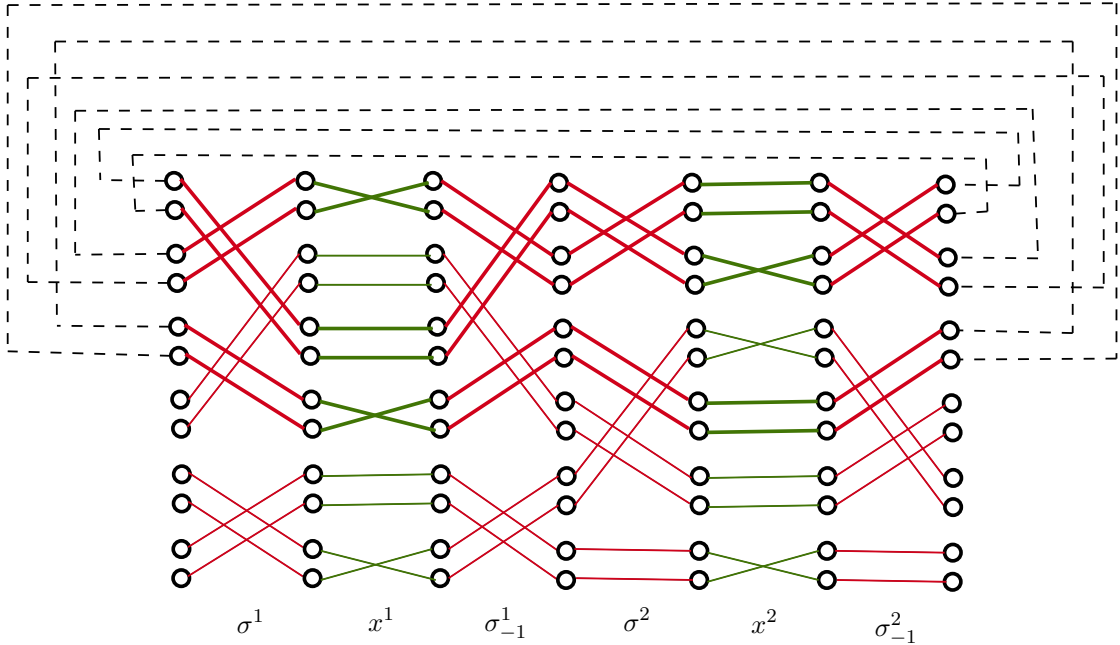}}
	\caption{An illustration of graph $G$ sampled from $\distNGC$ with strings  $X = (100101, 011001)$,  and permutations $\Sigma = \paren{(3,1,4,2,6,5), (2,1,4,5,3,6)}$ when $\theta = 1$. The thicker edges show that  $z_G(j) = 1$ for $j \in [3]$. The dashed edges are the auxiliary edges added to $\multiblock{X, \Sigma}$.}\label{fig:distngc1}
\end{figure}

We prove that the graphs sampled from $\distNGC$ satisfy the guarantee of the $\NGC$ problem. 

\begin{lemma}\label{lem:distNGC}
	For any graph $G \sim \distNGC$ for $\NGC_{n,k}$: 
	\begin{enumerate}[label=$(\roman*)$]
		\item if $\theta = 0$, $G$ has $(n/2k)$ vertex-disjoint cycles of length $k$; 
		\item if $\theta = 1$, $G$ has $(n/4k)$ vertex-disjoint cycles of length $2k$. 
	\end{enumerate}
	In either case, $G$ additionally has $(n/2k)$ vertex-disjoint paths of length $k-1$. 
\end{lemma}

\begin{proof}
	First let us show that irrespective of what $\theta$ is, $G$ has $(n/2k)$ vertex-disjoint paths of length $k$. Consider any group $g^1_j$ in layer $V^1$ of $G$ for $m < j \leq 2m$. We know vertices of $g^1_j$ are connected by a path to vertices of the group $g^k_j$ from \Cref{obs:multi-block} as $P_G(j) = j$. Either paths $(a^1_j \rightsquigarrow a^k_j), (b^1_j \rightsquigarrow b^k_j)$ exist, or paths $(a^1_j \rightsquigarrow b^k_j), (b^1_j \rightsquigarrow a^k_j)$ exist (which of these two cases happens depends on $z_G(j)$ which we have not conditioned on). The length of any such path is $k-1$ as edges in $G$ only go from one layer to the next. Moreover, this path is a separate connected component of $G$, as $G$ is a layered graph with no additional edges to any $a^1_j, a^k_j, b^1_j$ or $b^k_j$ for $m < j \leq 2m$.  There are a total of $2m = (n/2k)$ such vertices in layer $V^1$, and if we consider all the paths connecting these vertices to vertices of $V^k$, we get $(n/2k)$ vertex-disjoint paths. 
	
	If $\theta = 0$, we know $z_G(j) = 0$ for every $j \in [m]$. Hence, paths $(a^1_j \rightsquigarrow a^k_j), (b^1_j \rightsquigarrow b^k_j)$ of length $k-1$ exist by \Cref{obs:multi-block}. But additional edges $(a^1_j, a^k_j)$ and $(b^1_j , b^k_j)$ are added to $G$. Thus, $(a^1_j \rightsquigarrow a^k_j \rightarrow a^1_j)$ is a cycle in $G$ of length $k$. As $G$ is layered and the degree of each vertex is at most 2, this cycle is a connected component of $G$. In fact, there are $2m = n/2k$ vertex-disjoint cycles of length $k$ as $(a^1_j \rightsquigarrow a^k_j \rightarrow a^1_j)$ and $(b^1_j \rightsquigarrow b^k_j \rightarrow b^1_j)$ are cycles for each $j \in [m]$.
	
	If $\theta = 1$, we know $z_G(j) = 1$ for every $j \in [m]$. Paths $(a^1_j \rightsquigarrow b^k_j), (b^1_j \rightsquigarrow a^k_j)$ of length $k-1$ are present in $G$, again by \Cref{obs:multi-block}. These paths are similarly vertex-disjoint from other paths in $G$ as the degree of every vertex is at most 2. Thus, for each $j \in [m]$ we get the cycle $(a^1_j \rightsquigarrow b^k_j \rightarrow b^1_j \rightsquigarrow a^k_j \rightarrow a^1_j)$ of length $2k$ if we incorporate the auxiliary edges, giving us $m = (n/4k)$ vertex disjoint cycles of length $2k$. 
\end{proof}

For any graph $G$ sampled from distribution $\distNGC$, we have $z_G(j) = \theta$ for each $j \in [m]$. Thus, the $X$-bits on the paths from vertices of the first $m$ groups in layer $V^1 $ to vertices in layer $V^k$ are correlated. 
We remove this correlation as the next step to proving our lower bound. 

\subsection{The Distributional Hidden-XOR (\DHXOR) Problem}

In this section, we define the distributional hidden XOR ($\DHXOR$) problem and give a reduction to it from $\NGC_{n,k}$ over distribution $\distNGC$ with a decorrelation step. The intuition behinds this step is as follows. For any algorithm that attempts to solve $\NGC_{n,k}$ on distribution $\distNGC$, it is sufficient to find $z_G(j)$ for \emph{some} $j \in [m]$ with probability at least 2/3. However, we will show that such an algorithm also has to find $z_G(\istar)$ for a \emph{fixed} $\istar \in [m]$ albeit with a lower probability of (roughly) at least $\frac12 + \frac1{m}$.  The following problem and the subsequent lemma capture this intuition. 

\begin{problem}[\textbf{(Distributional) Hidden-XOR Problem} (\DHXOR)]\label{def:dhxor}
	In $\DHXOR_{\ww,t}$, we have a graph $G = \multiblock{X,\Sigma} \in \LG_{\ww,3t+1}$ for $X \in \paren{\set{0,1}^{\ww}}^{t}$ and $\Sigma \in \paren{\mathcal{S}_{\ww}}^t$ chosen independently and uniformly at random. 
	The goal is to output $z_G(1)$ in this graph when the edges of the graph are \emph{randomly} partitioned between Alice and Bob, and Alice can send a single message to Bob.  
\end{problem}

This is the main lemma of this subsection, giving a lower bound for $\NGC_{n,k}$ using $\DHXOR$. 

\begin{lemma}\label{lem:main-ngc}
	For any sufficiently large $n \geq 1$, $k = 3t+1$ for some $t \geq 1$ and $\ww = (n/4k)+1$, 
	\[
	\oDD{\distNGC}{1/3}(\NGC_{n,k}) \geq \oDD{\distDHX}{\frac{1}{2}-\frac{1}{6\ww}}(\DHXOR_{\ww,t}), 
	\] 
	where $\distDHX$ is the implicit uniform distribution in the definition of $\DHXOR$. 
\end{lemma}

We design a way for Alice and Bob to solve $\DHXOR_{\ww, t}$ using any protocol for $\ngc_{n,k}$ by a hybrid argument. Consider the following series of distributions for $0 \leq h \leq m$. 

\begin{tbox}
	\textbf{Distribution} $\distHNGC(h)$\textbf{.} Sample $X  = \paren{x^1, x^2, \ldots, x^t} \in  \paren{\set{0,1}^{2m}}^{t}$, $\Sigma =\paren{\sigma^1, \sigma^2, \ldots, \sigma^t} \in \paren{\mathcal{S}_{2m}}^t$  independently and uniformly at random conditioned on
	\begin{align*}
		z_G(j) = \bigoplus_{i \in [t]} x^i_{\sigma^i(j)} = 
		\begin{cases}
			0 &\textnormal{ if $j \leq h$,} \\
			1 & \textnormal{  otherwise}.
		\end{cases}
	\end{align*}
\end{tbox}
Notice that the distribution $\distNGC$ is just $\frac12\paren{\distHNGC(0) + \distHNGC(m)}$ along with the fixed auxiliary edges. 

We say that a protocol $\Pi$ \emph{distinguishes} between two distributions $\mathcal{D}_1$ and $\mathcal{D}_2$ with probability $p \in [0,1]$, if given a sample $s$ chosen uniformly from $\mathcal{D}_1$ and $\mathcal{D}_2$, 
running $\Pi$ on $s$ allows us to identify which distribution $s$ was sampled from with probability at least $p$.\footnote{To clarify, in general, the protocol $\Pi$ may not have been designed for this task and thus we cannot 
solely rely on its output (which may not even be well-defined in certain cases). Thus, we rather focus on the maximum likelihood estimator for the joint distribution of the message $\pi$ and Bob's input to distinguish between the two distributions.}
Let $\Pi$ be a protocol for $\NGC_{n,k}$ with probability of success at least $2/3$. 
We  use $\Pi$ in a non-black-box way to identify an index $\istar \in [m]$ such that $\Pi$ distinguishes samples from $\distHNGC(\istar-1)$ and $\distHNGC(\istar)$ w.p. above $1/2$. 

\begin{lemma}\label{lem:hybridindex}
	There exists an index $1 \leq \istar \leq m$ such that protocol $\Pi$ can distinguish the samples from $\distHNGC(\istar-1)$ and $\distHNGC(\istar)$ with probability at least $\frac12 + \frac1{6m}$. 
\end{lemma}

\begin{proof}
	With a slight abuse of notation, let $\outs(\Pi)$ denote the set of random variables used by Bob to determine the output of the protocol, namely, the message received by Alice and his own input. 
	Protocol $\Pi$ identifies whether any sample from $\distNGC$ is sampled from $\distHNGC(0)$ or $\distHNGC(m)$ with probability at least 2/3. 
	Consider two distributions $\outs(\Pi) \mid \distHNGC(0)$ and $\outs(\Pi) \mid \distHNGC(m)$. As we can distinguish these distributions with probability at least $\frac23$ or advantage $\frac16$, by \Cref{prop:tvdsample},  
	\begin{align*}
		\tvd{\outs(\Pi) \mid \distHNGC(0)}{\outs(\Pi) \mid \distHNGC(m)} \geq \frac16.
	\end{align*}
	However from the triangle inequality property of the total variation distance, we also have that,
	\begin{align*}
		\tvd{\outs(\Pi) \mid \distHNGC(0)}{\outs(\Pi) \mid \distHNGC(m)} \leq \sum_{i=1}^{m}\tvd{\outs(\Pi) \mid \distHNGC(i-1)}{\outs(\Pi) \mid \distHNGC(i)}.
	\end{align*}
	By a simple averaging argument, there exists an index $1 \leq \istar \leq m$ such that,
	\begin{align*}
		\tvd{\outs(\Pi) \mid \distHNGC(\istar)}{\outs(\Pi) \mid \distHNGC(\istar-1)} \geq \frac1{6m}. 
	\end{align*}
	That is, $\outs(\Pi)$ can be used to idenfity samples from $\distHNGC(\istar-1)$ and $\distHNGC(\istar)$ with advantage at least $1/6m$, again by \Cref{prop:tvdsample}. 
	Thus, we can use $\Pi$ to distinguish between samples from $\distHNGC(\istar)$ and $\distHNGC(\istar-1)$ with probability at least $\frac12 + \frac1{6m}$. 
\end{proof}

Now that we have an index $\istar \in [m]$, our plan to solve $\DHXOR_{\ww, t}$ is to embed the starting vertex of any instance of $\DHXOR$ given to Alice and Bob in $\istar$ of $\distHNGC$ and use protocol $\Pi$. To do so, Alice and Bob need a way to sample from the distribution $\distHNGC^* = \frac12(\distHNGC(\istar) + \distHNGC(\istar-1))$ and we will show that such a sampling process exists. We need to define some notation on partial assignments and permutations before we proceed. 

\paragraph{Notation.}
For any $I \subseteq [\ww]$,  $\cU_{\ww}^{I}$ denotes the set of all $ x \in \set{0,1,\star}^{\ww}$ such that $x_i \in \{0,1\}$ if $i \in I$ and $x_i = \star$ otherwise. That is,  the indices in subset $I$ are fixed to be either 0 or 1, whereas the indices outside $I$ are not fixed yet, with the ambiguous value represented by $\star$. Similarly, define $\mathcal{S}_{\ww}^I$ as the set of all permutations $\sigma$ of $[\ww]$ which fix $\sigma(i) \in [\ww]$ for $i \in I$, and $\sigma(i) = \star$ for any $i \notin I$.  For any $\sigma \in \mathcal{S}_{\ww}^I$, there are $(\ww-\card{I})!$ ways to extend $\sigma$ to a  permutation $\sigma' \in \mathcal{S}_{\ww}$ by fixing indices $i \notin I$. 

We can now describe the algorithm used by Alice and Bob to sample from $\distHNGC^*$ and solve $\DHXOR_{\ww, t}$. In the following, to avoid ambiguity, we use $Y$ and $\Tau$ as the inputs to
$\multiblock{Y,\Tau}$ for the $\DHXOR$ problem, and $X$ and $\Sigma$ as inputs to some $\multiblock{X,\Sigma}$ for $\NGC$ that $\Pi$ operates on. 

\begin{Algorithm}\label{alg:samplngch}
	A protocol for $\DHXOR_{\ww,t}$ using protocol $\Pi$ for $\NGC_{n,k}$ with $ \ww = \frac{n}{4k}+1, k = 3t+1$. 
	
	\smallskip
	
	\textbf{Input:} An instance of $\DHXOR_{\ww,t}$,  $G = \multiblock{Y, \Tau}$, $Y = \paren{y^1, y^2, \ldots, y^t} \in \paren{\set{0,1}^{\ww}}^t$, and $\Tau = (\ttau^1, \ttau^2, \ldots, \ttau^t) \in \paren{\mathcal{S}_{\ww}}^t$.  
	
	\smallskip
	
	\textbf{Output:} Answer to $\DHXOR_{\ww, t}$ on $G$, i.e., $z_G(1)$.
	\begin{enumerate}[label=$(\roman*)$]
		\item Sample $\Sigma_{[m]\setminus\set{\istar}} \coloneqq (\sigma^1, \ldots, \sigma^t) \in \paren{\mathcal{S}_{2m}^{[m]\setminus\set{\istar}}}^t$ uniformly at random.
		\item Let $I^i_{\circ} = \set{\sigma^i(j) \mid j \in \paren{[m]\setminus \set{\istar}}}$ for $i \in [t]$, that is, $I^i_{\circ}$ is the set of indices that have a non-$\star$ pre-image under permutation $\sigma^i$. The set $I^i_{\circ}$ has $m-1$ elements for all $i \in [t]$. 
		\item Sample $X_{[m]\setminus\set{\istar}} = (x^1, x^2, \ldots, x^t) \in \paren{\cU_{2m}^{I^1_{\circ}} \times \cU_{2m}^{I^2_{\circ}} \times \ldots \cU_{2m}^{I^t_{\circ}}}$ uniformly at random conditioned on the event that,
		\begin{align*}
			\bigoplus_{i \in [t]} x^i_{\sigma^i(j)} = 
			\begin{cases}
				0 &\textnormal{ if $1 \leq j \leq \istar-1$,} \\
				1 & \textnormal{  if $\istar+1 \leq j \leq m$}.
			\end{cases}
		\end{align*}
		\item Lexicographically map indices $[m+1]$ to $[2m] \setminus I^i_{\circ}$ for each $i \in [t]$. This is possible since $\card{[2m] \setminus I^i_{\circ}} = m+1 = \ww$. Denote this mapping by $f_i : [m+1] \rightarrow [2m] \setminus I^i_{\circ}$. 
		\item Update the positions with $\star$ value in $X_{[m]\setminus\set{\istar}}$ and $\Sigma_{[m]\setminus\set{\istar}}$ to get $X$ and $\Sigma$ respectively as follows. For each $i \in [t]$,
		\begin{align*}
			\sigma^i(j) = \begin{cases}
				f_i(\ttau^i(1)) &\textnormal{if $j = \istar$} \\
				f_i(\ttau^i(j-m+1)) &\textnormal{if $m < j \leq 2m$ }
			\end{cases}
			&& \textnormal{and} &&
			x^i_{f_i(j)} = y^i_{j} &\hspace{4mm} \textnormal{for $j \in [m+1]$.}
		\end{align*}
		\item Fix $G' = \multiblock{X, \Sigma}$. 
		\item Run protocol $\Pi$ on $G'$ to identify whether $G' \sim \distHNGC(\istar-1)$ or $G' \sim \distHNGC(\istar)$; in the former case output $0$ and in the latter output $1$.
	\end{enumerate}
\end{Algorithm}

First, we prove that the sampling process used by \Cref{alg:samplngch} gives a graph sampled from $\distHNGC^*$. 
\begin{lemma}\label{lem:hybridsampling}
	Given an instance $G$ of $\DHXOR_{\ww, t}$,  step $(vi)$ of  \Cref{alg:samplngch} samples $G' \sim \distHNGC^*$ without any communication using its randomness and the randomness of instance $G$. 
\end{lemma}

\begin{proof}
	We can view $\distHNGC^*$ as first sampling $t$ permutations from $\mathcal{S}_{\ww}$ independently and uniformly at random, and then sampling $t$ elements of $\{0,1\}^{\ww}$ conditioned on the choice of the permutations to give appropriate target $z_G(i)$ for each $i \in [m]\setminus\set{\istar}$. 
	\begin{itemize}
		\item \textbf{Permutations $\Sigma$ of $G'$:} In step $(i)$, $\sigma^i(j)$ is fixed for $m-1$ values of $j \in [m]\setminus[\istar]$, but these indices are picked uniformly at random from $[2m]$, and independently.  It is enough if the remaining $m+1$ indices in $([2m]\setminus[m]) \cup \set{\istar}$ are uniformly at random mapped to the set of indices in $[2m]$ with  no pre-image under $\sigma^i$, i.e., $[2m] \setminus I^i_{\circ}$. This is true because irrespective of the mapping $f_i$, we know that $\ttau^i $ is chosen uniformly at random from $\mathcal{S}_{\ww}$ and independently. 
		\item \textbf{Strings $X$ of $G'$:} From step $(iii)$, we know that $X_{[m]\setminus\set{\istar}}$ is picked at random and independently, conditioned on the choice of $\Sigma_{[m]\setminus\set{\istar}} $ with $z_{G'}(j) = 0$ for $j< \istar$, and $z_{G'}(j) = 1$ for $\istar+1 \leq j \leq m$ as required by $\distHNGC^*$. The rest of the indices in $X$ are assigned based on $Y$, but $Y$ is picked uniformly at random from $\paren{\set{0,1}^{\ww}}^t$. Hence the rest of the indices in $X$ are chosen with equal probability from $\set{0,1}$, and independently of each other.    
	\end{itemize}
	Finally, the fact that no communication is needed for the sampling follows by construction: the players sample $\Sigma_{[m]\setminus\set{\istar}}$ and $X_{[m]\setminus\set{\istar}}$ using public randomness, which fixes the functions 
	$f_i$'s as well. The rest is obtained privately by each player embedding the parts of their own private inputs $Y$ and $\Tau$ in proper places of $X$ and $\Sigma$. 
\end{proof}

Next, we show that \Cref{alg:samplngch} successfully embeds the starting group 1 of instance $G$ of $\DHXOR_{\ww, t}$ in the group $\istar$ of the sampled graph $G'$ as required. 

\begin{claim}\label{clm:istaris1}
	In graph $G'$ in step (vi) of \Cref{alg:samplngch}, $z_{G'}(\istar) = z_{G}(1)$. 
\end{claim}
\begin{proof}	
	We know for each $i \in [t]$, from step $(v)$,
	\[
	x^i_{\sigma^i(\istar)} = x^i_{f_i(\ttau^i(1))} = y^i_{\ttau^i(1)}.
	\]
	Hence,
	\begin{align*}
		z_{G'}(\istar) = \bigoplus_{i \in [t]} x^i_{\sigma^i(\istar)} 
		=  \bigoplus_{i \in [t]}  y^i_{\ttau^i(1)} = z_{G}(1).   \qedhere
	\end{align*}
\end{proof}

We can now use \Cref{lem:hybridsampling} and \Cref{clm:istaris1} to complete the proof of \Cref{lem:main-ngc}. 

\begin{proof}[Proof of \Cref{lem:main-ngc}]
	Let $\Pi$ be any protocol that solves $\NGC$ on distribution $\distNGC$ with probability at least 2/3. By \Cref{lem:hybridindex}, we know that there exists an $\istar \in [m]$ such that $\Pi$ distinguishes samples from $\distHNGC(\istar-1)$ and $\distHNGC(\istar)$ with probability at least $\frac12 + \frac1{6m}$. 
	
	Given an instance $G$ of $\DHXOR_{\ww, t}$ with $\ww = n/4k + 1$ and $t = (k-1)/3$, by \Cref{lem:hybridsampling}, in \Cref{alg:samplngch} graph $G'$ is sampled from $\distHNGC^* = \frac12\paren{\distHNGC(\istar-1) + \distHNGC(\istar)}$. \Cref{clm:istaris1} implies that identifying whether $G'$ is sampled from $\distHNGC(\istar-1)$ or $\distHNGC(\istar)$ is identical to solving $\DHXOR$ on the instance $G$. Hence Alice and Bob can run \Cref{alg:samplngch} to solve instances of $\DHXOR_{\ww, t}$ with probability at least $\frac12 + \frac1{6m} \geq \frac12 +\frac1{6w}$ since $w = m+1$. 
\end{proof}

\subsection{Proof of~\Cref{thm:main-ngc}}\label{thm:proof-main-ngc}

So far, we have proved that a lower bound for $\DHXOR_{\ww, t}$ with a low probability of success of $\frac12 + \frac1{6w}$ gives a lower bound for the distributional communication complexity of $\NGC_{n,k}$ for some values of $k$. We will state a lower bound for $\DHXOR_{\ww, t}$ even against this low probability, and use this to prove our main result \Cref{thm:main-ngc}. 

\begin{lemma}\label{lem:dhx-lower}
	For sufficiently large $\ww \geq 1$ and $t \geq 512 \cdot \ln w$, 
	\[
	\oDD{\distDHX}{\frac{1}{2}-\frac{1}{6\ww }}(\DHXOR_{\ww,t}) = \Omega(\ww \cdot t), 
	\]
	where $\distDHX$ is the implicit uniform distribution in the definition of $\DHXOR$. 
\end{lemma}

The proof of \Cref{lem:dhx-lower} is the main technical step of our paper and is postponed to the next section. Here, we show that this lemma plus the previous steps concludes the proof of our main lower bound in~\Cref{thm:main-ngc}. 
\begin{proof}[Proof of \Cref{thm:main-ngc}]
	Our hard distribution $\distNGC$ only gives instances with $k  = 3t+1$ for some $t > 0$. However, we can extend it to a hard distribution for any $\NGC_{n,k}$ by a simple padding argument. If $k \equiv 0 \mod 3$, then we pick a sample from $\distNGC$ for $\NGC_{n-2n/k,k-2}$, and add two additional layers $(V^{-1}, V^0)$ with $n/k$ vertices in each. We add identity perfect matchings from $V^{-1}$ to $V^0$ and from $V^{0}$ to $V^1$. The auxiliary edges are added from layer $V^{k-2}$ to layer $V^{-1}$.  Similarly if $ k \equiv 2 \mod 3$, we sample from $\NGC_{n-n/k,k-1}$ and add one additional layer $V^0$ of $n/k$ vertices with an identity perfect matching from $V^0$ to $V^1$, and auxiliary edges from $V^{k-1}$ to $V^0$. Thus, 
	we can focus on $\distNGC$ for the rest of the proof. 
	
	For $\ww  = (n/4k)+1$ and $t \geq (k-3)/3 \geq 512 \cdot \ln n \geq 512 \cdot \ln \ww$, we have, 
	\begin{align*}
		\oDD{1/3}{\distNGC}(\NGC_{n,k}) \geq \oDD{\distDHX}{\frac{1}{2}-\frac{1}{6\ww }}(\DHXOR_{\ww,t}) = \Omega(n/k \cdot k) = \Omega(n),
	\end{align*}
	where the first inequality follows from \Cref{lem:main-ngc}, and the second bound follows from \Cref{lem:dhx-lower}. Applying Yao's minimax principle in~\Cref{prop:minmax} now implies that 
	\[
		\oRR{1/3}(\NGC_{n,k}) \geq \oDD{1/3}{\distNGC}(\NGC_{n,k}) = \Omega(n), 
	\]
	concluding the proof. 
\end{proof}

\begin{remark}[``Reducing the noise'']\label{rem:reduce}
\emph{
	The notion of ``noise'' in $\NGC$, namely, the extra $(k-1)$-length paths, is used in our decorrelation step. In particular, using $m$ noisy paths allowed 
	us to reduce the lower bound for $\NGC$ to a proving a ``low-probability'' lower bound for $\DHXOR$ on layers of size $m+1$ each. On the other hand, the lower bound for $\DHXOR$ stated in~\Cref{lem:dhx-lower} (and proven in the next section), is 
	$\Omega(\ww \cdot t)$ as long as the layers are of size $w$ (and has some large depth $t$). Thus, as evident from the proof, the choice of $m$ noisy paths in our definition of $\NGC$ is \underline{not sacrosanct} and can be replaced by any other $\Theta(m)$ and still 
	lead to the same asymptotic lower bound (we only chose $m$ for simplicity of exposition). As such, we can effectively obtain the same lower bound of $\Omega(n)$ for $\NGC$ even if we reduce the number of noisy paths to $o(n/k)$.
}

\emph{	
	We could have alternatively entirely removed the noisy paths using a similar reduction as in~\cite[Section 3.5]{AssadiKSY20}. That however would forced us to replace the role of XOR in our construction with gadgets that implement 
	arithmetic modulo constants larger than $2$ (say, $5$). This in turn, requires an analogue of~\Cref{prop:kklparity} for larger fields to prove our lower bound (which do exist, see, e.g.~\cite[Lemma 2.3]{GuruswamiT19}). While this approach \underline{seems plausible}, given the benign role of the noise and the fact that it virtually does not change any of the subsequent reductions from $\NGC$, we opted to forego this step in favor of a simpler and more direct proof. 
	}
\end{remark}

\newcommand{\clean}[1]{\ensuremath{\textnormal{\textsf{clean}}(#1)}\xspace}
\newcommand{\act}[1]{\ensuremath{\textnormal{\textsf{active}}(#1)}\xspace}

\clearpage
\section{A Lower Bound for Distributional Hidden-XOR}\label{sec:hxorlower}

This section focuses on our lower bound for $\DHXOR_{\ww, t}$ in \Cref{lem:dhx-lower}, restated below. 

\begin{lemma}[Restatement of \Cref{lem:dhx-lower}]
	For sufficiently large $\ww \geq 1$ and $t \geq 512 \cdot \ln \ww$, 
	\[
\oDD{\distNGC}{1/3}(\NGC_{n,k}) \geq \oDD{\distDHX}{\frac{1}{2}-\frac{1}{6\ww}}(\DHXOR_{\ww,t}), 
	\]
		where $\distDHX$ is the implicit uniform distribution in the definition of $\DHXOR$. 
\end{lemma}

Our one-way communication lower bound for $\DHXOR_{\ww, t}$ uses the random partitioning of edges between Alice and Bob crucially. We will first show an alternative way of viewing the partitioning the inputs $(X,\Sigma)$ to the multi-block graph in $\DHXOR_{\ww, t}$. Then, we describe configurations of partitions that are favorable to us in that they force Alice to send longer messages and show that these partitions occur with a high probability in instances of $\DHXOR_{\ww, t}$. We use these  to prove that the messages of Alice to Bob must have length $\Omega(\ww t)$ to solve $\DHXOR$ and obtain the final lower bound. 

\subsection{Input Partitioning and Active Blocks}\label{sec:inputactive}

In this section, we argue that the partitioning of the edges in $G$ can be done independently of $(X, \Sigma)$, and define some partitions for which Alice has to send a lot of information about the edges in $E_A$. We will also prove that such partitions of $E$ occur with a high probability. 

Let $G = \multiblock{X,\Sigma} \in \LG_{\ww,3t+1}$ for $X \in (\set{0,1}^{\ww})^{t}$ and $\Sigma \in (\mathcal{S}_\ww)^t$ be chosen independently and uniformly at random (as in the distribution of $\DHXOR_{\ww,t}$). 
Let $B_1,\ldots,B_t$ be the blocks in $G$. We consider the following process for partitioning the edges of each block between Alice and Bob. 

\begin{tbox}
	\textbf{An alternative way of partitioning edges in a block $B$ (see~\Cref{fig:partitions}).}
	\begin{enumerate}[label=$(\roman*)$]
		\item Let $(V^1,V^2,V^3,V^4)$ be the layers in $B$, each a copy of $[2\ww]$, and  $\Left,\Mid,\Right: [2\ww] \rightarrow \set{\Aedge,\Bedge}$ be three given functions.   
		\item Partition the edges of $B$ between Alice and Bob as follows: 
		\begin{enumerate}
			\item For any $u \in V_2$, send the edge $(w,u)$ for $w \in V_1$ to Alice if $\Left(u) = \Aedge$, else to Bob;
			\item For any $u \in V_2$, send the edge $(u,v)$ for $v \in V_3$ to Alice if $\Mid(u)=\Aedge$, else to Bob; 
			\item For any $v \in V_3$, send the edge $(v,z)$ for $z \in V_4$ to Alice if $\Right(v) = \Aedge$, else to Bob. 
		\end{enumerate}
	\end{enumerate}
\end{tbox}

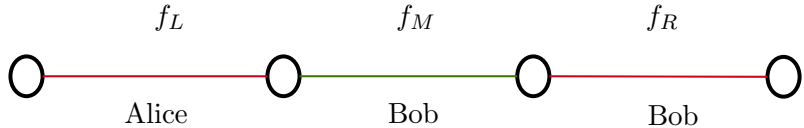
\begin{figure}[h!]
	\centering
	\tikzset{every picture/.style={line width=0.75pt}} 

\begin{tikzpicture}[x=0.75pt,y=0.75pt,yscale=-1,xscale=1]
	
	\draw  [line width=1.5]  (97,90.04) .. controls (97,84.49) and (100.63,80) .. (105.1,80) .. controls (109.57,80) and (113.2,84.49) .. (113.2,90.04) .. controls (113.2,95.58) and (109.57,100.08) .. (105.1,100.08) .. controls (100.63,100.08) and (97,95.58) .. (97,90.04) -- cycle ;
	\draw  [line width=1.5]  (226.6,90.04) .. controls (226.6,84.49) and (230.23,80) .. (234.7,80) .. controls (239.17,80) and (242.8,84.49) .. (242.8,90.04) .. controls (242.8,95.58) and (239.17,100.08) .. (234.7,100.08) .. controls (230.23,100.08) and (226.6,95.58) .. (226.6,90.04) -- cycle ;
	\draw  [line width=1.5]  (352.6,90.04) .. controls (352.6,84.49) and (356.23,80) .. (360.7,80) .. controls (365.17,80) and (368.8,84.49) .. (368.8,90.04) .. controls (368.8,95.58) and (365.17,100.08) .. (360.7,100.08) .. controls (356.23,100.08) and (352.6,95.58) .. (352.6,90.04) -- cycle ;
	\draw  [line width=1.5]  (478.8,90.27) .. controls (478.8,84.73) and (482.43,80.23) .. (486.9,80.23) .. controls (491.37,80.23) and (495,84.73) .. (495,90.27) .. controls (495,95.81) and (491.37,100.31) .. (486.9,100.31) .. controls (482.43,100.31) and (478.8,95.81) .. (478.8,90.27) -- cycle ;
	\draw [color={rgb, 255:red, 208; green, 2; blue, 27 }  ,draw opacity=1 ][fill={rgb, 255:red, 0; green, 0; blue, 0 }  ,fill opacity=1 ][line width=0.75]    (113.2,90.04) -- (226.6,90.04) ;
	\draw [color={rgb, 255:red, 65; green, 117; blue, 5 }  ,draw opacity=1 ][fill={rgb, 255:red, 0; green, 0; blue, 0 }  ,fill opacity=1 ][line width=0.75]    (242.8,90.04) -- (352.6,90.04) ;
	\draw [color={rgb, 255:red, 208; green, 2; blue, 27 }  ,draw opacity=1 ][fill={rgb, 255:red, 0; green, 0; blue, 0 }  ,fill opacity=1 ][line width=0.75]    (368.8,90.04) -- (478.8,90.27) ;
	
	\draw (168,52.4) node [anchor=north west][inner sep=0.75pt]    {$f_{L}$};
	\draw (290,52.4) node [anchor=north west][inner sep=0.75pt]    {$f_{M}$};
	\draw (416,52.4) node [anchor=north west][inner sep=0.75pt]    {$f_{R}$};
	\draw (153,102) node [anchor=north west][inner sep=0.75pt]   [align=left] {Alice};
	\draw (417,103) node [anchor=north west][inner sep=0.75pt]   [align=left] {Bob};
	\draw (285.89,102.2) node [anchor=north west][inner sep=0.75pt]  [rotate=-359.22] [align=left] {Bob};

\end{tikzpicture}\caption{An illustration of a partition for $j \in [2\ww]$ in block $B$ for $f_L(j) = \alpha, f_M(j) = \beta$ and  $f_R(j) = \beta$.}\label{fig:partitions}
\end{figure}

These functions, if picked at random, can be used to give a partition of the edges in $G$ obeying the conditions in \Cref{def:dhxor}, as shown in the following.  

\begin{observation}\label{obs:alternative-sampling}
	Let $G = \multiblock{X,\Sigma} \in \LG_{\ww,3t+1}$ and pick $t$ tuples of functions 
	\[
	\mathcal{F} = (\mathcal{F}_L,\mathcal{F}_M,\mathcal{F}_R) := (\Left^i,\Mid^i,\Right^i: [2\ww] \rightarrow \set{\Aedge,\Bedge})_{i=1}^{t}
	\]
	independently and uniformly at random. The partitioning process
	above using these functions leads to a uniform partitioning of edges of $G$ between Alice and Bob. 
\end{observation}
\begin{proof}
	The partitioning of edges only depends on the three functions $\Left, \Mid$ and $\Right$, and these functions are picked uniformly at random and independently. 
\end{proof}

The importance of these partitions is that $\mathcal{F}$ (henceforth referred to as \textbf{partition functions}) are chosen independently of $(X,\Pi)$. We now use this to state our main definitions in this section. 

\begin{Definition}[\textbf{Clean Indices}]\label{def:clean-index}
	Let $B$ be a block and  $\Left,\Mid,\Right: [2\ww] \rightarrow \set{\Aedge,\Bedge}$ be partition-functions. 
	We say that an index $j \in [\ww]$ is \textbf{clean} iff the following holds for vertices of groups $g^2_j = (a^2_j,b^2_j)$ and $g^3_j = (a^3_j,b^3_j)$ in layers two and three of $B$: 
	\[
	\Left(a^2_j) = \Left(b^2_j) = \Bedge, \quad \Mid(a^2_j) = \Mid(b^2_j) = \Aedge, \quad \text{and} \quad \Right(a^3_j) = \Right(b^3_j) = \Bedge.
	\]
	We let $\clean{B}$ denote the set of clean indices in $B$ under $(\Left,\Mid,\Right)$. 
\end{Definition}

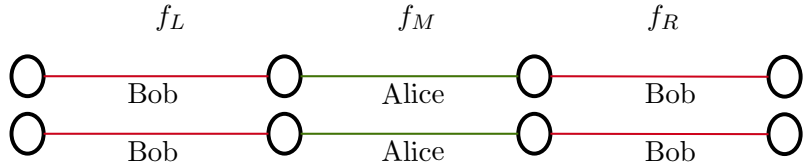
\begin{figure}[h!]
	\centering
	\tikzset{every picture/.style={line width=0.75pt}} 

\begin{tikzpicture}[x=0.75pt,y=0.75pt,yscale=-1,xscale=1]
	
	\draw  [line width=1.5]  (124,106.04) .. controls (124,100.49) and (127.63,96) .. (132.1,96) .. controls (136.57,96) and (140.2,100.49) .. (140.2,106.04) .. controls (140.2,111.58) and (136.57,116.08) .. (132.1,116.08) .. controls (127.63,116.08) and (124,111.58) .. (124,106.04) -- cycle ;
	\draw  [line width=1.5]  (253.6,106.04) .. controls (253.6,100.49) and (257.23,96) .. (261.7,96) .. controls (266.17,96) and (269.8,100.49) .. (269.8,106.04) .. controls (269.8,111.58) and (266.17,116.08) .. (261.7,116.08) .. controls (257.23,116.08) and (253.6,111.58) .. (253.6,106.04) -- cycle ;
	\draw  [line width=1.5]  (379.6,106.04) .. controls (379.6,100.49) and (383.23,96) .. (387.7,96) .. controls (392.17,96) and (395.8,100.49) .. (395.8,106.04) .. controls (395.8,111.58) and (392.17,116.08) .. (387.7,116.08) .. controls (383.23,116.08) and (379.6,111.58) .. (379.6,106.04) -- cycle ;
	\draw  [line width=1.5]  (505.8,106.27) .. controls (505.8,100.73) and (509.43,96.23) .. (513.9,96.23) .. controls (518.37,96.23) and (522,100.73) .. (522,106.27) .. controls (522,111.81) and (518.37,116.31) .. (513.9,116.31) .. controls (509.43,116.31) and (505.8,111.81) .. (505.8,106.27) -- cycle ;
	\draw [color={rgb, 255:red, 208; green, 2; blue, 27 }  ,draw opacity=1 ][fill={rgb, 255:red, 0; green, 0; blue, 0 }  ,fill opacity=1 ][line width=0.75]    (140.2,106.04) -- (253.6,106.04) ;
	\draw [color={rgb, 255:red, 65; green, 117; blue, 5 }  ,draw opacity=1 ][fill={rgb, 255:red, 0; green, 0; blue, 0 }  ,fill opacity=1 ][line width=0.75]    (269.8,106.04) -- (379.6,106.04) ;
	\draw [color={rgb, 255:red, 208; green, 2; blue, 27 }  ,draw opacity=1 ][fill={rgb, 255:red, 0; green, 0; blue, 0 }  ,fill opacity=1 ][line width=0.75]    (395.8,106.04) -- (505.8,106.27) ;
	\draw  [line width=1.5]  (124,135.04) .. controls (124,129.49) and (127.63,125) .. (132.1,125) .. controls (136.57,125) and (140.2,129.49) .. (140.2,135.04) .. controls (140.2,140.58) and (136.57,145.08) .. (132.1,145.08) .. controls (127.63,145.08) and (124,140.58) .. (124,135.04) -- cycle ;
	\draw  [line width=1.5]  (253.6,135.04) .. controls (253.6,129.49) and (257.23,125) .. (261.7,125) .. controls (266.17,125) and (269.8,129.49) .. (269.8,135.04) .. controls (269.8,140.58) and (266.17,145.08) .. (261.7,145.08) .. controls (257.23,145.08) and (253.6,140.58) .. (253.6,135.04) -- cycle ;
	\draw  [line width=1.5]  (379.6,135.04) .. controls (379.6,129.49) and (383.23,125) .. (387.7,125) .. controls (392.17,125) and (395.8,129.49) .. (395.8,135.04) .. controls (395.8,140.58) and (392.17,145.08) .. (387.7,145.08) .. controls (383.23,145.08) and (379.6,140.58) .. (379.6,135.04) -- cycle ;
	\draw  [line width=1.5]  (505.8,135.27) .. controls (505.8,129.73) and (509.43,125.23) .. (513.9,125.23) .. controls (518.37,125.23) and (522,129.73) .. (522,135.27) .. controls (522,140.81) and (518.37,145.31) .. (513.9,145.31) .. controls (509.43,145.31) and (505.8,140.81) .. (505.8,135.27) -- cycle ;
	\draw [color={rgb, 255:red, 208; green, 2; blue, 27 }  ,draw opacity=1 ][fill={rgb, 255:red, 0; green, 0; blue, 0 }  ,fill opacity=1 ][line width=0.75]    (140.2,135.04) -- (253.6,135.04) ;
	\draw [color={rgb, 255:red, 65; green, 117; blue, 5 }  ,draw opacity=1 ][fill={rgb, 255:red, 0; green, 0; blue, 0 }  ,fill opacity=1 ][line width=0.75]    (269.8,135.04) -- (379.6,135.04) ;
	\draw [color={rgb, 255:red, 208; green, 2; blue, 27 }  ,draw opacity=1 ][fill={rgb, 255:red, 0; green, 0; blue, 0 }  ,fill opacity=1 ][line width=0.75]    (395.8,135.04) -- (505.8,135.27) ;
	
	\draw (195,68.4) node [anchor=north west][inner sep=0.75pt]    {$f_{L}$};
	\draw (317,68.4) node [anchor=north west][inner sep=0.75pt]    {$f_{M}$};
	\draw (443,68.4) node [anchor=north west][inner sep=0.75pt]    {$f_{R}$};
	\draw (181,108) node [anchor=north west][inner sep=0.75pt]   [align=left] {Bob};
	\draw (442,108) node [anchor=north west][inner sep=0.75pt]   [align=left] {Bob};
	\draw (308.89,108.2) node [anchor=north west][inner sep=0.75pt]  [rotate=-359.22] [align=left] {Alice};
	\draw (181,137) node [anchor=north west][inner sep=0.75pt]   [align=left] {Bob};
	\draw (442,137) node [anchor=north west][inner sep=0.75pt]   [align=left] {Bob};
	\draw (308.89,137.2) node [anchor=north west][inner sep=0.75pt]  [rotate=-359.22] [align=left] {Alice};

\end{tikzpicture}\caption{An illustration of the partitions for a clean index from \Cref{def:clean-index}.}\label{fig:clean}
\end{figure}

For any clean index $j \in [\ww]$ in block $B = \block{x, \sigma}$, it is clear that only Alice has access to $x_j$, and Bob knows the inverse of $j$ under $\sigma$. Alice may know the inverse based on the other edges, but the edges from $g^1_{\invsigma(j)}$ to $g^{2}_{j} $ and the edges from $g^3_{j}$ to $g^4_{\invsigma(j)}$ are given only to Bob. We  prove that there are a lot of clean indices in any block $B$ using a simple application of Chernoff bound. 

\begin{lemma}\label{lem:clean-index}
	For any fixed block $B$ and a random choice of $\mathcal{F} = (\Left,\Mid,\Right: [2\ww] \rightarrow \set{\Aedge,\Bedge})$: 
	\[
	\Pr\Paren{\card{\clean{B}} \geq \ww/100} \geq 1- 1/\ww^4.
	\] 
\end{lemma}

\begin{proof}
	Let $\indclean_j$ be the indicator random variable for whether index $j \in [\ww]$ is clean in block $B$ under the choice of $\mathcal{F}$. We want to prove that $\indclean = \sum_{j \in [\ww]} \indclean_j$ is at least $\ww/100$ with high probability. There are six edges associated with each index $j $, and each of these edges must be distributed appropriately for $j$ to be clean. The functions are chosen independently and uniformly at random. Hence $\indclean_j = 1$ with probability $1/2^6$, and 0 otherwise for each $j \in [\ww]$ independently.  Chernoff bound from \Cref{prop:chernoff} is applicable here. We know $\expect{\indclean} = \sum_{j=1}^{\ww} \expect{\indclean_j} = \ww/2^6$. Thus,
	\begin{align*}
		\prob{\indclean < \frac{\ww}{100}} &= \prob{\indclean < \paren{1-\frac{36}{100}} \cdot \expect{\indclean}} \leq \exp\paren{\frac{-1}{3}\cdot \frac{0.36^2 \cdot \ww}{64}} \leq \frac1{\ww^4}. \qedhere
	\end{align*} 
	
\end{proof}

To make the analysis simpler, we restrict the number of clean indices in any block $B_i$ for $i \in [t]$ that has at least $w/100$ clean indices to be at most 
\begin{align}
\clw := \frac{w}{100}, \label{eq:clw}
\end{align}
 by  picking the lexicographically-first $\clw$ clean indices and discarding the rest (namely, no longer call them as clean indices). This process again depends only on the partition functions $\mathcal{F}$. 

In any block $B$, if there are a lot of clean indices, from the perspective of Alice, a lot of edges are missing from $V^1$ to $V^2$ and from $V^3$ to $V^4$.  Intuitively, any input index $j$ mapped to a clean index under $\sigma$ is tough for Bob to track, as he does not have any information about $x_{\sigma(j)}$. Alice cannot send $x_{\sigma(j)}$ because she has multiple choices for what $\sigma(j)$ could be.  Thus, if the starting index is mapped to a clean index, the value of $x$ at a lot of indices should be sent to Bob to communicate $x_{\sigma(1)}$. We formalize this intuition with active blocks as defined below.

\begin{Definition}[\textbf{Active Blocks}]\label{def:active-block}
	Let $B=\block{x,\sigma}$  and  $\Left,\Mid,\Right: [2\ww] \rightarrow \set{\Aedge,\Bedge}$ be partition-functions. 
	We say that $B$ is  \textbf{active} if the index $\sigma(1)$ is clean. 
	
	For $G = \multiblock{X,\Sigma}$ and $(\mathcal{F}_L,\mathcal{F}_M,\mathcal{F}_R)$, we define $\act{G}$ as the active blocks in $G$. 
\end{Definition}

\begin{figure}[h!]
	\centering
	\input{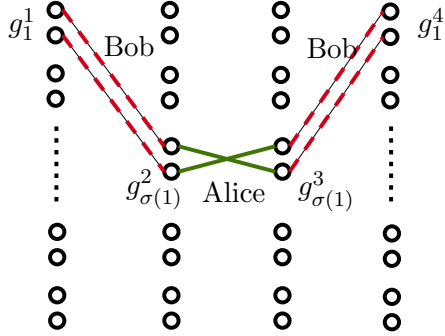}\caption{Illustration of an active block where the dashed (red) edges are given to Bob and the normal (green) edges are given to Alice from \Cref{def:active-block}.}\label{fig:active}
\end{figure}

In an active block $B$, from the perspective of Alice, there are no edges out of group $g^1_1$ and no edges into group $g^4_1$. She does not know the index $\sigma(1)$, and all indices $j \in \clean{B}$ are viable choices for $\sigma(1)$, as they do not have any edges into $g^2_j$ or out of $g^3_j$. The following lemma proves that $\sigma(1)$ can be any of these choices with equal probability. 

\begin{lemma}\label{lem:active-block-pij}
	Suppose we sample $B = \block{x,\sigma}$ and $\Left,\Mid,\Right: [2\ww] \rightarrow \set{\Aedge,\Bedge}$ conditioned on $B$ being active, and the partition functions $\mathcal{F} = (\Left, \Mid, \Right)$. 
	Then, $\sigma(1)$ is chosen uniformly at random from $\clean{B}$.
\end{lemma}

\begin{proof}
	First observe that any block $B$ is active with probability $\clean{B}/{\ww}$, since index $\sigma(1)$ is chosen uniformly at random from $[\ww]$ and $\sigma$ is sampled independently of partition functions $\mathcal{F}$. The set $\clean{B}$ depends only on $\mathcal{F}$. 
	
	Let $c $ be any index in $\clean{B}$. Then,
	\begin{align*}
		\prob{\sigma(1) = c \mid \textnormal{$B$ is active}} &= \frac{\prob{\sigma(1) = c }}{\prob{ \textnormal{$B$ is active}}} \tag{$c \in \clean{B}$, $\sigma(1)=c$ implies $B$ is active}  \\
		&= \prob{\sigma(1)=c} \cdot\frac{\ww}{\clean{B}} \tag{as $B$ is active with probability $\frac{\clean{B}}{\ww}$}  \\
		&= \frac1{\ww} \cdot \frac{\ww}{\clean{B}}  \tag{as $\sigma(1)$ is uniformly random over $[w]$} \\
		&= \frac1{\clean{B}}.  \qedhere
	\end{align*}
\end{proof}

We can also prove that the total number of active blocks in any $G = \multiblock{X, \Sigma}$ is high. 

\begin{lemma}\label{lem:active-block}
	For $G = \multiblock{X,\Sigma}$ and $\mathcal{F}$ sampled in $\DHXOR$, with $t \geq 512 \cdot \ln \ww$: 
	\[
	\Pr\Paren{\card{\act{G}} \geq \ln{\ww}} \geq 1-1/\ww^2. 
	\]
\end{lemma}

\begin{proof}
	Block $B_i$ being active is independent of other blocks for $i \in [t]$, as it depends only on $\sigma^i$ and $\Left^i, \Mid^i, \Right^i$. For any $i \in [t]$,
	\begin{align*}
		\prob{B_i \textnormal{ is active}} &= \prob{\sigma^i(1) \textnormal{ is clean}}  = \frac1{2^6},
	\end{align*}
	as there are 6 edges associated with whether any index is clean, and all these edges must be partitioned according to \Cref{def:clean-index}. 
	Let $R_i$ be the indicator variable for whether block $B_i$ is active. We want to prove that $R = \sum_{i=1}^t R_i$ is larger than $\ln \ww$ with high probability. By linearity of expectation, $\expect{R} = t \cdot \frac1{2^6}  \geq 8 \ln \ww$. Chernoff bound from \Cref{prop:chernoff} can be used here as $R_i$s are independent. 
	\begin{align*}
		\prob{R \geq  \ln \ww} &\geq \prob{R \geq \paren{1-\frac78} \cdot \expect{R}} \geq 1-\exp\paren{\frac{-49}{64} \cdot \frac{8 \ln \ww}{3}} \geq 1- \frac1{\ww^2}. \qedhere
	\end{align*}
\end{proof}

We get our lower bound from graphs with a lot of active blocks, and we have proved that all graphs have at least $ \ln \ww$ active blocks with high probability. In any active block $B_i$, Bob has no information about $x^i_{\sigma(1)}$ as argued next. 
\begin{observation}\label{obs:cond-X}
	Conditioned on a choice of $\act{G}$ for $G=\multiblock{X,\Sigma} \sim \DHXOR$, $X$ is still chosen uniformly at random. 
\end{observation}
\begin{proof}
	Whether a block is active depends only on $\Sigma$ and the partition functions $\mathcal{F}$. It is independent of $X$, thus conditioned on a choice of active blocks, $X$ is still chosen uniformly at random. 
\end{proof}

This concludes our section on the partitioning of the edges in $G$. In the next section, we will use the presence of a lot of active blocks to prove \Cref{lem:dhx-lower}. 

\subsection{The Lower Bound} 

In this section we prove a lower bound for $\DHXOR_{\ww, t}$, as stated in \Cref{lem:dhx-lower} against a probability of success $\frac12 + \frac1{6\ww}$. First, we condition the input graph so that only the edges connected to the clean indices remain random, decreasing the probability of success only slightly. 
Then we show that even with this conditioning, Bob can only have a very low chance of guessing $z_G(1)$ correctly, by relying on~\Cref{prop:kklparity}. 

We know that a lot of active blocks exist and each of these blocks contains a lot of clean indices in any sample $G$ from $\DHXOR_{\ww,t}$ with high probability from the previous section. First, let us condition on this event as the first step towards proving \Cref{lem:dhx-lower}.  

\newcommand{\eventA}{\event_{\texttt{active}}}
\newcommand{\eventC}{\event_{\texttt{clean}}}

\begin{claim}\label{clm:cond-act-g}
	Any protocol $\Pi$ solving $\DHXOR_{\ww, t}$ for $t \geq 2^9 \ln{w}$ with probability at least $\frac12+ \frac1{6w}$, solves $\DHXOR_{\ww, t}$ even conditioned on input graph $G$ satisfying $\card{\act{G}} \geq \ln \ww$ and $\card{\clean{B_i}} \geq \frac{\ww}{100}$ for each $i \in [t]$ with probability at least $\frac12 + \frac1{8\ww} $ for sufficiently large $\ww$. 
\end{claim}
\begin{proof}
	Let $\Pi(G)$ denote the output of protocol $\Pi$ on input graph $G$. Define: 
	\begin{itemize}
	\item event $\eventA$:  the event that there are at least $\ln \ww$ active blocks in $G$;
	\item  event $\eventC$: the event that each block in $G$ contains at least $\clw = \ww/100$ clean indices. 
	\end{itemize}
	Then, for any input $G$,
	\begin{align*}
		\prob{\Pi(G) = z_G(1)} &\leq \prob{\neg \eventA \textnormal{ or }\neg \eventC}+ \prob{\eventA, \eventC} \cdot \prob{\Pi(G) = z_G(1) \mid  \eventA, \eventC} \\
		&\leq \prob{\neg \eventA} +\prob{\neg \eventC} + \prob{\Pi(G) = z_G(1) \mid  \eventA, \eventC} \\
		&\leq \frac1{w^2} +  \prob{\neg \eventC} + \prob{\Pi(G) = z_G(1) \mid \eventA, \eventC} \tag{by \Cref{lem:active-block}} \\
		&\leq \frac1{w^2} + t \cdot \prob{ \card{\clean{B_i}} < \frac{\ww}{100}} + \prob{\Pi(G) = z_G(1) \mid \eventA,\eventC} \tag{by union bound on number of blocks} \\
		&\leq \frac2{w^2} + \prob{\Pi(G) = z_G(1) \mid \eventA,\eventC}.  \tag{by \Cref{lem:clean-index} and $t < \ww$}
	\end{align*}
	Hence, for sufficiently large $w$, 
	\begin{align*}
		\prob{\Pi(G) = z_G(1) \mid \eventA,\eventC} \geq \frac12 + \frac1{6w}-\frac2{w^2} \geq \frac12 + \frac1{8w}. \qedhere
	\end{align*}
\end{proof}

Let the total number of active blocks in graph $G$ be $\actt$. \Cref{clm:cond-act-g} states that even if we condition on $\actt \geq \ln \ww$ and that each block has at least $ \clw = \frac{\ww}{100}$ clean indices, the probability of success reduces only by a $1/24\ww$ additive factor. Thus, throughout the rest of this section, we condition on these events happening in our input graph. 

For simplicity of exposition, let us assume $\act{G} = \set{B_1, B_2, \ldots, B_{\actt}}$, that is only the first $\actt$ blocks are active (which is without loss of generality by renaming the blocks). Our next step is to condition on all the permutations $\sigma^i$ and strings $x^i$ for $i >\actt$, and some parts of the permutations $\sigma^i$, and some indices in $x^i$ for active blocks $i \in [\actt]$. 
We need some notation for these partial permutations and strings before we proceed. 

\paragraph{Notation.}
For any $i \in [t], j \in [\ww]$, we define the tuple $\sigmaandx{i}{j} := (j', j, x^i_j)$ as two indices in $[\ww]$ followed by a value from string $x^i$ such that $\sigma^i(j') = j$. Define \begin{align*}
	\sigmaxunclean{i} := \paren{\sigmaandx{i}{j}}_{j \in [w]\setminus\clean{B_i}} &&\textnormal{and} &&\sigmaxclean{i} := \paren{\sigmaandx{i}{j}}_{j \in \clean{B_i}},
\end{align*} 
as tuples of $\sigmaandx{i}{j}$ for indices that are not clean, and those which are clean in $B_i$ respectively (see~\Cref{fig:qcleanunclean}). 
Let $\clninv{B_i}$ denote the set of indices $j$ such that $\sigma^i(j) \in \clean{B_i}$ for $i \in [t]$. There are exactly $\clw = \frac{\ww}{100}$ indices in $\clninv{B_i}$ since we have restricted the size of each $\clean{B_i}$ to be exactly $\clw$ in the previous section. By definition, if $B_i$ is an active block, $1 \in \clninv{B_i}$. 

\begin{figure}[h!]
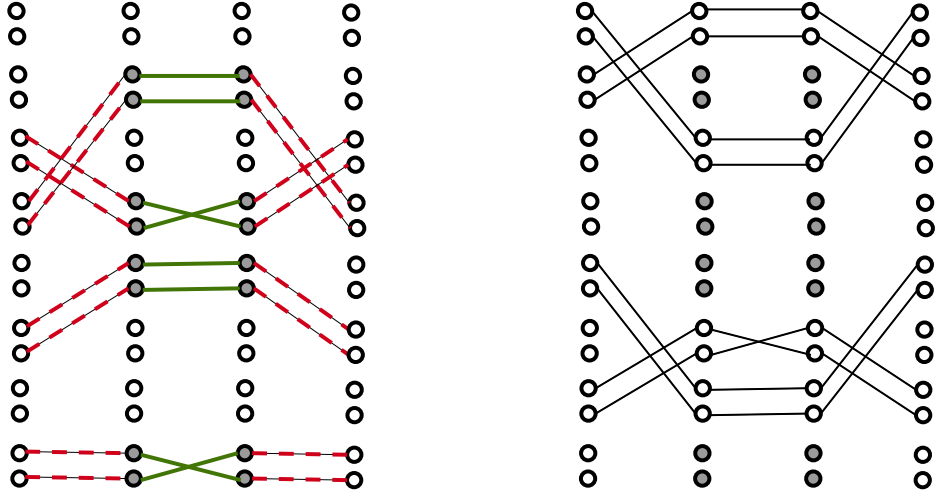

	\centering
	\subcaptionbox{$\sigmaxclean{i}$. The dashed (red) edges are sent to Bob and normal (green) edges are sent to Alice. \label{fig:qclean}}%
	[.45\linewidth]{
		\centering
		\input{qclean}
	}
	\subcaptionbox{$\sigmaxunclean{i}$. \label{fig:qunclean}}%
	[.45\linewidth]{
		\centering
		\input{qunclean}
	}
	\caption{An illustration of $\sigmaxclean{i}$ and $\sigmaxunclean{i}$ for an inactive block $B_i = \block{x, \sigma}$ for $x = (00010101)$ and $\sigma=(3,1,4,2,7,5,6,8)$. The gray vertices in layers $V^2$ and $V^3$ correspond to clean indices.}\label{fig:qcleanunclean}
\end{figure}

Observe that all the edges of block $B_i$ are fixed if both $\sigmaxclean{i}$ and $\sigmaxunclean{i}$ are fixed since they cover the entirety of permutation $\sigma^i$ and all the indices in $x^i$. 
Our aim is to condition on specific choices of both $\sigmaxclean{i}, \sigmaxunclean{i}$ for all blocks $B_i \notin \act{G}$, and only $\sigmaxunclean{i}$ for $B_i \in \act{G}$. This leaves only $\sigmaxclean{i}$ as random for $B_i \in \act{G}$. The following claim shows that this is possible.

\begin{claim}\label{clm:not-clean-fixed} For any protocol $\Pi$ which solves $\DHXOR_{\ww,t}$ on instance $G$ with probability of success at least $\frac12 + \frac1{8w}$ conditioned on $\card{\act{G}} \geq \ln \ww$ and $\card{\clean{B_i}} = \clw$ for each $i \in [t]$, 
	there exists a choice of tuples from the set
	\[ \set{\paren{\sigmaxunclean{i}, \sigmaxclean{i'}, \sigmaxunclean{i'}}_{i \in [\actt], \actt < i' \leq t}}
	\]
	such that the probability of success of $\Pi$ on $G$ conditioned on this choice over the randomness of $\sigmaxclean{i}$, $i \in [\actt]$ is at least $\frac12 + \frac1{8w}$.
\end{claim}

\begin{proof}
	We know that all $\sigma^i$ and $x^i$ are chosen uniformly  and independently for $i \in [t]$. Moreover, each $x^i_j$ is chosen uniformly at random and independently for $j \in [\ww], i \in [t]$. Also, permutation $\sigma^i$ restricted from $\clninv{B_i}$ to $\clean{B_i}$ is uniform over $\mathcal{S}_{\clw}$ (the indices may be numbered differently, but there are $\clw$ of them) conditioned on some choice of the rest of permutation $\sigma^i$ to indices $[\ww] \setminus \clean{B_i}$ for each $i \in [t]$.  
	
	Hence by the independence of these inputs, there exists a choice of permutations and strings in blocks that are not active, $\sigmaxclean{i}, \sigmaxunclean{i}$ for $\actt < i \leq t$, and partial permutations and strings in active blocks $\sigmaxunclean{i}$ for $i \in [\actt]$ such that the probability of success of $\Pi$ on $G$ over the randomness of $\sigmaxclean{i}$, $i \in [\actt]$ is at least $\frac12 + \frac1{8w}$, proving the claim.
\end{proof}

Henceforth, we assume that the tuples $\sigmaxunclean{i}, \sigmaxclean{i'}, \sigmaxunclean{i'}$ for $i \in[\actt]$ and $\actt < i' \leq t$ are fixed for input graph $G$ as the choice from \Cref{clm:not-clean-fixed}. We have conditioned our input on the partition functions $\mathcal{F}$, the active blocks $\act{G} = \set{B_1, B_2, \ldots, B_{\actt}}$ and a specific choice of the tuples from \Cref{clm:not-clean-fixed}, resulting in the following partition of edges of $\sigmaxclean{i}$ for $i \in [\actt]$ (which are still random). 
\begin{itemize}
	\item Alice has $\actt$ strings of length $\clw$ each,  $x^i_j$ for $i \in [\actt], j \in \clean{B_i}$ picked uniformly at random and independently. 
	\item Bob has $\actt$ permutations from $\mathcal{S}_{\clw}$ (though the numbering of the indices may differ), $\sigma^i$ restricted from $\clninv{B_i}$ to $\clean{B_i}$ for $i \in [\actt]$ chosen uniformly and independently. 
\end{itemize}  
In fact, these are the only edges in $G$ which are not fixed by our conditioning of the input graph. By a reordering of the indices in $\clean{B_i}$ for $i \in [\actt]$, we refer to Alice's input as $(x^1, x^2, \ldots, x^{\actt}) \in \paren{\set{0,1}^{\clw}}^{\actt}$. Similarly by a reordering of the indices in $\clninv{B_i}$ for $i \in [\actt]$, we can view Bob's input as $(\sigma^1, \sigma^2, \ldots, \sigma^{\actt}) \in \paren{\mathcal{S}_{\clw}}^{\actt}$. 

Alice and Bob are interested in 
\[
z_G(1) = \bigoplus_{i \in [t]} x^i_{\sigma^i(1)} = \paren{\bigoplus_{i \in [\actt]} x^i_{\sigma^i(1)}} \oplus \paren{\bigoplus_{\actt < i \leq t} x^i_{\sigma^i(1)}}.
\]
The second term $\bigoplus_{\actt < i \leq t} x^i_{\sigma^i(1)}$ is fixed and known to both Alice and Bob based on the choice of these parameters from \Cref{clm:not-clean-fixed}. However, by \Cref{lem:active-block-pij}, we know that $\sigma^i(1)$ is uniform over $\clean{B_i}$ for each $i \in [\actt]$. 
We will prove \Cref{lem:dhx-lower} using \Cref{prop:kklparity}, by showing that any protocol with $o(\clw \cdot \actt)$ communication for $\DHXOR_{\ww, t}$ incurs only a low bias on $\bigoplus_{i \in [\actt]} x^i_{\sigma^i(1)}$. We need some more notation about protocols and their messages for the rest of this section. 

\paragraph{Notation.}
Let $\Pi$ be any protocol for $\DHXOR_{\ww,t}$ with success probability at least $\frac12 + \frac1{6\ww}$, and let $\Pi(G)$ denote the output of $\Pi$ on input graph $G$. With a 
slight abuse of notation, we also use $\Pi$ to denote the random variable for the message $\pi$ sent by the protocol. For any message $\msg$, define $\tomsg{\msg} $ as the set of inputs on which Alice sends $\msg$.
We assume that all the messages sent by protocol $\Pi$ are of the same length $\msglgth$, by a standard padding argument. 

\begin{lemma}\label{lem:avgargument}
	For any protocol $\Pi$ solving $\DHXOR_{\ww, t}$ on $\msglgth$ bits with probability of success at least $\frac12 + \frac1{8\ww}$ over randomness of $\sigmaxclean{i}$ for $i \in [\actt]$, there exists a message $\msg$ such that for at least $2^{\clw\actt-4\msglgth}$ inputs  Alice sends $\msg$, and Bob succeeds with probability at least $\frac12 + \frac1{8w} - \frac1{2^{3\msglgth}}$ on outputting the correct answer given  $\msg$.
\end{lemma}

\begin{proof}
	The proof is by an averaging argument over the total number of possible inputs to Alice and the  number of possible messages. Alice has $\actt$ strings of $\clw$ length each so $2^{\clw \cdot \actt} $ possible inputs. 
	 Denote by $\goodmsg$ the set of messages $\msg$ with $\card{\tomsg{\msg}} \geq 2^{\clw \actt - 4 \msglgth}$. The probability that Alice sends a message $\msg \notin \goodmsg$ can be bounded as follows.  
	\begin{align*}
		\prob{\Pi = \msg, \msg \notin\goodmsg} &= \sum_{\msg \notin \goodmsg} \frac{\card{\tomsg{\msg}}}{2^{\clw\actt}} \tag{as inputs to Alice are uniform} \\
		&< 2^{\msglgth} \cdot  \frac{2^{\clw \actt-4\msglgth}}{2^{\clw \actt}} = 2^{-3\msglgth}.
	\end{align*}
	The total probability of success is at most,
	\begin{align*}
		\prob{\Pi(G) = z_G(1)} &\leq \prob{\Pi(G) = z_G(1) \mid \Pi = \msg, \msg \in \goodmsg} + \prob{\Pi = \msg, \msg \notin \goodmsg} \\
		&\leq  \prob{\Pi(G) = z_G(1) \mid \Pi = \msg, \msg \in \goodmsg} + \frac1{2^{3\msglgth}}. 
	\end{align*}
	
	Conditioned on the message $\msg$ having at least $2^{\clw \actt-4\msglgth}$ inputs mapped to it, the probability of success is at least $\frac12 + \frac1{8\ww} - \frac1{2^{3\msglgth}}$. Hence there exists a message with at least $2^{\clw \actt-4\msglgth}$ inputs mapped to it when protocol $\Pi$ succeeds with probability at least $\frac12 + \frac1{8\ww} - \frac1{2^{3\msglgth}}$.
\end{proof}

We are now ready to prove \Cref{lem:dhx-lower}. We know a lot of inputs are mapped to $\msg$ from \Cref{lem:avgargument}, and the input for Alice is chosen uniformly at random from strings of length $2^{\clw\actt}$. The input to Bob is $\actt$ permutations of $\mathcal{S}_{\clw}$ chosen independently and randomly. We will prove that only a low bias is present on $z_G(1)$ for inputs mapped to $\msg$ using \Cref{prop:kklparity}. 

\begin{proof}[Proof of \Cref{lem:dhx-lower}]
	Firstly, \Cref{clm:cond-act-g,clm:not-clean-fixed} allow us to condition the input graph such that only the parameters in $\sigmaxclean{i}$ for $i \in [\actt]$ are not fixed while keeping the probability of success at least $\frac12 + \frac1{8w}$.
	From \Cref{lem:avgargument}, we know there exists a $\msg$ such that $\tomsg{\msg}$ has at least $2^{\clw \actt-4\msglgth}$ elements on which the protocol succeeds with probability at least $\frac12 + \frac1{8w}-\frac1{2^{3\msglgth}}$. We want to prove that $\msglgth = \Omega(\ww \cdot t)$.  We assume towards a contradiction that $ \msglgth = \frac1{32e} \cdot \clw \actt$. 
	
	 Conditioned on $\Pi = \msg$, the input to Alice is uniformly distributed over the set $\tomsg{\msg}$, as the input is chosen uniformly at random and independently from $\set{0,1}^{\clw \actt}$. By \Cref{prop:kklparity}, the XOR of any $\actt$ size subset $J$ of the indices $[\clw\actt]$ picked uniformly at random of an element from $\tomsg{\msg}$, again sampled uniformly at random incurs a bias $\bias{\tomsg{\msg}}{J}$ of at most,
	\begin{align}\label{eqn:biasub}
		\Exp_J\bracket{\bias{\tomsg{\msg}}{J}^2} = O\paren{\frac1{\clw\actt} \cdot \log \paren{\frac{2^{\clw\actt}}{\card{\tomsg{\msg}}}}}^{\actt}.
	\end{align}

	We are interested in the bias incurred when $J = \set{\sigma^1(1), \sigma^2(1), \ldots, \sigma^{\actt}(1)}$. By \Cref{lem:active-block-pij}, we know that $\sigma^i(1)$ is uniform over the set $\set{\clw(i-1)+1, \clw(i-1)+2, \ldots, \clw\cdot i}$ for $i \in [\actt]$. We call a $\actt$-size subset $J$ of indices \emph{valid} if for all $i \in [\actt]$,
	\begin{align*}
		\card{J \cap \set{\clw(i-1)+1, \clw(i-1)+2, \ldots, \clw \cdot i}}  = 1.
	\end{align*} 
	That is, it contains one index from each substring $x^i$ for $i \in [\actt]$ of the input to Alice, which are the only possible choices for the set $\set{\sigma^i(1) \mid i \in [\actt]}$. 
	We can write the expectation over a random $\actt$ size subset $J$ of $[\clw\actt]$ as,
	\begin{align*}
		\Exp_J\bracket{\bias{\tomsg{\msg}}{J}^2} &\geq \prob{J\textnormal{ is valid}} \cdot \Exp_J\bracket{\bias{\tomsg{\msg}}{J}^2\mid J\textnormal{ is valid}} \\
		&= \paren{\clw}^{\actt} \cdot \frac1{\binom{\clw\actt}{\actt}} \cdot \Exp_J\bracket{\bias{\tomsg{\msg}}{J}^2\mid J\textnormal{ is valid}},
	\end{align*} 
where we have used that $J$ is valid with probability $(\clw)^{\actt} /\binom{\clw\actt}{\actt}$. 
	Bob's input indices $\set{\sigma^i(1) \mid i \in [\actt]}$ are random over the set of all valid $\actt$-size subsets. Therefore the bias over these valid subsets is at most,
	\begin{align*}
		\Exp_J\bracket{\bias{\tomsg{\msg}}{J}^2\mid J\textnormal{ is valid}} &\leq \binom{\clw\actt}{\actt} \cdot \frac1{\paren{\clw}^{\actt}} \cdot O\paren{\frac1{\clw\actt} \cdot \log \paren{\frac{2^{\clw\actt}}{\card{\tomsg{\msg}}}}}^{\actt} \tag{by \Cref{eqn:biasub}}\\
		&\leq e^{\actt} \cdot O\paren{\frac1{\clw\actt} \cdot \log \paren{\frac{2^{\clw\actt}}{\card{\tomsg{\msg}}}}}^{\actt} \tag{as $\binom{\clw\actt}{\actt} \leq \paren{\frac{e\clw\actt}{\actt}}^{\actt}$} \\
		&\leq O\paren{\frac{4e\msglgth}{\clw\actt}}^{\actt} \tag{as $\card{\tomsg{\msg}} \geq 2^{\clw\actt-4\msglgth}$} \\
		&\leq O\paren{1/8}^{\actt} = o(1/\ww^2). \tag{as $c = (1/32e) \cdot \clw \actt$ and $\actt \geq \ln \ww$}
	\end{align*}
	
	However, protocol $\Pi$ succeeds with probability at least $\frac12 + \frac1{8w}-\frac1{2^{3\msglgth}}$ when the transcript is $\msg$, giving a lower bound on the bias:  
	\begin{align*}
		\Pr\paren{\text{$\Pi$ is correct}} = \frac12 + \frac{\Exp_J\bracket{\bias{\tomsg{\msg}}{J}\mid J\textnormal{ is valid}}}2 &\geq \frac12 + \frac1{8w}-\frac1{2^{3\msglgth}}; \\
		\Exp_J\bracket{\bias{\tomsg{\msg}}{J}\mid J\textnormal{ is valid}} &\geq \frac1{4\ww}-\frac1{2^{3\msglgth-1}} \geq \frac1{8\ww}. \tag{as $c \geq \log \ww$}
	\end{align*}
	Comparing the two bounds on the expected bias, we get a contradiction as follows.
	\begin{align*}
		\frac1{64\ww^2} \leq \Exp_J\bracket{\bias{\tomsg{\msg}}{J}\mid J\textnormal{ is valid}}^2 \leq \Exp_J\bracket{\bias{\tomsg{\msg}}{J}^2\mid J\textnormal{ is valid}} \leq o\paren{\frac1{\ww^2}}. 
	\end{align*} 
	This concludes the proof. 
\end{proof}

\clearpage

\section{Implications to Random Order Streaming Algorithms}\label{sec:stream}

In this section, we show some implications of \Cref{thm:main-ngc} to proving space lower bounds for random order streaming algorithms, formalizing~\Cref{res:stream}. We also extend our lower bounds 
to the stochastic streaming model defined in~\cite{CrouchMVW16}. 
The reductions in this section are standard (or follow by easy modifications of known ones), see, e.g.~\cite{VerbinY11,AssadiKSY20,AssadiN21,KapralovMTWZ22}, and are provided here for completeness. The main difference 
with all these prior work is that since we proved a \emph{robust} communication complexity lower bound for $\NGC$, we can obtain streaming lower bounds for \emph{random-order} streams instead of adversarially ordered, as in all these prior work. 

Before moving on, the following important remark is in order. 
\begin{remark}\label{rem:small-eps}
		\emph{Our goal in our paper is to establish an \textbf{exponential-in-$\eps$-dependency} on the space complexity of random-order streaming algorithms that obtain $(1+\eps)$ multiplicative or additive approximation for various graph problems. 
		Hence, in many of our lower bounds, we take $\eps$ to be sub-constant, typically $\eps = 1/\Theta(\log{n})$, and thus obtain lower bounds on specific points of the space-approximation tradeoff curve.\footnote{This means  our lower bounds may leave out the possibility of having an algorithm that for constant/small values of $\eps$ have a better space-dependence
		on $\eps$, but once $\eps$ gets smaller, their space-dependence ``switches'' to exponential. Whether such algorithms can exist for these problems however seems quite unlikely to the authors.} For multiplicative approximation lower bounds, this is 
		not at all a problem as we can use a simple padding argument (by embedding the lower bound on a smaller part of the graph), 
		but this approach does not work when it comes to additive (in $n$) approximation lower bounds. As a result, for additive approximation results, our lower bounds inherently apply only to certain points of the tradeoff curve. We leave open the possibility of 
		extending our results to the entire space-approximation tradeoff curve as an interesting research question. 
		} 
		
		\emph{Finally, we note that given most algorithms for the problems we consider have some (mild) dependence on $n$ as well, say, $\polylog{(n)}$ to store counters or edges, we state our lower bounds 
		by even including some $n^{o(1)}$-dependence on the space; in other words, the exponential-dependence on $\eps$ in our lower bounds continues to hold even when one allows $n^{o(1)}$-space dependence (which is necessary
		to prove near-optimality of our bounds).} 
\end{remark}

\subsection{Number of Connected Components}

We start with one of our main results which is a lower bound for estimating the number of connected components. Given a graph $G=(V,E)$ in a random-order stream, estimate the number of connected components of $G$ to within an $\eps \cdot n$ 
\emph{additive} factor. 

\cite{PengS18} gave an algorithm for this problem with space $(1/\eps)^{O(1/\eps^3)} \cdot (\log{n})$ which was improved later by~\cite{ChipKKP22} to $(1/\eps)^{O(1/\eps)} \cdot (\log{n})$ space. As stated earlier in the introduction, the latter work 
also showed that $(1/\eps)^{\Omega(1/\eps)}$ space is needed for the component \emph{finding} problem, namely, finding a component of size $\Theta(1/\eps)$ in a graph that contains $\Theta(\eps \cdot n)$ of them (while also requiring the 
stream to be almost-random, or rather in the batch random-order model they introduced; see~\cite{ChipKKP22}). 

We prove a near-optimal lower bound for the original estimation problem and on (truly) random-order streams, settling a conjecture of~\cite{PengS18} in the necessity of exponential dependence on $(1/\eps)$ for any algorithm for this problem. 


\begin{corollary}\label{thm:ccest}
	Any single-pass streaming algorithm that for every given $\eps > 0$ processes edges of any graph $G$ with $n$ vertices in a random-order stream 
	in $2^{o(1/\eps)} \cdot n^{o(1)}$ space cannot estimate the number of connected components in $G$ to within an  $\eps \cdot n$ additive factor with probability at least $2/3$. 
\end{corollary}
\begin{proof}
	Consider the $\ngc_{n,k}$ problem for $k \geq 1600 \ln{n}$. In the two possible cases of $\ngc_{n,k}$, the number of connected components of $G$ differs by $n/4k$: the number of noisy paths is the same in both cases of $\NGC$
	and the $k$-cycle case has $n/4k$ more cycles/connected components compared to the $2k$-cycle case.  Thus, if we take $\eps < 1/8k$, then 
	estimating the number of connected components of $G$ in $\ngc_{n,k}$ up to an $\eps \cdot n$ additive factor would also solve this problem. 
	
	Combining this with~\Cref{prop:stream-cc} and our robust communication lower bound of $\Omega(n)$
	for $\ngc_{n,k}$ in~\Cref{thm:main-ngc}, we get that any random order streaming algorithm for estimating the number of connected components up to $\eps \cdot n$ additive factor for $\eps < 1/8k$ requires $\Omega(n)$ space. 
	Given the choice of $\eps = 1/\Theta(\log{n})$, a space of $2^{o(1/\eps)} \cdot n^{o(1)}$ for the algorithm translates to $o(n)$ space, which contradicts the above lower bound. This concludes the proof. 
\end{proof}

\subsection{Minimum Spanning Tree} 

In this section, we prove a lower bound for random-order streaming algorithms which estimate the size of the minimum spanning tree up to an $(1+\eps)$ multiplicative approximation factor. Given a connected weighted graph $G=(V,E)$ with integral weights from the set $ [W]$, estimate the weight of the minimum spanning tree of $G$ to within an $(1+\eps)$ multiplicative approximation factor. 

\cite{PengS18} gave an algorithm for this problem which uses  $O\paren{\frac1{\eps}}^{\Ot(W^3/\eps^3)}$ space. We give the first non-trivial lower bound for this problem. 

\begin{corollary}\label{thm:mstest}
	Any single-pass streaming algorithm that for every given $\eps > 0$ processes edges of any graph $G$ with $n$ vertices and integer weights in $[W]$ in a random-order stream 
	in $2^{o(W/\eps)} \cdot n^{o(1)}$ space cannot estimate the weight of minimum spanning trees to within an $(1+\eps)$ multiplicative approximation  factor with probability at least $2/3$. 
\end{corollary}

\begin{proof}
	Our hard instances for $\NGC_{n,k}$ are sampled from distribution $\distNGC$. Recall that $\distNGC$ outputs valid instances of $\NGC_{n,k}$ by adding auxiliary edges to multi-block graphs such that there are cycles of length $2k$ when $\theta = 1$ and cycles of length $k$ otherwise (see~\Cref{fig:distngc1}). 
	
	Let $G$ be a graph sampled from $\distNGC$ for $k \leq \frac{W-1}{12\eps}$. The weights of all the initial edges in $G$ are set as 1.  We add the following additional edges so that there is a large difference in the weight of the minimum spanning tree when $\theta = 1$ and when $\theta = 0$. 
	\begin{itemize}
		\item Add edge $(a^k_j, b^k_j)$ of weight $W$ for each $j \in [m]$. 
		\item Add edges $(a^1_j, a^1_{m+j}), (a^1_j, b^1_{m+j})$ of weight 1 for each $j \in [m]$. 
		\item Add edges $(a^1_j, b^1_{j+1})$ for each $j \in [m-1]$, and edge $(a^1_m, b^1_1)$ of weight 1.
	\end{itemize}

Let $G'$ be used to denote $G$ after these edges are added. 
	
	When $\theta = 1$, that is the instance $G$ consists of cycles of length $2k$, adding these additional edges makes $G'$ connected through edges of weight 1. There are $n/4k$ cycles of length $2k$, and all the edges in these cycles have weight 1. Each path of length $k-1$ from vertex $a^1_{m+j}$ or $b^1_{m+j}$ is connected to a cycle of length $2k$ containing vertex $a^1_j$ for each $j \in [m]$ by edges of weight 1.  The cycles are all connected to each other through edges $(a^1_1, b^1_2), (a^1_2, b^1_3), \ldots, (a^1_{m-1}, b^1_m), (a^1_m, b^1_1)$ of weight 1. Hence, the weight of the minimum spanning tree is $n-1$. 
	
	When $\theta = 0$ and the instance $G$ consists of cycles of length $k$, the weight of any spanning tree of $G'$ is at least $n-m+W(m-1)$. If we consider only edges of weight 1, there are at least $m$ connected components since, for each $j \in [m-1]$, the component with the vertex $a^1_j$ only contains $4k = n/m$ vertices. It is only connected to the cycle of length $k$ passing through $a^k_j$ with $k$ vertices, two paths of length $k$ beginning at $a^1_{m+j}$ and $b^1_{m+j}$ respectively of $k$ vertices each, and the cycle containing $b^1_{j+1}$ with $k$ vertices. Similarly, the component containing $a^1_m$ contains only $n/m$ vertices. At least $m-1$ edges of weight $W$ connecting $a^k_j$ to $b^k_j$ for $j \in [m]$ are added to any spanning tree of $G'$. 
	
	From our choice of $k$, we know that $(1+\eps)(n-1) < (1-\eps)(n-W+m(W-1))$. Estimating the weight of a minimum spanning tree of $G'$ to within an $(1+\eps)$ multiplicative approximation would also solve $\NGC_{n,k}$ over the hard distribution $\distNGC$. Thus by \Cref{prop:stream-cc} and \Cref{thm:main-ngc}, we know that the space required by such an algorithm is $\Omega(n)$ bits. For $k= \Theta(\log n)$ and $\eps = \Theta(W/k)$, the algorithm requires $\Omega(n) = 2^{\Omega(W/\eps)}$ space. 
\end{proof}

\subsection{Maximum Matching and Matrix Rank}

We next prove a lower bound for random-order streaming algorithms for estimating the size of the maximum matching even in bounded-degree graphs (and planar ones). Given a graph $G=(V,E)$ in a random-order stream with the promise that maximum degree of $G$ 
is bounded by some given $d=O(1)$, estimate the size of a maximum matching in $G$ to within an $\eps \cdot n$ additive factor. 

Maximum matchings have been studied extensively in the streaming model and  listing the prior results on this problem is beyond the scope of our work. We only note that~\cite{KapralovKS14,KapralovMNT20} gave $\polylog{(n)}$-space
streaming algorithms for $\polylog{(n)}$ (multiplicative) approximation of matching size on arbitrary graphs in random-order streams, and~\cite{EsfandiariHLMO15,McGregorV16,CormodeJMM17,McGregorV18} gave $\polylog{(n)}$-space
streaming algorithms for $O(1)$ (multiplicative) approximation of matching size on planar graphs in adversarial-order streams. Most related to us however is the work of~\cite{MonemizadehMPS17} who gave an algorithm 
for estimating the size of maximum matching to within an $\eps \cdot n$ additive factor on bounded-degree graphs (the problem we started with above) in random-order streams using $2^{2^{\text{poly}{(1/\eps)}}} \cdot (\log{n})$ space, i.e., with doubly-exponential
dependence on $\eps$. 

We prove the first lower bound for this problem showing that at least an exponential dependence on $\eps$ is necessary even on an extremely simple family of graphs.

\begin{corollary}\label{thm:matchingest}
	Any single-pass streaming algorithm that for every given $\eps > 0$ processes edges of any graph $G$ with $n$ vertices in a random-order stream 
	in $2^{o(1/\eps)} \cdot n^{o(1)}$ space cannot estimate the size of maximum matching in $G$ to within an  $\eps \cdot n$ additive factor with probability at least $2/3$. 
	
	The lower bound holds even for planar graphs of maximum degree two, namely, on disjoint-unions of cycles and paths. 
\end{corollary}
\begin{proof}
	Firstly, it is easy to see that any instance of $\NGC_{n,k}$ is a planar graph with degree two for any $n,k$ since it is just a collection of cycles and paths. 
	
	For any \emph{odd} $k$, in any instance of $\NGC_{n,k}$, each path of length $k-1$ has a matching of $(k-1)/2$ size. If the instance contains cycles of length $2k$, each such cycle has a matching with $k$ edges, giving us a maximum matching of size $(n/4k) \cdot k + (n/2k) \cdot (k-1)/2$. If the instance instead has cycles of length $k$, then each such cycle has a matching of $(k-1)/2$ edges, and the maximum matching size is $(n/2k) \cdot (k-1)/2 + (n/2k) \cdot (k-1)/2$. There is a difference of $(n/4k)$ in the maximum matching size in either case of $\NGC_{n,k}$. 

	The rest of the lower bound follows exactly as in~\Cref{thm:ccest} by combining~\Cref{prop:stream-cc} with our~\Cref{thm:main-ngc}. 
\end{proof}

\begin{remark}\label{rem:applies}
	For $(1+\eps)$-multiplicative approximation instead, we can prove a lower bound of $2^{\Omega(1/\eps)}$ on space even for constant choices of $\eps > 0$, by simply using $\NGC_{2^{\Theta(1/\eps)},\Theta(1/\eps)}$ and pad 
	the remaining graph with singleton vertices. 
\end{remark}

As a corollary of the standard equivalence between estimating matching size and computing
the rank of the Tutte matrix~\cite{Tutte47} with random entries established by~\cite{Lovasz79} (see \cite{BuryS15} streaming implementation of the reduction), we get the following result as well.

\begin{corollary}\label{thm:rankest}
	Any single-pass streaming algorithm that for every given $\eps > 0$ processes entries of any $n$-times-$n$ matrix $A$ in a random-order stream 
	in $2^{o(1/\eps)} \cdot n^{o(1)}$ space cannot estimate the rank of $A$ to within an  $\eps \cdot n$ additive factor with probability at least $2/3$. 
	
	The lower bound holds even for sparse matrices with two non-zero entries per row and column. 
\end{corollary}

\subsection{Planar Maximum Independent Set} 

We can also prove a similar lower bound for the problem of estimating the maximum independent set size on planar graphs. Given a planar graph $G=(V,E)$ in a random-order stream, 
estimate the size of the largest independent set in $G$ to within an additive $\eps \cdot n$ factor. 

\cite{PengS18} gave an algorithm for this problem that outputs a $(1+\eps)$ multiplicative approximation using $2^{2^{2^{\log^{O(1)}\!{(1/\eps)}}}} \cdot (\log{n})$ space. 
We prove the first lower bound for this problem showing that at least an exponential dependence on $\eps$ is necessary even on an extremely simple family of graphs. 

\begin{corollary}\label{thm:indsetest}
	Any single-pass streaming algorithm that for every given $\eps > 0$ processes edges of any graph $G$ with $n$ vertices in a random-order stream 
	in $2^{o(1/\eps)} \cdot n^{o(1)}$ space cannot estimate the size of maximum independent set in $G$ to within an  $\eps \cdot n$ additive factor with probability at least $2/3$. 
	
	The lower bound holds even for planar graphs of maximum degree two, namely, on disjoint-unions of cycles and paths. 
\end{corollary}
\begin{proof}
	For any odd $k$, in an instance of $\NGC_{n,k}$, each path of length $k-1$ has an independent set of size at most $(k+1)/2$ vertices. If the instance contains cycles of length $2k$, each such cycle has an independent set with $k$ vertices. The size of the maximum independent set is $(n/4k) \cdot k + (n/2k) \cdot \paren{(k+1)/2}$. If the instance instead has cycles of length $k$, then each such cycle has an independent set of $(k-1)/2$ vertices, and the maximum independent set is of size $(n/2k) \cdot \paren{(k+1)/2} + (n/2k) \cdot (k-1)/2$. There is a difference of $(n/4k)$ in either case of $\NGC_{n,k}$. 
	
	The rest of the lower bound follows exactly as in~\Cref{thm:ccest} by combining~\Cref{prop:stream-cc} with our~\Cref{thm:main-ngc}.  
\end{proof}

We note that~\Cref{rem:applies} applies to this problem as well. 

\subsection{Random Walk Generation}\label{sec:rw}

Next, we show that we also get an exponential lower bound for random walk generation in random order streams. To define this problem, we need some further notation. 

\begin{definition}[Pointwise $\eps$-closeness of distributions]\label{def:randwalkcloseness}
	We say that a distribution $p$ over support $\Omega$ is $\eps$-close to distribution $q$ over the same support $\Omega$, if for each $s \in \Omega$, $p(s) \in [1-\eps, 1+\eps] \cdot q(s)$.  
\end{definition}

\begin{definition}[$\eps$-approximate sample]\label{def:randwalkeps}
	Given a graph $G = (V,E)$ and a vertex $u \in V$, we say that $(X_0, X_1, \ldots, X_{\ell})$ is an $\eps$-approximate sample of the $\ell$ step random walk starting at $u$ if the distribution of $(X_0, \ldots, X_{\ell})$ is $\eps$-close pointwise (from \Cref{def:randwalkcloseness}) to the distribution of the $\ell$-step random walk starting at vertex $u$. 
\end{definition}

\begin{problem}[$\bm{(\ell,b,\eps,\delta)}$\textbf{-Random Walk}]\label{prob:randomwalk}
	Given a graph $G = (V, E)$ in a random order stream, a length $\ell$, a budget $b$, and error parameters $\eps, \delta$, generate $b$ independent $\eps$-approximate samples of the $\ell$-step random walk started at a vertex picked uniformly at random from $V$ with error bounded by $\delta$ in total variation distance.
\end{problem}

\cite{KallaugherKP22} gives an upper bound of $(1/\eps)^{O(\ell)} \cdot 2^{O(\ell^2)} \cdot b$ space for this problem for $\delta = 1/10$. For constant $\eps = \delta = 1/10$, \cite{ChipKKP22} gives a lower bound of $\ell^{\Omega(\sqrt{\ell})}$ when $b= 1$, and a lower bound of $\ell^{\Omega(\ell)}$ when $b = \Omega(4^{\ell})$ in the batch random order model. We give exponential in $\ell$ lower bounds in the random stream model for generating random walks.  

\begin{corollary}\label{thm:randwalklb}
	Any algorithm for generating $(\ell, c_1, 1/10, 1/10)$-random walks requires $2^{\Omega(\sqrt{\ell})}$ space and  generating $(\ell, c_2\cdot 4^{\ell}, 1/10, 1/10)$-random walks requires $2^{\Omega(\ell)}$ space for sufficiently large constants $c_1, c_2$, 
	and $\ell = \Theta(\log{n})$.  
\end{corollary}
To prove this result, we need the  following classical result on the cover time of random walks. 
\begin{fact}[cf.~{\cite[Chapter 6]{MotwaniR95}}]\label{fac:covtime}
	Any random walk in a connected graph $G=(V,E)$ visits all the vertices in $G$ in expected time $O(\card{V} \cdot \card{E})$. 
\end{fact}

\begin{proof}[Proof of~\Cref{thm:randwalklb}]
	For the first lower bound, given an instance $G$ of $\NGC_{n,k}$, we pick $\ell =  c_3 k^2$. We generate $c_1$ independent $1/10$-approximate random walks of length $\ell$ with $1/10$ error in total variation distance. 
	We know that if $c_1$ vertices are picked uniformly at random, at least one vertex is a part of a cycle in $G$ with probability at least $1-\frac1{2^{c_1}}$. 	This is because only half of the vertices in $G$ are present in vertex disjoint paths of length $k-1$.
	
	For sufficiently large constant $c_3$, any random walk beginning at $v \in V$ of length at least $c_3 k^2$ visits all the vertices in the connected component of $G$ containing $v$ by \Cref{fac:covtime} with constant probability. Since we generate $1/10$-approximate random walks, the probability that any such walk visits all the vertices reduces by at most a $9/10$ multiplicative factor. The $1/10$ error in total variation distance reduces this probability by at most a $1/10$ additive factor. 
	
 If $v$ is a part of a cycle in $G$, Alice and Bob can use an algorithm for generating random walks to find if $v$ is part of a $k$-cycle or a $2k$-length cycle with probability at least 2/3. As $k =\sqrt{\ell/c_3} =  \Theta(\ln n)$, we require $\Omega(n) = 2^{\Omega(\sqrt{\ell})}$ space to do so by \Cref{thm:main-ngc}. 
	
	For the next lower bound, we pick $\ell =  2k$ and generate $c_2 \cdot 4^{\ell}$ $1/10$-approximate random walks with error at most $1/10$. As we are finding $c_2 4^{\ell}$ independent $1/10$-approximate random walks, with high probability, at least a constant fraction of these random walks begin with cycles in $G$.
	
	The probability that a truly random walk of length at least $2k$ in a cycle with at most $2k$ vertices visits all the vertices in the cycle is at least $2^{-2k}$ (the random walk has to always pick the unvisited edge which happens with probability $1/2$ for $2k $ steps).
	For a $1/10$-approximate random walk, this happens with probability at least $9/10 \cdot 2^{-2k} = 9/10 \cdot 2^{-\ell}$. 
	
	 For a large constant $c_3$, at least one of the $1/10$-approximate random walks visits all the vertices in the cycle with constant probability as the error due to the total variation distance is at most $1/10$.   
	 A cycle of $G$ is found with probability at least 2/3 and Alice and Bob can use this algorithm to find whether $G$ has cycles of length $k$ or $2k$. For $\ell = 2 k = \Theta(\ln n)$, we know that we require $\Omega(n) = 2^{\Omega(\ell)}$ space, again by \Cref{thm:main-ngc}. 
\end{proof}

\subsection{Lower Bounds for Stochastic Streams}\label{sec:stochastic} 

In the stochastic stream model, the algorithm receives $c \cdot m$ edges picked uniformly at random and independently from the graph for some $c > 0$, where $m$ denotes the number of edges in $G$. The goal as before is to estimate a certain parameter of the underlying graph 
with minimal space, and now additionally, minimal number of samples also, i.e., by minimizing the parameter $c$. To our knowledge, this model was first studied by~\cite{CrouchMVW16}. 

We show that our lower bounds continue to hold on stochastic streams of $c \cdot m$ edges for any constant $c > 0$. 
We will begin by defining the $\NGCst_{n,k}$ communication problem, an adaptation of $\NGC_{n,k}$ from random streams for a stochastic stream of $cm $ edges.  

\begin{problem}[\textbf{Stochastic Noisy Gap Cycle Counting (\NGCst)}]\label{prob:ngcstoch}
	For any integers $n,k \geq 1$, In $\NGCst_{n,k}$, we have a graph $G = (V,E)$ on $n $ vertices such that $G$ either contains: $(i)$  $(n/4k)$ vertex-disjoint cycles of length $2k$, or $(ii)$ $(n/2k)$ vertex-disjoint cycles of length $k$.
	In addition, in both cases, $G$ contains $(n/2k)$ vertex-disjoint paths of length $k-1$. 
	
	We pick $c \cdot \card{E}/2$ edges uniformly and independently from $E$ (with repetition) and add them to $E_A$, and similarly $c \cdot \card{E}/2$ edges are picked at random and added to $E_B$ (here, $c$, is the same constant as the one defined above for stochastic streams). Alice is given $E_A$ and Bob is given $E_B$. Their goal is to decide which 
	case the graph belongs to by Alice sending a single message to Bob. 
\end{problem}

Lower bounds on the communication of $\NGCst_{n,k}$ translate into lower bounds for stochastic streams, as follows.
\begin{proposition}\label{obs:stream-stoch}
	Any single-pass streaming algorithm that given a graph $G = (V, E)$ from $\mathcal{G}_n$ with $n/2k$ paths of length $k-1$, and either
	\begin{enumerate}[label=$(\roman*)$]
		\item $n/2k$ cycles of length $k$, or,
		\item $n/4k$ cycles of length $2k$,
	\end{enumerate} 
in a stochastic stream with $c \cdot \card{E}$ edges, computes which case the graph belongs to w.p. at least $1-\delta$ requires at least as much space as the communication complexity of $\NGCst_{n,k}$ with error at most $\delta$. 
\end{proposition}
\begin{proof}
	Alice and Bob can use the streaming algorithm to solve $\NGCst_{n,k}$: 
	Alice randomly permutes the edges of $G_A$, runs the algorithm over them, and sends the memory content as the message $\pi$ to Bob; 
	Bob takes a random permutation of $G_B$, continues running the algorithm, and outputs the answer. 
	This simulates a stochastic stream of $c\card{E}$ samples as each edge is sampled independently and given to Alice and Bob. The algorithm solves $\NGCst_{n,k}$ with error at most $\delta$. 
\end{proof}

By building on the approach in~\Cref{thm:main-ngc}, we further show that $\NGCst_{n,k}$ requires $\Omega(n)$ bits of communication using the same hard distribution $\distNGC$ for the underlying graphs. 

\begin{lemma}\label{lem:ngcstoch}
	For sufficiently large $n\geq 1$ and constant $c = O(1)$ with $k \geq 24 \cdot e^{9c}  \cdot \ln n$, 
	any one-way communication protocol for $\NGCst_{n,k}$ with probability of success at least $2/3$ requires $\Omega(n)$ communication. 
\end{lemma}
The proof of \Cref{lem:ngcstoch} can be found in \Cref{app:ngcstoch}. 

All our lower bounds in for random order streams (\Cref{thm:ccest,thm:mstest,thm:matchingest,thm:rankest,thm:indsetest,thm:randwalklb}) apply for stochastic streams as well, as a direct consequence of \Cref{lem:ngcstoch} and \Cref{obs:stream-stoch}. In particular,

\begin{corollary}\label{thm:stochastic}
	Let $c > 0$ be a fixed constant. Any streaming algorithm that for every given $\eps > 0$ and graph $G = (V,E)$ processes $c \cdot \card{E}$ uniform samples of edges from $E$ in a stochastic stream and uses
	$2^{o(1/\eps)} \cdot n^{o(1)}$ space cannot solve any of the following problems with probability of success at least $2/3$: 
	\begin{enumerate}[label=$(\roman*)$]
		\item estimating the number of connected components in $G$ to within an  $\eps \cdot n$ additive factor; 
		\item estimating the weight of minimum spanning trees to within an $(1+\eps)$ multiplicative approximation factor (even when the weights are integers and constant); 
		\item estimating the size of maximum matching in $G$ to within an  $\eps \cdot n$ additive factor (even on bounded-degree planar $G$); 
		\item estimating the size of maximum independent set in $G$ to within an  $\eps \cdot n$ additive factor (even on bounded-degree planar  $G$).
	\end{enumerate}
\end{corollary}

This concludes our list of implications of~\Cref{thm:main-ngc} to random order (and stochastic) streams.

\bigskip

\section*{Acknowledgement} 

We thank Michael Kapralov and Huacheng Yu for helpful discussions. 

\bigskip

\bibliographystyle{alpha}
\bibliography{ref}

\newcommand{\etalchar}[1]{$^{#1}$}
\begin{thebibliography}{CMVW16}

\bibitem[ACL{\etalchar{+}}22]{AssadiCLMW22}
Sepehr Assadi, Vaggos Chatziafratis, Jakub Lacki, Vahab Mirrokni, and Chen
  Wang.
\newblock Hierarchical clustering in graph streams: Single-pass algorithms and
  space lower bounds.
\newblock In Po{-}Ling Loh and Maxim Raginsky, editors, {\em Conference on
  Learning Theory, 2-5 July 2022, London, {UK}}, volume 178 of {\em Proceedings
  of Machine Learning Research}, pages 4643--4702. {PMLR}, 2022.

\bibitem[AKSY20]{AssadiKSY20}
Sepehr Assadi, Gillat Kol, Raghuvansh~R. Saxena, and Huacheng Yu.
\newblock Multi-pass graph streaming lower bounds for cycle counting, max-cut,
  matching size, and other problems.
\newblock In Sandy Irani, editor, {\em 61st {IEEE} Annual Symposium on
  Foundations of Computer Science, {FOCS} 2020, Durham, NC, USA, November
  16-19, 2020}, pages 354--364. {IEEE}, 2020.

\bibitem[AMS96]{AlonMS96}
Noga Alon, Yossi Matias, and Mario Szegedy.
\newblock The space complexity of approximating the frequency moments.
\newblock In Gary~L. Miller, editor, {\em Proceedings of the Twenty-Eighth
  Annual {ACM} Symposium on the Theory of Computing, Philadelphia,
  Pennsylvania, USA, May 22-24, 1996}, pages 20--29. {ACM}, 1996.

\bibitem[AN21]{AssadiN21}
Sepehr Assadi and Vishvajeet N.
\newblock Graph streaming lower bounds for parameter estimation and property
  testing via a streaming {XOR} lemma.
\newblock In Samir Khuller and Virginia~Vassilevska Williams, editors, {\em
  {STOC} '21: 53rd Annual {ACM} {SIGACT} Symposium on Theory of Computing,
  Virtual Event, Italy, June 21-25, 2021}, pages 612--625. {ACM}, 2021.

\bibitem[BCK{\etalchar{+}}18]{BravermanCKLWY18}
Vladimir Braverman, Stephen~R. Chestnut, Robert Krauthgamer, Yi~Li, David~P.
  Woodruff, and Lin~F. Yang.
\newblock Matrix norms in data streams: Faster, multi-pass and row-order.
\newblock In Jennifer~G. Dy and Andreas Krause, editors, {\em Proceedings of
  the 35th International Conference on Machine Learning, {ICML} 2018,
  Stockholmsm{\"{a}}ssan, Stockholm, Sweden, July 10-15, 2018}, volume~80 of
  {\em Proceedings of Machine Learning Research}, pages 648--657. {PMLR}, 2018.

\bibitem[BJK04]{BarJK04}
Ziv Bar{-}Yossef, T.~S. Jayram, and Iordanis Kerenidis.
\newblock Exponential separation of quantum and classical one-way communication
  complexity.
\newblock In L{\'{a}}szl{\'{o}} Babai, editor, {\em Proceedings of the 36th
  Annual {ACM} Symposium on Theory of Computing, Chicago, IL, USA, June 13-16,
  2004}, pages 128--137. {ACM}, 2004.

\bibitem[BS15]{BuryS15}
Marc Bury and Chris Schwiegelshohn.
\newblock Sublinear estimation of weighted matchings in dynamic data streams.
\newblock In Nikhil Bansal and Irene Finocchi, editors, {\em Algorithms - {ESA}
  2015 - 23rd Annual European Symposium, Patras, Greece, September 14-16, 2015,
  Proceedings}, volume 9294 of {\em Lecture Notes in Computer Science}, pages
  263--274. Springer, 2015.

\bibitem[CCM08]{ChakrabartiCM08}
Amit Chakrabarti, Graham Cormode, and Andrew McGregor.
\newblock Robust lower bounds for communication and stream computation.
\newblock In {\em Proceedings of the 40th Annual {ACM} Symposium on Theory of
  Computing, Victoria, British Columbia, Canada, May 17-20, 2008}, pages
  641--650, 2008.

\bibitem[CFPS20]{CzumajFPS20}
Artur Czumaj, Hendrik Fichtenberger, Pan Peng, and Christian Sohler.
\newblock Testable properties in general graphs and random order streaming.
\newblock In Jaroslaw Byrka and Raghu Meka, editors, {\em Approximation,
  Randomization, and Combinatorial Optimization. Algorithms and Techniques,
  {APPROX/RANDOM} 2020, August 17-19, 2020, Virtual Conference}, volume 176 of
  {\em LIPIcs}, pages 16:1--16:20. Schloss Dagstuhl - Leibniz-Zentrum f{\"{u}}r
  Informatik, 2020.

\bibitem[CJMM17]{CormodeJMM17}
Graham Cormode, Hossein Jowhari, Morteza Monemizadeh, and S.~Muthukrishnan.
\newblock The sparse awakens: Streaming algorithms for matching size estimation
  in sparse graphs.
\newblock In Kirk Pruhs and Christian Sohler, editors, {\em 25th Annual
  European Symposium on Algorithms, {ESA} 2017, September 4-6, 2017, Vienna,
  Austria}, volume~87 of {\em LIPIcs}, pages 29:1--29:15. Schloss Dagstuhl -
  Leibniz-Zentrum f{\"{u}}r Informatik, 2017.

\bibitem[CKKP22]{ChipKKP22}
Ashish Chiplunkar, John Kallaugher, Michael Kapralov, and Eric Price.
\newblock Factorial lower bounds for (almost) random order streams.
\newblock In {\em 63rd {IEEE} Annual Symposium on Foundations of Computer
  Science, {FOCS} 2022}. {IEEE}, 2022.

\bibitem[CKP{\etalchar{+}}23]{ChenKPSSY23}
Lijie Chen, Gillat Kol, Dmitry Paramonov, Raghuvansh Saxena, Zhao Song, and
  Huacheng Yu.
\newblock Towards multi-pass streaming lower bounds for optimal approximation
  of max-cut.
\newblock In {\em Proceedings of the Thirty-Fourth Annual {ACM-SIAM} Symposium
  on Discrete Algorithms, {SODA} 2023}. {SIAM}, 2023.

\bibitem[CMVW16]{CrouchMVW16}
Michael~S. Crouch, Andrew McGregor, Gregory Valiant, and David~P. Woodruff.
\newblock Stochastic streams: Sample complexity vs. space complexity.
\newblock In Piotr Sankowski and Christos~D. Zaroliagis, editors, {\em 24th
  Annual European Symposium on Algorithms, {ESA} 2016, August 22-24, 2016,
  Aarhus, Denmark}, volume~57 of {\em LIPIcs}, pages 32:1--32:15. Schloss
  Dagstuhl - Leibniz-Zentrum f{\"{u}}r Informatik, 2016.

\bibitem[CR11]{ChakrabartiR11}
Amit Chakrabarti and Oded Regev.
\newblock An optimal lower bound on the communication complexity of
  gap-hamming-distance.
\newblock In Lance Fortnow and Salil~P. Vadhan, editors, {\em Proceedings of
  the 43rd {ACM} Symposium on Theory of Computing, {STOC} 2011, San Jose, CA,
  USA, 6-8 June 2011}, pages 51--60. {ACM}, 2011.

\bibitem[CRT01]{ChazelleRT01}
Bernard Chazelle, Ronitt Rubinfeld, and Luca Trevisan.
\newblock Approximating the minimum spanning tree weight in sublinear time.
\newblock In Fernando Orejas, Paul~G. Spirakis, and Jan van Leeuwen, editors,
  {\em Automata, Languages and Programming, 28th International Colloquium,
  {ICALP} 2001, Crete, Greece, July 8-12, 2001, Proceedings}, volume 2076 of
  {\em Lecture Notes in Computer Science}, pages 190--200. Springer, 2001.

\bibitem[DP09]{DubhashiP09}
Devdatt~P. Dubhashi and Alessandro Panconesi.
\newblock {\em Concentration of Measure for the Analysis of Randomized
  Algorithms}.
\newblock Cambridge University Press, 2009.

\bibitem[EHL{\etalchar{+}}15]{EsfandiariHLMO15}
Hossein Esfandiari, Mohammad~Taghi Hajiaghayi, Vahid Liaghat, Morteza
  Monemizadeh, and Krzysztof Onak.
\newblock Streaming algorithms for estimating the matching size in planar
  graphs and beyond.
\newblock In Piotr Indyk, editor, {\em Proceedings of the Twenty-Sixth Annual
  {ACM-SIAM} Symposium on Discrete Algorithms, {SODA} 2015, San Diego, CA, USA,
  January 4-6, 2015}, pages 1217--1233. {SIAM}, 2015.

\bibitem[FKM{\etalchar{+}}05]{FeigenbaumKMSZ05}
Joan Feigenbaum, Sampath Kannan, Andrew McGregor, Siddharth Suri, and Jian
  Zhang.
\newblock On graph problems in a semi-streaming model.
\newblock {\em Theor. Comput. Sci.}, 348(2-3):207--216, 2005.

\bibitem[FKM{\etalchar{+}}08]{FeigenbaumKMSZ08}
Joan Feigenbaum, Sampath Kannan, Andrew McGregor, Siddharth Suri, and Jian
  Zhang.
\newblock Graph distances in the data-stream model.
\newblock {\em {SIAM} J. Comput.}, 38(5):1709--1727, 2008.

\bibitem[GKK{\etalchar{+}}07]{GavinskyKKRW07}
Dmitry Gavinsky, Julia Kempe, Iordanis Kerenidis, Ran Raz, and Ronald de~Wolf.
\newblock Exponential separations for one-way quantum communication complexity,
  with applications to cryptography.
\newblock In David~S. Johnson and Uriel Feige, editors, {\em Proceedings of the
  39th Annual {ACM} Symposium on Theory of Computing, San Diego, California,
  USA, June 11-13, 2007}, pages 516--525. {ACM}, 2007.

\bibitem[GT19]{GuruswamiT19}
Venkatesan Guruswami and Runzhou Tao.
\newblock Streaming hardness of unique games.
\newblock In Dimitris Achlioptas and L{\'{a}}szl{\'{o}}~A. V{\'{e}}gh, editors,
  {\em Approximation, Randomization, and Combinatorial Optimization. Algorithms
  and Techniques, {APPROX/RANDOM} 2019, September 20-22, 2019, Massachusetts
  Institute of Technology, Cambridge, MA, {USA}}, volume 145 of {\em LIPIcs},
  pages 5:1--5:12. Schloss Dagstuhl - Leibniz-Zentrum f{\"{u}}r Informatik,
  2019.

\bibitem[GVV17]{GuruswamiVV17}
Venkatesan Guruswami, Ameya Velingker, and Santhoshini Velusamy.
\newblock Streaming complexity of approximating max 2csp and max acyclic
  subgraph.
\newblock In Klaus Jansen, Jos{\'{e}} D.~P. Rolim, David Williamson, and
  Santosh~S. Vempala, editors, {\em Approximation, Randomization, and
  Combinatorial Optimization. Algorithms and Techniques, {APPROX/RANDOM} 2017,
  August 16-18, 2017, Berkeley, CA, {USA}}, volume~81 of {\em LIPIcs}, pages
  8:1--8:19. Schloss Dagstuhl - Leibniz-Zentrum f{\"{u}}r Informatik, 2017.

\bibitem[HP16]{HuangP16}
Zengfeng Huang and Pan Peng.
\newblock Dynamic graph stream algorithms in o(n) space.
\newblock In Ioannis Chatzigiannakis, Michael Mitzenmacher, Yuval Rabani, and
  Davide Sangiorgi, editors, {\em 43rd International Colloquium on Automata,
  Languages, and Programming, {ICALP} 2016, July 11-15, 2016, Rome, Italy},
  volume~55 of {\em LIPIcs}, pages 18:1--18:16. Schloss Dagstuhl -
  Leibniz-Zentrum f{\"{u}}r Informatik, 2016.

\bibitem[IW03]{IndykW03}
Piotr Indyk and David~P. Woodruff.
\newblock Tight lower bounds for the distinct elements problem.
\newblock In {\em 44th Symposium on Foundations of Computer Science {(FOCS}
  2003), 11-14 October 2003, Cambridge, MA, USA, Proceedings}, pages 283--288.
  {IEEE} Computer Society, 2003.

\bibitem[KK15]{KoganK15}
Dmitry Kogan and Robert Krauthgamer.
\newblock Sketching cuts in graphs and hypergraphs.
\newblock In Tim Roughgarden, editor, {\em Proceedings of the 2015 Conference
  on Innovations in Theoretical Computer Science, {ITCS} 2015, Rehovot, Israel,
  January 11-13, 2015}, pages 367--376. {ACM}, 2015.

\bibitem[KKL88]{KahnKL88}
Jeff Kahn, Gil Kalai, and Nathan Linial.
\newblock The influence of variables on boolean functions (extended abstract).
\newblock In {\em 29th Annual Symposium on Foundations of Computer Science,
  White Plains, New York, USA, 24-26 October 1988}, pages 68--80. {IEEE}
  Computer Society, 1988.

\bibitem[KKP22]{KallaugherKP22}
John Kallaugher, Michael Kapralov, and Eric Price.
\newblock Simulating random walks in random streams.
\newblock In Joseph~(Seffi) Naor and Niv Buchbinder, editors, {\em Proceedings
  of the 2022 {ACM-SIAM} Symposium on Discrete Algorithms, {SODA} 2022, Virtual
  Conference / Alexandria, VA, USA, January 9 - 12, 2022}, pages 3091--3126.
  {SIAM}, 2022.

\bibitem[KKS14]{KapralovKS14}
Michael Kapralov, Sanjeev Khanna, and Madhu Sudan.
\newblock Approximating matching size from random streams.
\newblock In Chandra Chekuri, editor, {\em Proceedings of the Twenty-Fifth
  Annual {ACM-SIAM} Symposium on Discrete Algorithms, {SODA} 2014, Portland,
  Oregon, USA, January 5-7, 2014}, pages 734--751. {SIAM}, 2014.

\bibitem[KKS15]{KapralovKS15}
Michael Kapralov, Sanjeev Khanna, and Madhu Sudan.
\newblock Streaming lower bounds for approximating {MAX-CUT}.
\newblock In Piotr Indyk, editor, {\em Proceedings of the Twenty-Sixth Annual
  {ACM-SIAM} Symposium on Discrete Algorithms, {SODA} 2015, San Diego, CA, USA,
  January 4-6, 2015}, pages 1263--1282. {SIAM}, 2015.

\bibitem[KMNT20]{KapralovMNT20}
Michael Kapralov, Slobodan Mitrovic, Ashkan Norouzi{-}Fard, and Jakab Tardos.
\newblock Space efficient approximation to maximum matching size from uniform
  edge samples.
\newblock In Shuchi Chawla, editor, {\em Proceedings of the 2020 {ACM-SIAM}
  Symposium on Discrete Algorithms, {SODA} 2020, Salt Lake City, UT, USA,
  January 5-8, 2020}, pages 1753--1772. {SIAM}, 2020.

\bibitem[KMT{\etalchar{+}}22]{KapralovMTWZ22}
Michael Kapralov, Amulya Musipatla, Jakab Tardos, David~P. Woodruff, and Samson
  Zhou.
\newblock Noisy boolean hidden matching with applications.
\newblock In Mark Braverman, editor, {\em 13th Innovations in Theoretical
  Computer Science Conference, {ITCS} 2022, January 31 - February 3, 2022,
  Berkeley, CA, {USA}}, volume 215 of {\em LIPIcs}, pages 91:1--91:19. Schloss
  Dagstuhl - Leibniz-Zentrum f{\"{u}}r Informatik, 2022.

\bibitem[Lov79]{Lovasz79}
L{\'{a}}szl{\'{o}} Lov{\'{a}}sz.
\newblock On determinants, matchings, and random algorithms.
\newblock In Lothar Budach, editor, {\em Fundamentals of Computation Theory,
  {FCT} 1979, Proceedings of the Conference on Algebraic, Arthmetic, and
  Categorial Methods in Computation Theory, Berlin/Wendisch-Rietz, Germany,
  September 17-21, 1979}, pages 565--574. Akademie-Verlag, Berlin, 1979.

\bibitem[LW16]{LiW16}
Yi~Li and David~P. Woodruff.
\newblock On approximating functions of the singular values in a stream.
\newblock In Daniel Wichs and Yishay Mansour, editors, {\em Proceedings of the
  48th Annual {ACM} {SIGACT} Symposium on Theory of Computing, {STOC} 2016,
  Cambridge, MA, USA, June 18-21, 2016}, pages 726--739. {ACM}, 2016.

\bibitem[McG14]{McGregor14}
Andrew McGregor.
\newblock Graph stream algorithms: a survey.
\newblock {\em {SIGMOD} Rec.}, 43(1):9--20, 2014.

\bibitem[MMPS17]{MonemizadehMPS17}
Morteza Monemizadeh, S.~Muthukrishnan, Pan Peng, and Christian Sohler.
\newblock Testable bounded degree graph properties are random order streamable.
\newblock In Ioannis Chatzigiannakis, Piotr Indyk, Fabian Kuhn, and Anca
  Muscholl, editors, {\em 44th International Colloquium on Automata, Languages,
  and Programming, {ICALP} 2017, July 10-14, 2017, Warsaw, Poland}, volume~80
  of {\em LIPIcs}, pages 131:1--131:14. Schloss Dagstuhl - Leibniz-Zentrum
  f{\"{u}}r Informatik, 2017.

\bibitem[MR95]{MotwaniR95}
Rajeev Motwani and Prabhakar Raghavan.
\newblock {\em Randomized Algorithms}.
\newblock Cambridge University Press, 1995.

\bibitem[Mut05]{Muthukrishnan05}
S.~Muthukrishnan.
\newblock Data streams: Algorithms and applications.
\newblock {\em Found. Trends Theor. Comput. Sci.}, 1(2), 2005.

\bibitem[MV16]{McGregorV16}
Andrew McGregor and Sofya Vorotnikova.
\newblock Planar matching in streams revisited.
\newblock In Klaus Jansen, Claire Mathieu, Jos{\'{e}} D.~P. Rolim, and Chris
  Umans, editors, {\em Approximation, Randomization, and Combinatorial
  Optimization. Algorithms and Techniques, {APPROX/RANDOM} 2016, September 7-9,
  2016, Paris, France}, volume~60 of {\em LIPIcs}, pages 17:1--17:12. Schloss
  Dagstuhl - Leibniz-Zentrum f{\"{u}}r Informatik, 2016.

\bibitem[MV18]{McGregorV18}
Andrew McGregor and Sofya Vorotnikova.
\newblock A simple, space-efficient, streaming algorithm for matchings in low
  arboricity graphs.
\newblock In Raimund Seidel, editor, {\em 1st Symposium on Simplicity in
  Algorithms, {SOSA} 2018, January 7-10, 2018, New Orleans, LA, {USA}},
  volume~61 of {\em OASIcs}, pages 14:1--14:4. Schloss Dagstuhl -
  Leibniz-Zentrum f{\"{u}}r Informatik, 2018.

\bibitem[PS18]{PengS18}
Pan Peng and Christian Sohler.
\newblock Estimating graph parameters from random order streams.
\newblock In Artur Czumaj, editor, {\em Proceedings of the Twenty-Ninth Annual
  {ACM-SIAM} Symposium on Discrete Algorithms, {SODA} 2018, New Orleans, LA,
  USA, January 7-10, 2018}, pages 2449--2466. {SIAM}, 2018.

\bibitem[SSSV23]{SaxenaSSV23}
Raghuvansh~R. Saxena, Noah Singer, Madhu Sudan, and Santhoshini Velusamy.
\newblock Streaming complexity of csps with randomly ordered constraints.
\newblock In {\em Proceedings of the Thirty-Fourth Annual {ACM-SIAM} Symposium
  on Discrete Algorithms, {SODA} 2023}. {SIAM}, 2023.

\bibitem[Tut47]{Tutte47}
William~T Tutte.
\newblock The factorization of linear graphs.
\newblock {\em Journal of the London Mathematical Society}, 1(2):107--111,
  1947.

\bibitem[VY11]{VerbinY11}
Elad Verbin and Wei Yu.
\newblock The streaming complexity of cycle counting, sorting by reversals, and
  other problems.
\newblock In Dana Randall, editor, {\em Proceedings of the Twenty-Second Annual
  {ACM-SIAM} Symposium on Discrete Algorithms, {SODA} 2011, San Francisco,
  California, USA, January 23-25, 2011}, pages 11--25. {SIAM}, 2011.

\bibitem[Wol08]{Wol08}
Ronald Wolf.
\newblock A brief introduction to fourier analysis on the boolean cube.
\newblock {\em Theory of Computing, Graduate Surveys}, 1:1--20, 01 2008.

\bibitem[Yao77]{Yao77}
Andrew Chi-Chin Yao.
\newblock Probabilistic computations: Toward a unified measure of complexity.
\newblock In {\em 18th Annual Symposium on Foundations of Computer Science
  (sfcs 1977)}, pages 222--227, 1977.

\end{thebibliography}

\clearpage
\appendix

\section{Noisy Gap Cycle Counting in Stochastic Streams}\label{app:ngcstoch}

This section is dedicated to proving \Cref{lem:ngcstoch}. First, let us define the $\DHXORst_{\ww, t}$ problem, similar to $\DHXOR_{\ww, t}$ except for how the edges in $E_A, E_B$ are distributed for a stochastic stream of $cm$ edges.

\begin{problem}[\textbf{Stochastic Hidden-XOR Problem} (\DHXORst)]\label{def:dhxorstoch}
	In $\DHXORst_{\ww,t}$, we have a graph $G = \multiblock{X,\Sigma} \in \LG_{\ww,3t+1}$ for $X \in (\set{0,1}^{\ww})^{t}$ and $\Sigma \in (\mathcal{S}_{\ww})^t$ chosen independently and uniformly at random. The goal is to output $z_G(1)$ in this graph.
	
	We pick $c \cdot \card{E}/2$ edges uniformly and independently from $E$ (with repetition) and add them to $E_A$, and similarly $c \cdot \card{E}/2$ edges are picked at random and added to $E_B$ (here, $c$, is the same constant as the one defined above for stochastic streams). Alice is given $E_A$ and Bob is given $E_B$. Their goal is to decide which 
	case the graph belongs to by Alice sending a single message to Bob. \end{problem}

We can prove that $\DHXORst_{\ww, t}$ reduces to $\NGCst_{n,k}$ in stochastic streams as well. 

\begin{lemma}\label{lem:ngctodhxorstoch}
	For any sufficiently large $n \geq 1$ and $k = 3t+1$ for some $t \geq 1$ and $\ww = (n/4k)+1$, the distributional communication complexity of $\DHXORst_{\ww, t}$ for probability of success at least $\frac12 + \frac1{6\ww}$ is at most as much as the determinisitic communication complexity of $\NGCst_{n,k}$ over distribution $\distNGC$ for probability of success at least $2/3$.
\end{lemma}
\begin{proof}
	Alice and Bob use the first $(vi)$ steps of \Cref{alg:samplngch} to get an instance of $\NGCst_{n,k}$ given an instance $G = (V, E)$ of $\DHXORst_{\ww, t}$. \Cref{lem:hybridindex,lem:hybridsampling} still hold for $\NGCst_{n,k}$ and $\DHXORst_{\ww, t}$ as well because they are only about distribution $\distHNGC$ and the sampling process in \Cref{alg:samplngch}. Given a streaming algorithm for $\NGCst_{n,k}$, Alice and Bob have an instance $G'  = (V', E')$ of $\NGCst_{n,k}$ where the edges added to $G'$ which are not in $G$ are public. $c\card{E}/2$ edges in $G$ are given privately to both Alice and Bob, with each picked independently and randomly. They can each generate a stochastic stream of $c \card{E'}/2$ edges in $G'$ with public randomness and the edges given to them from $G$ without any additional communication. By \Cref{clm:istaris1}, we know that they can output $z_G(1)$ with probability at least $\frac12 + \frac1{6\ww}$. 
\end{proof}

If we prove a lower bound on the distributional communication complexity of $\DHXORst_{\ww, t}$ against probability of success $\frac12 + \frac1{6\ww}$ of $\Omega(n)$, we are done. We will show that all of the results in \Cref{sec:hxorlower} can be applied to stochastic streams as well. 

We will start by updating our definitions of clean indices and active blocks for stochastic streams. We will then show that these clean indices and active blocks appear with high probability in $\DHXORst_{\ww,t}$ problem. 

\begin{definition}\label{def:cleanindstoch}
	We say that an index $j \in [\ww]$ in some block $B = \block{x, \sigma}$ is \textbf{clean} iff the following holds for vertices of groups $g^2_j = (a^2_j,b^2_j)$ and $g^3_j = (a^3_j,b^3_j)$ in layers $V^2$ and $V^3$ of $B$.
	\begin{itemize} 
		\item Edges $(a^1_{\sigma^{-1}(j)}, a^2_j), (b^1_{\sigma^{-1}(j)}, b^2_j), (a^3_j, a^4_{\sigma^{-1}(j)}), (b^3_j, b^4_{\sigma^{-1}(j)})$ do not belong to $E_A$, and belong to $E_B$.  
		\item Edges from group $g^2_j$ to group $g^3_j$ do not belong to $E_B$, and belong to $E_A$. 
	\end{itemize}
We let $\clean{B}$ denote the set of clean indices in $B$.
\end{definition}

To analyze the number of clean indices, we will first lower bound the probability that an index is clean. 
\begin{claim}\label{clm:ind-clean-stoch}
	For any index $j \in [\ww]$, the probability that $j $ is a clean index is at least $\paren{\frac1{e^{9c}}}$. 
\end{claim}

\begin{proof}
		For any edge $e $ in block $B$, since $c\card{E}/2$ edges are given to Alice and $c \card{E}/2$ edges are given to Bob, for sufficiently large $\card{E}$,
	\begin{align*}
		\prob{ e \notin E_B} &= \paren{1-\frac1{\card{E}}}^{c \card{E}/2} \geq \frac1{e^{c}}, \\
		\prob{e \notin E_B, e \in E_A} &\geq \frac1{e^c} \cdot \paren{1-\paren{1-\frac1{\card{E}}}^{c\card{E}/2}} \geq \frac1{e^c} \cdot \paren{1-\frac1{e^{c/2}}} \geq \frac1{e^{3c/2}} \\
		\prob{ e \notin E_A, e \in E_B} &\geq \frac1{e^{3c/2}},
	\end{align*}
where we have used that $e^{-2x} \leq 1-x \leq e^{-x}$ for $0 < x < 1/2$. 
For index $j$ to be clean, the 6 edges associated with it must be distributed according to \Cref{def:cleanindstoch}. As the process of distributing edges is independent,
the probability that index $j$ is clean is at least $  \frac1{e^{9c}}$.
\end{proof}

 We can prove that a constant fraction of the indices are clean even with our updated definition.
\begin{lemma}\label{lem:clean-indexstoch}
	For any fixed block $B$, over the randomness of edges given to $E_A$ and $E_B$, 
	\[
	\Pr\Paren{\card{\clean{B}} \geq \ww/2e^{9c}} \geq 1- 1/\ww^4,
	\] 
	for some constant $\cons$. 
\end{lemma}

\begin{proof}
	Let $\indclean_j$ be the indicator function for whether index $j \in [\ww]$ is clean in block $B$. We want to prove that $\indclean = \sum_{j \in [w]} \indclean_j$ is at least $\ww/2e^{9c}$ with high probability. 

	For each $j \in [\ww]$, $\indclean_j = 1$ with probability at least $1/e^{9c}$  and 0 otherwise for each $j \in [\ww]$ independently by \Cref{clm:ind-clean-stoch}. Chernoff bound from \Cref{prop:chernoff} can be applied here. We know $\expect{\indclean} = \sum_{j=1}^{\ww} \expect{\indclean_j} = \ww/e^{9c}$. Thus, for large $\ww$, 
	\begin{align*}
		\prob{\indclean < \frac{\ww}{2e^{9c}}} &= \prob{\indclean < \paren{1-\frac12} \cdot \expect{\indclean}} \\
		&\leq \exp\paren{\frac{-1}{3} \cdot \frac{\ww}{2e^{9c}} \cdot \frac1{4}} \\
		&\leq \frac1{\ww^4}. \qedhere
	\end{align*}
\end{proof}

We retain our definition for active blocks, that is blocks where index $\sigma^i(1)$ is clean are active.
It is easy to see that even with our new notion of clean indices, $\sigma(1)$ is uniform over $\clean{B}$ conditioned on $B$ being active. 

\begin{lemma}\label{lem:active-block-pijstoch}
	Suppose we sample $B = \block{x,\sigma}$ conditioned on $B$ being active. Then, $\sigma(1)$ is chosen uniformly at random from $\clean{B}$.
\end{lemma}

\begin{proof}
	In any block, whether or not an index $j \in [\ww]$ is clean depends only on how the edges in block $B$ are given to Alice and Bob in $E_A$ and $E_B$. It does not depend on choice of permutation $\sigma(j)$. 
	Any block is active with probability $\clean{B}/{\ww}$, since index $\sigma(1)$ is chosen uniformly at random from $[\ww]$. 
	
	Let $\ell$ be any index in $\clean{B}$. Then,
	\begin{align*}
		\prob{\sigma(1) = \ell \mid \textnormal{$B$ is active}} &= \frac{\prob{\sigma(1) = \ell }}{\prob{ \textnormal{$B$ is active}}} \tag{$\ell \in \clean{B}$, $\sigma(1)=\ell$ implies $B$ is active}  \\
		&= \prob{\sigma(1)=\ell} \cdot\frac{\ww}{\clean{B}} \tag{as $B$ is active with probability $\frac{\clean{B}}{\ww}$}  \\
		&= \frac1{\ww} \cdot \frac{\ww}{\clean{B}}  \tag{as $\sigma(1)$ is uniformly random over $[\ww]$} \\
		&= \frac1{\clean{B}}.  \qedhere
	\end{align*}
\end{proof}

 We are able to prove that the number of active blocks is high in any multi-block graph in stochastic streams as well.

\begin{lemma}\label{lem:active-blockstoch}
	For $G = \multiblock{X,\Sigma}$ and $\mathcal{F}$ sampled in $\DHXORst_{w,t}$ with $t \geq 8 \cdot e^{9c}\cdot \ln \ww$: 
	\[
	\Pr\Paren{\card{\act{G}} \geq \ln{\ww}} \geq 1-1/\ww^2. 
	\]
\end{lemma}

\begin{proof}
	Block $B_i$ being active is independent of other blocks for $i \in [t]$, as it depends only on $\sigma^i$ and how the edges in block $B_i$ are distributed in $E_A$ and $E_B$. For any $i \in [t]$, by \Cref{clm:ind-clean-stoch},
	\begin{align*}
		\prob{B_i \textnormal{ is active}} &= \prob{\sigma^i(1) \textnormal{ is clean}}  \geq 1/e^{9c}.
	\end{align*}
	Let $R_i$ be the indicator variable for whether block $B_i$ is active. We want to prove that $R = \sum_{i=1}^t R_i$ is larger than $ \ln w$ with high probability. By linearity of expectation, $\expect{R} \geq t \cdot \frac1{e^{9c}} \geq 8 \ln w$. Chernoff bound from \Cref{prop:chernoff} can be used here as $R_i$s are independent. 
	\begin{align*}
		\prob{R \geq  \ln \ww} &= \prob{R \geq \paren{1-\frac{7}{8}} \cdot \expect{R}} \\
		&\geq 1-\exp\paren{\frac{-49}{64} \cdot \frac{8 \ln \ww}{3}} \\
		&\geq 1- \frac1{\ww^2}. \qedhere
	\end{align*}
\end{proof}

\Cref{clm:cond-act-g,clm:not-clean-fixed} still hold here as they are only related to the probability of success of the protocol and not how the edges are distributed. Hence, after conditioning on the input graph having $\actt \geq \ln \ww$ active blocks,  $\clw = \ww/2{e^{9c}}$ clean indices in each block, $(x^i, \sigma^i) $ for $B_i \notin \act{G}$, and all the edges associated with indices $[\ww] \setminus \clean{B_i}$ for $B_i \in \act{G}$, Alice and Bob are again left with the following scenario.

	\begin{itemize}
		\item Alice has $\actt \geq  \ln \ww$ strings of length $\clw = \ww/2e^{9c}$ each,  $x^i_j$ for $i \in [\actt], j \in \clean{B_i}$ picked uniformly at random and independently. 
		\item Bob has $\actt$ permutations from $\mathcal{S}_{\clw}$ (though the numbering of the indices may differ), $\sigma^i$ restricted from $\clninv{B_i}$ to $\clean{B_i}$ for $i \in [\actt]$ chosen uniformly at random and independently. 
	\end{itemize}  
 \Cref{lem:avgargument} holds here as well for the updated number of clean indices. We can prove the following lemma for stochastic streams.
 \begin{lemma}\label{lem:dhxor-lower-stoch}
 	For sufficiently large $\ww \geq 1$ and constant $c = O(1)$ with $t \geq  8 \cdot e^{9c} \cdot \ln \ww$, the distributional communication complexity of $\DHXORst_{\ww, t}$ for probability of success at least $\frac12 + \frac1{6\ww} $ is $\Omega(\ww \cdot t)$.
 \end{lemma}
The proof of \Cref{lem:dhxor-lower-stoch} is identical to the proof of \Cref{lem:dhx-lower} for $\clw = \ww/2e^{9c}$ and $\actt \geq \ln \ww$.

Putting things together, we prove \Cref{lem:ngcstoch} to finish this section. 

\begin{proof}[Proof of \Cref{lem:ngcstoch}]
	We can extend our hard distribution $\distNGC$ for any $n,k$ by a padding argument described in the proof of \Cref{thm:main-ngc}. 
	We know by \Cref{prop:minmax} that the randomized communication complexity of $\NGCst_{n,k}$ is lower bounded by the distributional communication complexity of $\NGCst_{n,k}$ over distribution $\distNGC$. 

	We also know that, for $\ww  = (n/4k)+1$ and $t \geq (k-3)/3 \geq  8 \cdot e^{9c}  \cdot \ln w$, the deterministic communication complexity of $\NGCst_{n,k}$ over distribution $\distNGC$ for probability of success at least $2/3$ is lower bounded by the distributional communication complexity of $\DHXORst_{\ww, t}$ for probability of success at least $\frac12 + \frac1{6\ww}$ by \Cref{lem:ngctodhxorstoch}.
	
	We also know that $\DHXORst_{\ww, t}$ requires $\Omega(\ww \cdot t) = \Omega(n/k \cdot k) = \Omega(n)$ space by \Cref{lem:dhxor-lower-stoch}, which completes the proof.  
\end{proof} 
\clearpage

\end{document}